\documentclass[10pt]{article} 
\usepackage{graphicx}
\usepackage{amsmath,amssymb,amsfonts}%
\usepackage{amsthm}%
\usepackage{mathrsfs}%
\usepackage[title]{appendix}%
\usepackage{xcolor}%
\usepackage{textcomp}%
\usepackage{manyfoot}%
\usepackage{booktabs}%
\usepackage{algorithm}%
\usepackage{algorithmicx}%
\usepackage{algpseudocode}%
\usepackage{listings}%
\usepackage{makecell}
\usepackage{multirow}%
\usepackage{subfig}
\usepackage[
  letterpaper,
  margin=1.5in,
  top=1.25in,
  bottom=1.5in,
  nohead
]{geometry}

\newtheorem{proposition}{Proposition}%
\newtheorem*{mainproblem}{Main Problem}{\bf}{}
\newtheorem{remark}{Remark}%
\newtheorem{definition}{Definition}%

\newtheorem{lemma}{Lemma}

\newcommand{\jj}[0]{\mathrm{j}}

\newcommand{\bernbas}{\mathscr{B}}
\newcommand{\re}[1]{\operatorname{Re}#1}
\newcommand{\rebrack}[1]{\operatorname{Re}\left\{#1\right\}}
\newcommand{\Hmat}{\mathbb{H}}
\newcommand{\vet}[1]{\boldsymbol{#1}}
\newcommand{\mat}[1]{\mathbf{#1}}
\newcommand{\mvartheta}{\boldsymbol{\vartheta}}
\newcommand{\eye}{\mathbb{I}}
\newcommand{\dotprod}[2]{\langle #1, #2 \rangle}

\newcommand{\norm}[1]{\left\vert\left\vert#1\right\vert\right\vert}\newcommand{\snorm}[1]{\vert\vert#1\vert\vert}
\newcommand{\junk}[1]{}
\newcommand{\data}[1]{\breve{#1}}

\begin{document}

\title{Stability Enforcement in Multivariate Rational Approximation of Parametric Transfer Functions
}

\author{Antonio Carlucci
\thanks{Department of Electronics and Telecommunications, Politecnico di Torino, Italy. Email: antonio.carlucci@polito.it
\\
{\it Preprint. This manuscript has not undergone peer review.}}
}
\date{}

\maketitle

\begin{abstract}
Preserving stability is a central problem in data-driven model order reduction of dynamical systems. For linear systems whose dynamics depend on geometric or physical parameters, multivariate rational approximation algorithms such as the Parameterized Sanathanan-Koerner iteration and the pAAA algorithm construct parameterized reduced models from sampled transfer function data. 
In this setting, stability must be enforced robustly across the parameter domain.
This paper introduces a necessary and sufficient criterion for characterizing the stability of parameterized models. Within a unified framework, the results apply to models with general rational as well as polynomial dependence on the parameters. Building on this criterion, we develop and demonstrate a rational approximation algorithm that includes robust stability constraints through convex optimization. Relative to the state of the art, the approach enforces stability without conservatism while allowing increased flexibility in the choice of model structure.
\end{abstract}


\section{Introduction}
Rational approximation is a key tool for reduced-order modeling of Linear Time-Invariant (LTI) systems. LTI models arise as mathematical representations of physical processes governed by linear differential equations. Deriving system equations directly from basic physical laws, for example by semi-discretization of a PDE, often leads to representations with many variables. Consequently, solving for the system response to a given input excitation implies a high computational cost, typically unacceptable in many engineering fields that heavily rely on simulations for design and verification. Model Order Reduction (MOR) techniques \cite{benner-morbook,antoulas-book} address this problem by constructing a simpler yet accurate representation of a given LTI system to be used as a surrogate in numerical simulations.

The input-output behavior of an LTI system can be represented in terms of a Transfer Function (TF), that is a function $\data{H}(s)$ of a complex variable $s$. 
The relevance of rational approximation algorithms for reduced-order modeling lies in their ability to approximate LTI systems by approximating their TFs. For this purpose, a low-order approximant $H(s)\approx \data{H}(s)$ is derived starting from sampled evaluations of the original TF $\data{H}(s)$. The resulting rational model $H(s)$ itself corresponds to a reduced LTI system with significantly fewer variables than the initial first-principles model. This is a purely data-driven approach that only requires availability of TF evaluations, which is appealing in practice.
Consequently, several MOR algorithms have rational approximation as an essential component \cite{morbook-rational}, including Vector Fitting (VF) \cite{gustavsen-vf}, the interpolatory Loewner framework \cite{mayo2007loewner}, the Adaptive Antoulas-Anderson (AAA) algorithm \cite{nakatsukasa2018aaa,nakatsukasa5years} and RKFIT \cite{rkfit}.

In engineering applications, reduced models are used to support design processes based on numerical simulations with the goal of optimizing the system for a specific objective. In these scenarios, the behavior of the system is modeled as depending on physical and geometric parameters $\mvartheta$. Consequently, the TF is a multivariate function $\data{H}(s,\mvartheta)$ of $\mvartheta$ as well as the Laplace variable $s$.
In this parameterized setting, it is natural to turn to multivariate rational approximation to construct reduced models of parameter-dependent TFs. 
Established algorithmic solutions include the Parameterized Sanathanan-Koerner (PSK) iteration \cite{triverio2009psk,psk-deschriver} and the pAAA algorithm \cite{paaa,paaa-scattered}. In addition, purely interpolatory approaches, such as Loewner and projection frameworks, have also been extended to the parametric case in \cite{ionita2014parametric,benner-parametric-survey}.

Motivated by applications in engineering design where physical parameters are associated with real numbers ranging in given intervals, we focus on deriving rational approximants of TFs that depend on a complex variable $s\in\mathbb{C}$ and $k$ real variables $\mvartheta = (\vartheta_1,\dots,\vartheta_k)\in[0,1]^k$.  
The dataset $\mathcal{D}$ is a collection of data points,
\begin{equation}\label{eq:dataset-def}
    \mathcal{D}= \{(s^{(i)},\mvartheta^{(i)}, \data{H}^{(i)}): \;  s^{(i)}\in \mathbb{C},\; \mvartheta^{(i)}\in[0,1]^k, i = 1,\dots,I\},
\end{equation}
where $\data{H}^{(i)}=\data{H}(s^{(i)}, \mvartheta^{(i)})$ are samples of a scalar transfer function $\data{H}(s,\mvartheta)$ depending on $s$ and $\mvartheta$.
In both pAAA and PSK, the model structure for the rational approximant $H(s,\mvartheta)$ is the common doubly-barycentric form \cite{berrut2004bary,antoulas-anderson},
\begin{equation}
    H(s,\mvartheta) = \frac{\alpha_0(\mvartheta)+
    \sum_{n=1}^{\nu}\frac{\alpha_n(\mvartheta)}{s-\sigma_n}
    }{\beta_0(\mvartheta) + \sum_{n=1}^{\nu}\frac{\beta_n(\mvartheta)}{s-\sigma_n}},
    \label{eq:bary}
\end{equation}
where $\{\alpha_n(\mvartheta)\}_{n=0}^{\nu}$ and $\{\beta_n(\mvartheta)\}_{n=0}^{\nu}$ are multivariate functions of $\mvartheta$. The \emph{basis poles} $\sigma_n\in\mathbb{C}$ are predefined auxiliary constants used to define the model \eqref{eq:bary}. 

The model format \eqref{eq:bary} is sufficiently general to encompass both the PSK iteration and the pAAA algorithm depending on the choice of $\alpha_n(\mvartheta)$, $\beta_n(\mvartheta)$, which we refer to as the \emph{parameterization} of the model.
 In PSK, $\alpha_n(\mvartheta)$ and $\beta_n(\mvartheta)$ are usually chosen to be polynomials \cite{bradde-stability,triverio2009psk}, but other choices are possible \cite{zanco-hd}. On the other hand, \eqref{eq:bary} matches the structure of pAAA if $\alpha_n(\mvartheta)$ and $\beta_n(\mvartheta)$ are rational functions, specifically linear combinations of (multivariate) partial fractions.
 The original pAAA formulation requires $\alpha_0(\mvartheta)=\beta_0(\mvartheta)=0$.
Focusing on scalar $\data{H}(s,\mvartheta)$ is not a limitation because the case of transfer matrices (arising in multi-input multi-output systems) can be handled through a barycentric form like \eqref{eq:bary} with matrix-valued numerator and scalar denominator \cite{gosea-mimo}. The results of this paper apply without modification in that case, hence we focus on scalar TFs to avoid obscuring the main ideas with unnecessary details.

This paper addresses the problem of enforcing stability of the rational model $H(s,\mvartheta)$, which is an essential property to ensure that the resulting reduced model can be used reliably in simulations. 
In many applications it is known from physical considerations that the system under analysis has bounded-input bounded-output behavior. However, data-driven algorithms can still yield unstable models if not explicitly constrained, even when the data samples are generated by a stable system.

Stability enforcement has been studied for PSK models in \cite{zanco-uniformly} and \cite{bradde-stability}. Both works achieve stability conservatively by enforcing conditions that are sufficient but not necessary. In addition, the results of \cite{zanco-uniformly,bradde-stability} are essentially tied to particular choices of model parameterizations, in which the functions $\beta_n(\mvartheta)$ are required, rather restrictively, to be constructed as combinations of specific basis functions with certain positivity properties. 
With the novel approach of this paper, we show that stability can in fact be enforced without imposing restrictive requirements on the parameterization. In particular, we provide stability criteria and an algorithm to construct stable models where $\alpha_n(\mvartheta)$, $\beta_n(\mvartheta)$ are general rational functions of $\mvartheta$. This includes, as special cases, multivariate polynomials expressed in arbitrary bases, rather than being restricted to Bernstein polynomials as in \cite{bradde-stability};  combinations of partial fractions as in pAAA; and rational radial basis functions such as inverse quadratics.
The class of parameterizations considered herein is defined precisely later in Sec.~\ref{sec:model-param}. A central contribution beyond this increased flexibility is the formulation of an exact stability
 criterion for parameterized models that is both necessary and sufficient. In this way, the conservatism inherent in the approaches of \cite{zanco-uniformly,bradde-stability} is removed.

\subsection{Notation and Definitions}
The imaginary unit is denoted as $\jj$. The open Right-Half Plane (RHP) is $\mathbb{C}_+ = \{s\in\mathbb{C}:\,\re s >0\}$, while the closed RHP is $\bar{\mathbb{C}}_+ = \mathbb{C}_+\cup\jj\mathbb{R}$. 
The symbol $\Hmat^n$ denotes the set of $n\times n$ Hermitian matrices, and $\Hmat_+^n$ and $\Hmat_{++}^n$ indicate positive semidefinite and positive definite matrices, respectively, of size $n\times n$. The superscript $n$ is omitted when the size is clear from the context or irrelevant. For $\mat{A}$, $\mat{B}\in\Hmat^n$, a scalar product is defined as $\dotprod{\mat{A}}{\mat{B}}= \mathrm{tr}(\mat{A}\mat{B})$. The superscript $\mat{A}^*$ indicates the Hermitian conjugate of $\mat{A}$, or the complex conjugate $x^*$ of a complex number $x$. $\mat{A}^T$ the matrix transpose. We use $\star$ to complete Hermitian matrices, and write $(\vet{x};\vet{y})$ to indicate vertical concatenation of vectors or matrices $\vet{x}$, $\vet{y}$. The symbol $\eye_n$ is the identity matrix of size $n$. The Frobenius norm of a matrix $\mat{A}$ is $\snorm{\mat{A}}_F$.

For a multi-index $\boldsymbol{n}=(n_1,\dots,n_k)$, the inequality $\boldsymbol{n}\leq l$ is intended element-wise, that is $n_{k'}\leq l$ for all $k'=1,\dots,k$. When considering multivariate polynomials depending on variables $\boldsymbol{x}=(x_1,\dots,x_k)\in\mathbb{R}^k$, the monomial $x_1^{n_1}\cdots x_k^{n_k}$ is abbreviated as $\boldsymbol{x}^{\boldsymbol{n}}$ and the entries of the multi-index $\boldsymbol{n}$, also used as exponent, are intended to be natural numbers including zero. Following \cite{boyd-cvx}, if $\mathcal{K}$ is a closed convex cone, the notation $\mat{X}\succeq_{\mathcal{K}}0$ and $\mat{X}\succ_{\mathcal{K}}0$ means, respectively, $\mat{X}\in\mathcal{K}$ and $\mat{X}\in\mathrm{int}\,\mathcal{K}$. 
As usual, $\mat{X}\succeq 0$, $\mat{X}\succ 0$ indicate that $\mat{X}$ is positive semidefinite ($\mat{X}\in\Hmat_+$) or positive definite ($\mat{X}\in\Hmat_{++}$). 
For a set $\mathcal{Q}\subseteq \Hmat^n$, the notation $\mathcal{Q}^*$ indicates its dual cone defined as $\mathcal{Q}^* = \{\mat{Z}\in\Hmat^n:\; \dotprod{\mat{Q}}{\mat{Z}}\geq 0\;\forall\, \mat{Q}\in\mathcal{Q}\}$. 
The convex hull of a set $\mathcal{S}$ is $\mathrm{conv}\,\mathcal{S}$. 
Following \cite{scherer2005}, the term \emph{robust Linear Matrix Inequality (LMI)} is used to indicate LMI conditions of the form $\mat{F}(\boldsymbol{x},\mvartheta)\succeq 0$ $\forall \mvartheta\in \mathcal{K}$ where $\mat{F}(\boldsymbol{x},\mvartheta)$ is affine in $\boldsymbol{x}$ for any fixed $\mvartheta$, and the condition is required to hold for all $\mvartheta$ in a given domain $\mathcal{K}$. In this paper, all robust LMIs have coefficients that depend polynomially on $\mvartheta$.
The Bernstein polynomial basis of degree $\delta$ in a single variable $x\in\mathbb{R}$ is
\begin{equation}
    \mathscr{B}_{n}^{({\delta})}(x) = \binom{\delta}{n}x^{n}(1-x)^{\delta-n} ,\qquad n = 0,\dots,\delta.
\end{equation}
In the multivariate case with $\boldsymbol{x}=(x_1,\dots,x_k)\in\mathbb{R}^{k}$, multivariate Bernstein polynomials are indexed by a multi-index $\boldsymbol{n} = (n_1,\dots,n_k)$. The Bernstein basis of degree $\delta$ in each variable is
\begin{equation}
    \mathscr{B}_{\boldsymbol{n}}^{({\delta})}(\boldsymbol{x}) = \mathscr{B}_{n_1}^{({\delta})}(x_1)\mathscr{B}_{n_2}^{({\delta})}(x_2)\cdots\mathscr{B}_{n_k}^{({\delta})}(x_k),\qquad 0\leq \boldsymbol{n}\leq\delta.
\end{equation}

\subsection{Model Parameterization and Optimization}\label{sec:model-param}
A general and unified approach is adopted here by assuming that $\alpha_n(\mvartheta)$ and $\beta_n(\mvartheta)$, for $n=0,\dots,\nu$, in \eqref{eq:bary} are linear combinations of a predefined set of $\rho$ basis functions $\varphi_1(\mvartheta),\dots,\varphi_{\rho}(\mvartheta)$. The
row vector $\boldsymbol{\varphi}(\mvartheta)$ collecting these functions is assumed to admit the rational representation
\begin{equation}\label{eq:phi-def}
     \vet{\varphi}(\mvartheta) = \begin{pmatrix}\varphi_1(\mvartheta) & \cdots & \varphi_\rho(\mvartheta)
\end{pmatrix} = \vet{\eta}[\mat{\Theta}(\mvartheta)-\mat{\Pi}]^{-1}\mat{B}_{\theta},
\end{equation}
with matrix coefficients $\mat{\Pi}\in\mathbb{R}^{K\times K}$, $\mat{B}_{\theta}\in\mathbb{R}^{K\times \rho}$, and a row vector $\vet{\eta}\in\mathbb{R}^{1\times K}$. The matrix $\mat{\Theta}(\mvartheta)\in\mathbb{R}^{K\times K}$ is only required to be an affine function of $\mvartheta$.
The functions $\alpha_n(\mvartheta)$ and $\beta_n(\mvartheta)$ appearing in the model structure \eqref{eq:bary} are expressed as linear combinations of the entries of $\vet{\varphi}(\mvartheta)$ according to 
\begin{equation}
    \alpha_n(\mvartheta) = \boldsymbol{\varphi}(\mvartheta)\boldsymbol{\alpha}_n,\quad 
    \beta_n(\mvartheta) = \boldsymbol{\varphi}(\mvartheta)\boldsymbol{\beta}_n,
\end{equation}
where coefficients are collected in the vectors $\boldsymbol{\alpha}_n$, $\boldsymbol{\beta}_n \in\mathbb{C}^{\rho}$.  

This format is sufficiently general because the entries of $\boldsymbol{\varphi}(\mvartheta)$ can be multivariate polynomials or general rational functions depending on the choice of $\mat{\Theta}(\mvartheta)$ and $\mat{\Pi}$. The model formats used in PSK or pAAA are particular cases of \eqref{eq:bary}, as shown constructively in Appendix \ref{app:parameterization}.

In the following, the term \emph{polynomial parameterization} is used to indicate the case in which $\boldsymbol{\varphi}(\mvartheta)$ contains only multivariate polynomials.
 The term \emph{rational parameterization} refers to the more general case in which $\vet{\varphi}(\mvartheta)$ contains rational functions. Polynomially parameterized models are regarded here as a special case of rationally parameterized models. 
The numerator and denominator in \eqref{eq:bary} are denoted by
\begin{equation}
\alpha(s,\mvartheta)
=\alpha_0(\mvartheta) + \sum_{n=1}^{\nu}\frac{\alpha_n(\mvartheta)}{s-\sigma_n}, \qquad
    \beta(s,\mvartheta)=\beta_0(\mvartheta) + \sum_{n=1}^{\nu}\frac{\beta_n(\mvartheta)}{s-\sigma_n}.
\end{equation}

Determining the coefficients $\{\boldsymbol{\alpha}_n\}$, $\{\boldsymbol{\beta}_n\}$ of the best rational approximant to a given dataset is not straightforward. In fact, a direct attempt to minimize the least-squares (LS) error between model and data, i.e. $\sum_i |\data{H}^{(i)}-H(s^{(i)}, \mvartheta^{(i)})|^2$, leads to a challenging non-convex optimization problem owing to the nonlinear dependence of $H(s,\mvartheta)$ on $\{\boldsymbol{\beta}_n\}$. 
Consequently, both PSK and pAAA use Levy's linearization \cite{levy} to introduce a simplified cost function $J_0$ that is convex in the model coefficients. In both cases, $J_0$ has the form 
\begin{equation}
    J_0(\{\vet{\alpha}_n\},\{\vet{\beta}_n\}) = \sum_{i\in\mathcal{I}}w_i\left\vert
    \data{H}^{(i)}\beta(s^{(i)}, \mvartheta^{(i)}) - \alpha(s^{(i)}, \mvartheta^{(i)}) 
    \right\vert^2\label{eq:j-def}
\end{equation}
where $w_i$ are positive weights and $\mathcal{I}$ is a subset of the indices $\{1,\dots,I\}$ corresponding to data points included in the cost function. 
In PSK, the set of weights is updated in each iteration to compensate for the fact that $J_0$ is an inexact surrogate of the model-data error. In this case $\mathcal{I} = \{1,\dots,I\}$, that is, all data points are used in $J_0$.
In pAAA, no weighting is used ($w_i = 1$) and $\mathcal{I}$ selects only a strict subset of the available data points to define $J_0$, while exact interpolation is enforced for the remaining points.

When no constraints are imposed on $\boldsymbol{\alpha}_n$, convexity of $J_0$ allows to eliminate the numerator coefficients to work with the cost function $J(\{\boldsymbol{\beta}_n\}) = \min_{\{\boldsymbol{\alpha}_n\}}J_0(\{\boldsymbol{\alpha}_n\},\{\boldsymbol{\beta}_n\})$.
Minimizing the cost $J(\{\boldsymbol{\beta}_n\})$ without additional constraints on $\{\boldsymbol{\beta}_n\}$ may yield a solution for which $H(s,\mvartheta)$ has unstable poles for some fixed values of $\mvartheta$, corresponding to zeros of $\beta(s,\mvartheta)$ with $s$ in the RHP. 
This paper focuses on the PSK iteration and formulates a constrained version of the optimization problem solved in each iteration in order to guarantee model stability. 

\section{Problem Statement}\label{sec:problem-statement}
The main objective of this paper is to construct multivariate rational approximants $H(s,\mvartheta)$ of parametric TFs that correspond to stable systems. 
Among the various notions of stability, \emph{Bounded-Input Bounded-Output (BIBO) stability} of a dynamical system is particularly relevant, as it is formulated in terms of input-output behavior. A system mapping $u(t)$ to $y(t)$ is BIBO stable \cite{bibo-note} if the response $y(t)$ to any bounded input $u(t)$ is bounded.

BIBO stability requires $H(s,\mvartheta)$ to have no poles for any $s$ in the closed RHP and any $\mvartheta\in[0,1]^k$, i.e. $H(s,\mvartheta)$ must have no poles in the region
\begin{equation}
    \Lambda = \{(s,\mvartheta):\;s\in\bar{\mathbb{C}}_+,\; \mvartheta \in[0,1]^k\}.
\end{equation} 
In this paper, it is assumed that the basis poles $\{\sigma_n\}$ defining \eqref{eq:bary} have negative real part. Hence, poles of $H(s,\mvartheta)$ in $\Lambda$ correspond to zeros of the denominator function $\beta(s,\mvartheta)$, and stability requires that $\beta(s,\mvartheta)\neq 0$ for all $(s,\mvartheta)\in\Lambda$. 
In addition, $H(s,\mvartheta)$ cannot have poles at infinity, that is, $H(\infty,\mvartheta)=\lim_{s\to\infty} H(s,\mvartheta)$ is finite for all $\mvartheta\in[0,1]^k$. In fact, if this condition does not hold, $H(s,\mvartheta)$ is improper and corresponds to a system whose response depends on derivatives of the input. Such a model may produce an unbounded response even when the input is bounded. The barycentric form \eqref{eq:bary} has no poles at infinity provided that $\beta(\infty,\mvartheta)=\lim_{s\to \infty}\beta(s,\mvartheta)\neq 0$ for all $\mvartheta\in[0,1]^k$.

Consistently with the notion of BIBO stability, a model of the form \eqref{eq:bary} is termed \emph{stable} if its denominator has no zeros in $\Lambda$ or at infinity. Equivalently, $\beta(s,\mvartheta)\neq 0$ in the extended stability region $\bar{\Lambda} = \Lambda \cup (\infty, [0,1]^k)$. 
\begin{mainproblem} Starting from a given dataset $\mathcal{D}$ as in \eqref{eq:dataset-def}, the problem considered herein is to find coefficients $\{\vet{\alpha}_n\}$, $\{\vet{\beta}_n\}$ of a multivariate rational model $H(s,\mvartheta)$ in the form \eqref{eq:bary} minimizing the cost function $J(\{\boldsymbol{\beta}_n\})$ with the constraint that $\beta(s,\mvartheta) \neq 0$ for all $(s,\mvartheta)\in\bar{\Lambda}$, i.e. the model is stable.
\end{mainproblem}

Previous work in \cite{bradde-stability} followed an approach that guarantees a weaker notion of stability by constraining the denominator to be a Positive Real (PR) function. A rational function $f(s):\mathbb{C}\to\mathbb{C}$ is called PR \cite{anderson1967pr} if it is analytic in the open RHP $\mathbb{C}_+$, and $\re{f(s)}\geq 0$ $\forall\,s\in\mathbb{C}_+$. Note that, as in \cite{Guiver2017}, the usual additional conjugate symmetry assumption $f(s^*) = f(s)^*$ is here omitted because it is not needed. The relevance of PR-ness in this context derives from the fact that a PR function cannot have zeros in $\mathbb{C}_+$ unless it is identically zero. Therefore, constraining $\beta(s,\mvartheta)$ to be PR for all fixed $\mvartheta\in[0,1]^k$ gives, as a consequence, $\beta(s,\mvartheta)\neq 0$ over $\mathbb{C}_+\times[0,1]^k$. It is crucial to observe that PR-ness of the denominator is sufficient but not necessary to ensure that the model has no poles in the open RHP. Hence, constraining $\beta(s,\mvartheta)$ to be PR may yield a suboptimal model. Nevertheless, this approach is appealing because the Positive Real Lemma (PRL) \cite{anderson1967pr} gives algebraic conditions on the coefficients $\{\boldsymbol{\beta}_n\}$ that are equivalent to PR-ness of $\beta(s,\mvartheta)$ and are compatible with computationally tractable semidefinite optimization. 

Enforcement of PR-ness does not prevent the denominator from vanishing for values of $s$ on the imaginary axis or at infinity, implying that a PR denominator could still have unwanted zeros within the stability region $\bar{\Lambda}$ considered herein. We shall thus adopt the following slightly stronger notion of positivity that guarantees the absence of zeros in $\bar{\Lambda}$, 
\begin{definition}[Extended Strictly Positive (ESP) function, \cite{weiqian1994}]
    A proper rational transfer function $f(s): \mathbb{C}\to\mathbb{C}$ is Extended Strictly Positive if it is analytic in an open right half-plane containing $\bar{\mathbb{C}}_+$ and $\re{f(\jj\omega)}\geq \epsilon$ $\forall\, \omega \in\mathbb{R}$ for some real $\epsilon >0$.
\end{definition}
Similarly to PR functions, ESP functions are nonzero on the extended closed RHP $\bar{\mathbb{C}}_+\cup\infty$. Therefore, a conservative way to obtain a stable model is to constrain $\beta(s,\mvartheta)$ to be ESP for all fixed $\mvartheta\in[0,1]^k$. Requiring the ESP property to obtain model stability as a consequence is attractive because, as with PR-ness, the Kalman-Yakubovich-Popov (KYP) \cite{kyp-kalman} lemma allows one to express the ESP conditions equivalently as an LMI in the coefficients $\{\boldsymbol{\beta}_n\}$. The ESP property is again sufficient, but not necessary, to ensure that $\beta(s,\mvartheta)\neq 0$ in $\bar{\Lambda}$, thus representing a stronger requirement whose enforcement may lead to suboptimal model accuracy.  

\subsection{Outline and Contributions}
A solution to the main problem is developed in several steps:
\begin{itemize}
\item[a)] Section \ref{sec:preliminary} discusses preliminary aspects regarding a state-space realization of $\beta(s,\mvartheta)$. Section \ref{sec:background} reviews an existing approach for enforcing model stability in the context of multivariate rational fitting. 
    \item[b)] Section  \ref{sec:stability-condition} presents a matrix inequality condition that is necessary and sufficient for model stability. The result is an exact characterization of stable models that makes it possible to enforce stability in a non-conservative way. The results of this section are given as robust matrix inequalities.
    \item[c)] Section \ref{sec:polya} discusses the Bernstein relaxation approach for converting robust matrix inequalities into finitely many conditions so as to formulate computationally tractable optimization problems. The convergence of the relaxation is established by adapting the results of \cite{scherer2005}.
    \item[d)] Section \ref{sec:relax} uses the theoretical results to formulate semidefinite programming (SDP) problems that can be solved numerically \cite{boyd-cvx,lmibook} via convex optimization. Since the main stability condition leads to a nonconvex constraint, a solution is devised to transform the problem into a convex one.

    \item[e)] Section \ref{sec:results} presents numerical results on four examples.
\end{itemize}

This paper contributes to the state of the art in several aspects. First, it provides theoretical stability criteria and algorithmic solutions that apply to general rationally parameterized models. Stability is not tied to a particular choice of model parameterization, unlike in \cite{zanco-uniformly,bradde-stability}. This increased flexibility in the proposed formulation enables enforcing stability for polynomially parameterized models using any polynomial basis (e.g. Legendre, Chebyshev) or rationally parameterized models (e.g. pAAA-type). In addition, stability is enforced exactly rather than conservatively through more restrictive properties such as PR-ness, offering a way to substantially reduce the model-data error for a given model complexity. 

\section{Preliminaries and Background}
\subsection{Realization of the denominator function}\label{sec:preliminary}
This section introduces a state-space realization of the denominator function $\beta(s,\mvartheta)$.
Let $\mat{\Sigma} = \textrm{diag}\{\sigma_i\}_{i=1}^{\nu}\in\mathbb{C}^{\nu\times \nu}$ and $\vet{b}_{\sigma}=\mathbf{1}$, where $\mathbf{1}$ is a vector of $\nu$ ones. The denominator of \eqref{eq:bary} can be rewritten as
\begin{equation}
    \beta(s,\mvartheta) =  \beta_0(\mvartheta) + \begin{pmatrix}
        \beta_1(\mvartheta)&\cdots & \beta_\nu(\mvartheta)
    \end{pmatrix}[s\eye-\mat{\Sigma}]^{-1}\vet{b}_{\sigma}.
\end{equation}
The model coefficients are grouped in $
    \boldsymbol{\beta} = \begin{pmatrix}
        \boldsymbol{\beta}_0 & \cdots & \boldsymbol{\beta}_\nu 
    \end{pmatrix}\in\mathbb{C}^{\rho\times (\nu+1)}$.
We also introduce the matrix $\mat{A}_{21}$ and the vector $\vet{b}_{\theta}$ that depend linearly on the denominator coefficients $\boldsymbol{\beta}$,
\begin{equation}
\mat{A}_{21}=
    \mat{B}_{\theta}\begin{pmatrix}
     \boldsymbol{\beta}_1
     &
     \cdots
     &
     \boldsymbol{\beta}_{\nu}
    \end{pmatrix},\qquad \vet{b}_{\theta} = \mat{B}_{\theta}\boldsymbol{\beta}_0.
\end{equation}
With these definitions, the denominator is written as $
    \beta(s,\mvartheta) =
\vet{\eta}[\mat{\Theta}(\mvartheta) - \mat{\Pi}]^{-1}\left[\mat{A}_{21}(s\eye-\mat{\Sigma})^{-1}\vet{b}_{\sigma} + \vet{b}_{\theta}\right]
$ and admits the realization
\begin{equation}\label{eq:beta-def-1}
    \beta(s,\mvartheta) = \begin{pmatrix}
      \vet{0} &\vet{\eta}
    \end{pmatrix}\left[
    \begin{pmatrix}
        s\eye_{\nu} & \mat{0}\\
        \mat{0}& \mat{\Theta}(\mvartheta)
    \end{pmatrix}
    -
    \begin{pmatrix}
        \mat{\Sigma} & \mat{0}
        \\
        \mat{A}_{21} & \mat{\Pi}
    \end{pmatrix}
    \right]^{-1}\begin{pmatrix}
         \vet{b}_{\sigma}\\\vet{b}_{\theta}
    \end{pmatrix}.
\end{equation}
Hence, $\beta(s,\mvartheta)$ can be regarded as a multivariate function of $(s, \mvartheta)$ with the explicit expression $\beta(s,\mvartheta) = \vet{c}[\mat{\Delta}(s,\mvartheta) - \mat{A}]^{-1}\vet{b}$ in terms of the realization matrices
\begin{equation}\label{eq:beta-realz-matrices}
    \vet{c} = \begin{pmatrix}
       \vet{0} & \vet{\eta}
    \end{pmatrix},
    \quad \vet{b} = \begin{pmatrix}
         \vet{b}_{\sigma}\\\vet{b}_{\theta}
    \end{pmatrix},
    \quad \mat{\Delta}(s,\mvartheta) = \begin{pmatrix}
        s\eye_{\nu} & \mat{0}\\\mat{0}
        & \mat{\Theta}(\mvartheta)
    \end{pmatrix},
    \quad 
    \mat{A} =\begin{pmatrix}
       \mat{\Sigma} & \mat{0}
        \\
        \mat{A}_{21} & \mat{\Pi}
    \end{pmatrix}.
\end{equation}

Equivalently, the same function can be regarded as a rational transfer function in the variable $s$ with parameter-dependent coefficients by introducing
\begin{equation}
    \vet{c}(\mvartheta) =\vet{\eta}[\mat{\Theta}(\mvartheta) - \mat{\Pi}]^{-1}\mat{A}_{21},\qquad d(\mvartheta) = \vet{\eta}[\mat{\Theta}(\mvartheta) - \mat{\Pi}]^{-1}\vet{b}_{\theta},\label{eq:c-d-def}
\end{equation}
to rewrite \eqref{eq:beta-def-1} equivalently as $
 \beta(s,\mvartheta) = \vet{c}(\mvartheta)[s\eye-\mat{\Sigma}]^{-1}\vet{b}_{\sigma} + d(\mvartheta)
$. This equivalent format emphasizes that the denominator function may be regarded as a standard transfer function in $s$ for fixed $\mvartheta$.

\begin{remark}
    If the transfer function to be approximated is known to have conjugate symmetry, the relation $H(s^*,\mvartheta) = H(s,\mvartheta)^*$ can be structurally embedded in the model. In this case, $\{\sigma_n\}_{n=1}^{\nu}$ is chosen to consist of $n_r$ real poles $\sigma_1,\dots,\sigma_{n_r}$ followed by $n_c$ complex conjugate pairs (i.e. $\nu = n_r+2n_c$). Assume for simplicity that conjugate pairs are indexed consecutively. Then $\beta_n(\mvartheta)$ is constrained to be real-valued if $\sigma_n$ is real, and conjugate symmetry $\beta_{n}(\mvartheta) = \beta_{n+1}(\mvartheta)^*$ is imposed for conjugate pairs $\sigma_{n}$, $\sigma_{n+1} = \sigma_{n}^*$. In the above notation, enforcement of conjugate symmetry amounts to choosing $\mat{\Sigma} = \mat{T}^{-1}\mathrm{diag}\{\sigma_1,\dots,\sigma_\nu\}\mat{T}$, $\vet{b}_{\sigma} = \mat{T}^{-1}\mathbf{1}$ with $\mat{T} = \mathrm{diag}\{\eye_{n_r}, \mat{T}_c\}$,
    \begin{equation}
        \mat{T}_c=\eye_{n_c}\otimes \begin{pmatrix}
            1 & \jj
            \\
            1 & -\jj
        \end{pmatrix}.
    \end{equation}
    In this case $\mat{\Sigma}$, $\vet{b}_{\sigma}$ are real and the model coefficients $\boldsymbol{\beta}$ to be estimated are also constrained to be real. This procedure is well-known and, as explained in \cite[Sec. 8.3.8]{triverio-vf}, it is simply a change of variables whereby complex model coefficients are expressed in terms of their real and imaginary parts.
\end{remark}
\subsection{Review of Stability Enforcement}
\label{sec:background}
An early approach to enforcing stability of PSK models was described in \cite{zanco-uniformly}. The main idea is to choose basis functions that are nonnegative over the parametric domain and impose constraints so as to obtain a PR denominator. The resulting method imposes even more restrictive constraints than PR-ness, which is itself not necessary. Another strategy proposed in \cite{austin-multivariate} is limited to real-valued rational functions.

The methodology of \cite{zanco-uniformly} was improved in \cite{bradde-stability}. The resulting framework is restricted to polynomially parameterized models, where $\vet{c}(\mvartheta)$ and $d(\mvartheta)$ are polynomials in $\mvartheta$. Moreover, it is inherently conservative, as it constrains $\boldsymbol{\beta}$ to obtain a PR denominator, which is a sufficient, but not necessary, condition to rule out denominator zeros in the open RHP.
The approach hinges on applying the PRL to $\beta(s,\mvartheta)$ at every point $\mvartheta$ in its domain. The PRL can be stated in this parameterized setting as follows.
    \begin{proposition}[Positive-Real Lemma, \cite{anderson1967pr,weiqian1994,dickinson-pseudopositive,gusev-kyp}]\label{prop:prl}
    Assume that $\mat{\Sigma}$ is Hurwitz, $(\mat{\Sigma},\vet{b}_{\sigma})$ is controllable and $\vet{c}(\mvartheta)$, $d(\mvartheta)$ are as in \eqref{eq:c-d-def}.
    Let $\mat{X}(\mvartheta)\in \Hmat$ denote a Hermitian-valued function and
     \begin{equation}
     \mat{W}(\mvartheta) =
    \begin{pmatrix}
        \mat{X}({\mvartheta})\mat{\Sigma} + \mat{\Sigma}^* \mat{X}({\mvartheta}) & \mat{X}({\mvartheta})\vet{b}_{\sigma}- \vet{c}(\mvartheta)^*
        \\
        \star & -d(\mvartheta)-d(\mvartheta)^*
    \end{pmatrix}.\label{eq:pr-theta-lmi}
    \end{equation}
    Then the following statements hold for $\beta(s,\mvartheta)  = \vet{c}(\mvartheta)[s\eye-\mat{\Sigma}]^{-1}\vet{b}_{\sigma} + d(\mvartheta)$:
    \begin{itemize}
        \item [a)] $\beta(\cdot,\mvartheta)$ is PR for all $\mvartheta\in[0,1]^k$ if and only if there exists $\mat{X}({\mvartheta})\succeq 0$ satisfying $\mat{W}(\mvartheta)\preceq 0$  $\forall\,\mvartheta\in[0,1]^k$;
        \item[b)] $\beta(\cdot,\mvartheta)$ is ESP for all $\mvartheta\in[0,1]^k$ if and only if there exists $\mat{X}({\mvartheta})\succ 0$ satisfying $\mat{W}(\mvartheta)\prec 0$  $\forall\,\mvartheta\in[0,1]^k$.
    \end{itemize}
\end{proposition}
In principle, conditions \emph{a)} or \emph{b)} ensure that $\beta(s,\mvartheta)$ is respectively PR or ESP throughout the parametric domain if a \emph{certificate} matrix $\mat{X}(\mvartheta)$ can be found independently for each $\mvartheta \in [0,1]^k$. In turn, the PR or ESP properties imply $\beta(s,\mvartheta)\neq 0$ in $\mathbb{C}_+$ or $\bar{\mathbb{C}}_+\cup\infty$, respectively.
It is important to remark that no requirement is explicitly imposed on the dependence of $\mat{X}(\mvartheta)$ on $\mvartheta$. However, in practice, the search for a certificate $\mat{X}(\mvartheta)$ is made tractable only by assuming a certain functional dependence on $\mvartheta$. Considering polynomially parameterized models where $\vet{c}(\mvartheta)= \sum_{\boldsymbol{n}\leq {\delta}} \vet{c}_{\boldsymbol{n}}\mvartheta^{\boldsymbol{n}}$, $d(\mvartheta)= \sum_{\boldsymbol{n}\leq {\delta}}d_{\boldsymbol{n}}\mvartheta^{\boldsymbol{n}}$ are multivariate polynomials with degree at most $\delta$ in each variable, the method of \cite{bradde-stability} makes condition \eqref{eq:pr-theta-lmi} computationally tractable by restricting $\mat{X}(\mvartheta)$ to be polynomial in $\mvartheta$ with the same maximum degree $\delta$, that is $
    \mat{X}(\mvartheta) = \sum_{\boldsymbol{n}\leq \delta}\mat{X}_{\boldsymbol{n}}\mvartheta^{\boldsymbol{n}}
$.  
Restricting the search to polynomial $\mat{X}(\mvartheta)$ is backed theoretically by the results in \cite{bliman,gusev-parameterdependent}.
Substituting this ansatz into the ESP condition $\mat{W}(\mvartheta)\prec0$ of Prop. \ref{prop:prl}(b) yields a robust LMI in the auxiliary unknowns $\mat{X}_{\boldsymbol{n}}$ along with $\vet{c}_{\boldsymbol{n}}$, $d_{\boldsymbol{n}}$ which, in turn, depend linearly on the model coefficients $\boldsymbol{\beta}$ to be optimized.
The resulting condition reads
\begin{equation}
   \mat{W}(\mvartheta) =\sum_{\boldsymbol{n}\leq \delta} 
 \begin{pmatrix}
        \mat{X}_{\boldsymbol{n}}\mat{\Sigma} + \mat{\Sigma}^* \mat{X}_{\boldsymbol{n}} & \mat{X}_{\boldsymbol{n}}\vet{b}_{\sigma}- \vet{c}_{\boldsymbol{n}}^*
        \\
        \star & -d_{\boldsymbol{n}}-d_{\boldsymbol{n}}^*
    \end{pmatrix} \mvartheta^{\boldsymbol{n}}
    =\sum_{\boldsymbol{n}\leq \delta} \mat{W}_{\boldsymbol{n}}\mvartheta^{\boldsymbol{n}}\prec 0,
    \label{eq:w-poly-exp}
\end{equation}
where $\mat{W}(\mvartheta)$ is expressed in terms of Hermitian coefficients $\mat{W}_{\boldsymbol{n}}$ that depend linearly on the auxiliary variables $\mat{X}_{\boldsymbol{n}}$ and the model coefficients $\boldsymbol{\beta}$ through $ \vet{c}_{\boldsymbol{n}}, d_{\boldsymbol{n}}$. Considering the PR condition instead would lead to the non-strict version of \eqref{eq:w-poly-exp}, which is the alternative adopted in \cite{bradde-stability}. We henceforth focus on the ESP criterion as it ensures the absence of model poles in the closed RHP and at infinity, consistently with the notion of BIBO stability adopted herein. 
Condition \eqref{eq:w-poly-exp} is used to formulate a constrained optimization problem
\begin{subequations}
\begin{align}
     \operatorname*{minimize}_{\mat{X}_{\boldsymbol{n}}, \boldsymbol{\beta}}\;&J(\boldsymbol{\beta})\;\mathrm{s.t.}
     \\
     &\mat{W}(\mvartheta)\prec 0, \; \forall\,\mvartheta\in [0,1]^k
     \label{eq:robust-lmi-w}
     \end{align}
     \label{prob:background}
\end{subequations}
whose solution yields coefficients $\boldsymbol{\beta}$ of a model with ESP denominator, hence stable. In \eqref{prob:background}, $J(\vet{\beta})$ is a shorthand for $J(\{\vet{\beta}_n\})$.

In this way, the problem is reduced to enforcing the robust LMI $\mat{W}(\mvartheta)\prec 0 \; \forall \,\mvartheta\in[0,1]^k$. As discussed in \cite{scherer2005}, the problem of establishing conditions under which a polynomial is positive on a given set has been studied extensively in the mathematical literature with classical results such as P\'olya's and Handelman's theorems in the case of simplicial and polytopic domains, respectively. These provide the basis for relaxing robust LMIs such as \eqref{prob:background} into finitely many positivity conditions involving the coefficients $\mat{W}_{\boldsymbol{n}}$ \cite{scherer2005,peet-polyopt}. For parameterized macromodeling, the work in \cite{bradde-stability} represents $\mat{W}(\mvartheta)$ in the Bernstein basis with an \emph{elevated} degree $\bar{{\delta}}\geq {\delta}$,
\begin{equation}\label{eq:w-bern-basis}
    \mat{W}(\mvartheta) = \sum_{\boldsymbol{n}\leq \bar{\delta}} \mat{W}^{(\bar{\delta})}_{{\rm ber},\boldsymbol{n}}
    \mathscr{B}^{(\bar{{\delta}})}_{\boldsymbol{n}}(\mvartheta),
\end{equation}
where the Bernstein coefficients $\mat{W}^{(\bar{{\delta}})}_{{\rm ber},\boldsymbol{n}}$ depend implicitly on $\{\mat{X}_{\boldsymbol{n}}\}, \{\vet{c}_{\boldsymbol{n}}\}, \{d_{\boldsymbol{n}}\}$. They are obtained as linear combinations of the coefficients $\mat{W}_{\vet{n}}$ in \eqref{eq:w-poly-exp} by applying a change of basis from the monomial to the degree-$\bar{{\delta}}$ Bernstein basis. Then the robust LMI \eqref{eq:robust-lmi-w} is enforced by requiring the stronger conditions
\begin{equation}
    \mat{W}^{(\bar{{\delta}})}_{{\rm ber},\boldsymbol{n}}\prec 0,\qquad  \boldsymbol{n}\leq \bar{\delta}.
\label{eq:bernstein-positivity}
\end{equation}

The process of replacing the robust LMI $\mat{W}(\mvartheta)\prec 0\;\forall\,\mvartheta\in[0,1]^k$ with the finitely many standard LMIs \eqref{eq:bernstein-positivity} is referred to as \emph{Bernstein relaxation}. 
Because Bernstein polynomials $\mathscr{B}_{\boldsymbol{n}}(\mvartheta)$ are nonnegative for $\mvartheta\in[0,1]^k$ and form a partition of unity, conditions \eqref{eq:bernstein-positivity} imply the inequality \eqref{eq:w-poly-exp}. The main idea of \cite{bradde-stability} is to consider representations with increasingly large $\bar{{\delta}}$. In fact, even if replacing \eqref{eq:w-poly-exp} with the stronger conditions \eqref{eq:bernstein-positivity} may introduce additional conservatism, especially for small $\bar{{\delta}}$, this conservatism is gradually reduced as $\bar{{\delta}}$ is increased. 
Relaxations based on the Bernstein basis are also discussed in \cite{leroy-bernstein-simplicial,boudaoud2008bernstein,kojima2019bernstein,hilhorst-bernstein} and related methods are leveraged in robust control \cite{peet-polyopt,chen-handelman-lyapunov,hamadneh-bernstein-lyapunov,chesi-lmi}.
We refer the reader to \cite{bradde-stability} for a complete description of this approach, and do not dwell on the Bernstein relaxation of robust LMIs as this topic will be revisited in Sec. \ref{sec:polya}, where we provide a generalization for the purposes of this paper.

The above approach is successful in yielding accurate and stable models when combined with the PSK iteration, as documented in \cite{bradde-stability}. However, it has two main limitations that the developments presented herein aim to overcome. First, it is limited to polynomially parameterized models, and it is not applicable in the case of rational parameterization. Second, it enforces $\rebrack{\beta(s,\mvartheta)}>0$ on $\bar{\Lambda}$, which is more restrictive than the condition of stability, i.e. $\beta(s,\mvartheta)\neq 0$. Therefore, even when the conservatism arising from the Bernstein relaxation is removed by degree elevation, the approach remains inherently conservative. 
 The following developments present solutions to both issues. In addition, an alternative approach is introduced that avoids explicitly parameterizing $\mat{X}(\mvartheta)$ by using a different stability criterion that guarantees stability through a parameter-independent certificate matrix.

\section{Stability Conditions}
\label{sec:stability-condition}
\label{sec:global-cert}
This section is concerned with deriving criteria to establish whether a given denominator function $\beta(s,\mvartheta)$ corresponds to a stable model. As in Sec. \ref{sec:background}, these criteria are given in terms of matrix inequality conditions involving the model coefficients $\boldsymbol{\beta}$ to be estimated. 
The results of this section are stated as robust matrix inequalities, as the conversion to a finite set of standard LMIs will be discussed in Sec. \ref{sec:polya}.

An exact non-conservative stability condition is derived first, as it is the cornerstone of the algorithm presented herein, followed by a simpler but conservative condition that is used only instrumentally in the final proposed algorithm. 

\subsection{Exact Stability Condition}
We view the denominator $\beta(s,\mvartheta) = \vet{c}[\mat{\Delta}(s,\mvartheta)-\mat{A}]^{-1}\vet{b}$ as a rational function of the combined variables $(s,\mvartheta)$ as defined in \eqref{eq:beta-def-1} in terms of the realization matrices $\mat{A}$, $\vet{b}$, $\vet{c}$. The objective is to ensure that $\beta(s,\mvartheta)\neq 0$ over the domain $\bar{\Lambda}$.
The derivation of the stability condition follows from a variant of the S-procedure \cite{terlaky-survey-s} known as the \emph{full-block S-procedure}, introduced in \cite{scherer2001} and also used in \cite{scherer2005}. The first step is to equivalently rephrase the model stability condition $\beta(s,\mvartheta)\neq 0$ as the positivity condition $|\beta(s,\mvartheta)|^2>0\;\forall\, (s,\mvartheta)\in\bar{\Lambda}$.
As a consequence of the full-block S-procedure \cite{scherer2005}, the inequality
\begin{multline}
    |\beta(s,\mvartheta)|^2=\\
    \begin{pmatrix}
        [\mat{\Delta}(s,\mvartheta)-\mat{A}]^{-1}\vet{b}
        \\
        1
    \end{pmatrix}^*
    \begin{pmatrix}
        \vet{c}^*\vet{c} & \vet{0}\\\star & 0
    \end{pmatrix}
    \begin{pmatrix}
        [\mat{\Delta}(s,\mvartheta)-\mat{A}]^{-1}\vet{b}
        \\
        1
    \end{pmatrix}>0\quad \forall (s,\mvartheta)\in\bar{\Lambda}
\label{eq:positivity-condition}
\end{multline}
is satisfied if there is a Hermitian matrix $\mat{P}$ that satisfies the LMI
\begin{equation}\label{eq:stab-cond}
    \begin{pmatrix}
        \vet{c}^*\vet{c} & \vet{0}\\\star & 0
    \end{pmatrix}
     - 
     \begin{pmatrix}
         \mat{A} & \vet{b} 
         \\
         \eye &0
     \end{pmatrix}^*
     \mat{P}
     \begin{pmatrix}
         \mat{A} & \vet{b} 
         \\
         \eye &0
     \end{pmatrix}\succ 0
\end{equation}
and
\begin{equation}
    \begin{pmatrix}
        \mat{\Delta}(s,\mvartheta) \\ \eye
    \end{pmatrix}^*\mat{P}
    \begin{pmatrix}
        \mat{\Delta}(s,\mvartheta) \\ \eye
    \end{pmatrix} \succeq 0\qquad \forall \,(s,\mvartheta)\in{\Lambda}
    .
    \label{eq:p-def}
\end{equation}
In other words, the matrix $\mat{P}$ satisfying \eqref{eq:stab-cond} is to be sought in a set of admissible matrices $\mat{P}$ that satisfy the robust LMI \eqref{eq:p-def}.
As remarked in \cite{SCHERER2006} and originally proven in \cite{scherer2001}, the conditions of the full-block S-procedure also become necessary when the domain is compact, which is not the case for $\Lambda$. In addition, the robust inequality \eqref{eq:p-def} is still too abstract and unsuitable for numerical implementation because it is only an implicit definition of the admissible matrices $\mat{P}$, as observed in \cite{scherer2005}. This issue is solved by restricting the search to a subset of admissible matrices $\mat{P}$ for which an explicit characterization can be provided, as shown next.

Because the parameters $\mvartheta$ vary in the hypercube, while $s$ ranges over $\bar{\mathbb{C}}_+$, it is convenient to express \eqref{eq:p-def} as a quadratic form in $s$ with coefficients depending on $\mvartheta$,
by means of the identity
\begin{equation}
\begin{pmatrix}
    \mat{\Delta}(s,\mvartheta) \\ \eye
\end{pmatrix}=
    \begin{pmatrix}
        s\eye_\nu & \mat{0} \\
        \mat{0} & \mat{\Theta}(\mvartheta) \\ 
        \eye_\nu & \mat{0} 
        \\
        \mat{0} & \eye_K
    \end{pmatrix}=
    \underbrace{
     \begin{pmatrix}
        \eye_\nu &\mat{0} &\mat{0}
        \\
        \mat{0}
        &
        \mat{0}
        &
        \mat{\Theta}(\mvartheta)
        \\
        \mat{0}  
        &
        \eye_\nu
        & \mat{0}
        \\
        \mat{0}  
        &
        \mat{0}
        & \eye_{K}
    \end{pmatrix}
    }_{\triangleq \mat{M}(\mvartheta)}
     \begin{pmatrix}
        s\eye_{\nu} & \mat{0}
        \\
        \eye_{\nu} & \mat{0}
        \\
       \mat{0} &\eye_K
    \end{pmatrix}\label{eq:def-of-M}
\end{equation}
so that \eqref{eq:p-def} is
\begin{equation}\label{eq:s-highlight}
     \begin{pmatrix}
    \mat{\Delta}(s,\mvartheta) \\ \eye
\end{pmatrix}^*\mat{P}\begin{pmatrix}
    \mat{\Delta}(s,\mvartheta) \\ \eye
\end{pmatrix}
    =
    \begin{pmatrix}
        s\eye_{\nu} & \mat{0}
        \\
        \eye_{\nu} & \mat{0}
        \\
        \mat{0}&\eye_K
    \end{pmatrix}^*\mat{M}(\mvartheta)^*\mat{P}\mat{M}(\mvartheta)
    \begin{pmatrix}
        s\eye_{\nu} & \mat{0}
        \\
        \eye_{\nu} & \mat{0}
        \\
       \mat{0} &\eye_K
    \end{pmatrix}\succeq 0.
\end{equation}
Having isolated the $s$ variable, let us observe that for any $\mat{X}\succeq 0$
\begin{equation}\label{eq:skew-x}
\begin{pmatrix}
    s\eye & \mat{0}
    \\
    \eye & \mat{0}
    \\
    \mat{0} & \eye
\end{pmatrix}^*
    \begin{pmatrix}
        \mat{0} & \mat{X}&\mat{0} \\ \mat{X} & \mat{0} & \mat{0}\\
        \mat{0} & \mat{0} & \mat{0}
    \end{pmatrix}
    \begin{pmatrix}
    s\eye & \mat{0}
    \\
    \eye & \mat{0}
    \\
    \mat{0} & \eye
\end{pmatrix}\succeq 0\qquad \forall \;s\in\bar{\mathbb{C}}_+.
\end{equation}
 A similar observation appears in \cite{SCHERER2006} as a way to relax semi-infinite constraints like \eqref{eq:s-highlight}. Moreover, the skew-block-diagonal structure in \eqref{eq:skew-x} is also related to certain matrix sets considered in \cite{generalized-kyp,generalized-s-procedure} to describe curves or regions in the complex plane in a non-parameterized setting. In our case, where parameters $\mvartheta$ are present in \eqref{eq:s-highlight} alongside the complex variable $s$, we can expand on this insight to introduce the cone $\mathcal{Q}$ of Hermitian matrices defined as
\begin{equation}\label{eq:q-def}
    \mathcal{Q} = \left\{\mat{Q}\in\Hmat^{2\nu+K}: \exists \, \mat{X}\in\Hmat^{\nu},\;\mat{X}\succeq 0,\; 
    \mat{Q}\succeq \begin{pmatrix}
        \mat{0} & \mat{X} &\mat{0} \\ \mat{X} & \mat{0} & \mat{0}
        \\ \mat{0} & \mat{0} & \mat{0}
    \end{pmatrix}\right\}.
\end{equation}
By comparing \eqref{eq:s-highlight} and \eqref{eq:skew-x}, it can be directly seen that if, for a fixed $\mvartheta$, the matrix $\mat{M}(\mvartheta)^*\mat{P}\mat{M}(\mvartheta)$ belongs to $\mathcal{Q}$, i.e. $\mat{M}(\mvartheta)^*\mat{P}\mat{M}(\mvartheta)\succeq_{\mathcal{Q}}0$, then inequality \eqref{eq:s-highlight} is verified for all $s\in\bar{\mathbb{C}}_+$. Therefore, if $\mat{P}$ satisfies the robust LMI $\mat{M}(\mvartheta)^*\mat{P}\mat{M}(\mvartheta)\succeq_{\mathcal{Q}} 0\;\;\forall\, \mvartheta \in [0,1]^k$, then it also satisfies \eqref{eq:p-def}. Hence, the set 
\begin{equation}
    \mathcal{P} = \{\mat{P}\in\Hmat: \mat{M}(\mvartheta)^*\mat{P}\mat{M}(\mvartheta)\succeq_{\mathcal{Q}} 0\;\;\forall\, \mvartheta \in [0,1]^k\}\label{eq:p-set-def}
\end{equation} is a subset of admissible matrices $\mat{P}$ for \eqref{eq:p-def}.

In Appendix \ref{sec:appendix2}, we use duality arguments as in \cite{SCHERER2006,generalized-s-procedure,generalized-kyp} to show that considering this subset $\mathcal{P}$ is enough to obtain a necessary and sufficient stability criterion based on a parameter-independent certificate matrix as summarized in the following proposition.
\begin{proposition} \label{prop:stab-criterion}
Let $\beta(s,\mvartheta)=\vet{c}[\mat{\Delta}(s,\mvartheta)-\mat{A}]^{-1}\vet{b}$ be as defined in \eqref{eq:beta-def-1}-\eqref{eq:beta-realz-matrices} and assume that $\mat{\Delta}(s,\mvartheta)-\mat{A}$ is nonsingular for all $(s,\mvartheta)\in\Lambda$. Then $\beta(s,\mvartheta)\neq 0$ for all $(s,\mvartheta) \in \bar{\Lambda}$ if and only if there exists a matrix $\mat{P}\in\Hmat$ such that
\begin{align}
&\mat{M}(\mvartheta)^*\mat{P}\mat{M}(\mvartheta)\succeq_{\mathcal{Q}} 0 \qquad \forall \,\mvartheta\in [0,1]^k,\label{eq:q-cond}
     \\
     &\begin{pmatrix}
        \vet{c}^*\vet{c} & \vet{0} \\ \star & 0
    \end{pmatrix}
     - 
     \begin{pmatrix}
         \mat{A} & \vet{b} 
         \\
         \eye &\vet{0}
     \end{pmatrix}^*
     \mat{P}
     \begin{pmatrix}
         \mat{A} & \vet{b} 
         \\
         \eye &\vet{0}
     \end{pmatrix}\succ 0.
     \label{eq:lmi-stab-cond}
\end{align}
\end{proposition}
\begin{proof}
    See Appendix \ref{sec:appendix2}.
\end{proof}
A convenient way to write \eqref{eq:lmi-stab-cond}, which also makes it explicit that $\mat{A}$ and $\vet{b}$ depend on the model coefficients $\boldsymbol{\beta}$, is $\mat{\Gamma}_{\rm st}-\mat{K}_{\rm st}(\boldsymbol{\beta})^*\mat{P}\mat{K}_{\rm st}(\boldsymbol{\beta})\succ 0$ using the abbreviations
\begin{equation}
    \mat{\Gamma}_{\rm st} = \begin{pmatrix}
        \vet{c}^*\vet{c} & \vet{0} \\ \star & 0
    \end{pmatrix},
    \qquad 
    \mat{K}_{\rm st}(\boldsymbol{\beta})=\begin{pmatrix}
         \mat{A} & \vet{b} 
         \\
         \eye &\vet{0}
     \end{pmatrix}.
\end{equation}
Proposition \ref{prop:stab-criterion} states that the condition $\beta(s,\mvartheta)\neq 0$ on $\bar{\Lambda}$ is equivalent to the matrix inequality \eqref{eq:lmi-stab-cond} and the robust condition \eqref{eq:q-cond} that the polynomial matrix $\mat{M}(\mvartheta)^*\mat{P}\mat{M}(\mvartheta)$, with quadratic dependence on $\mvartheta$, is contained in $\mathcal{Q}$.
The latter constraint is still too abstract to be practically applicable because, in addition to being a robust inequality in $\mvartheta$, it is stated in terms of generalized positivity with respect to $\mathcal{Q}$. Nevertheless, the Bernstein relaxation remains useful because it can be adapted to convert \eqref{eq:q-cond} into a finite set of standard LMIs. Hence, the issue of making condition \eqref{eq:q-cond} computationally practical is addressed in the next section through the Bernstein relaxation.

It is important to observe that the matrix inequality \eqref{eq:lmi-stab-cond} is not linear in the model coefficients $\boldsymbol{\beta}$ to be optimized, since they appear in $\mat{A}$ and $\vet{b}$. Although a way to handle the resulting non-convexity will be described later, we first complete the picture by introducing a natively convex but conservative criterion obtained by considering the ESP property again. 

\subsection{ESP Condition}

The ideas described in the previous section can also be used to derive a conservative stability criterion based on ESP functions and a parameter-independent certificate. In fact, for the ESP property, we need to verify positivity of the quadratic form
\begin{equation}
    2\re \beta(s,\mvartheta) =\begin{pmatrix}
        [s\eye-\mat{\Sigma}]^{-1}\vet{b}_{\sigma}
        \\
        [\mat{\Theta}(\mvartheta)^*-\mat{\Pi}^*]^{-1}\vet{\eta}^*
        \\
        1
    \end{pmatrix}^* \begin{pmatrix}
        0 & \mat{A}_{21}^* & \vet{0} \\ \star & \mat{0} &\vet{b}_{\theta}
        \\
       \star&\star&0
    \end{pmatrix}\begin{pmatrix}
        [s\eye-\mat{\Sigma}]^{-1}\vet{b}_{\sigma}
        \\
        [\mat{\Theta}(\mvartheta)^*-\mat{\Pi}^*]^{-1}\vet{\eta}^*
        \\
        1
    \end{pmatrix}>0
    \label{eq:pr-positivity-condition}
\end{equation}
for all $s\in\jj\mathbb{R}\cup\infty$ and $\mvartheta\in[0,1]^k$. The full-block S-procedure applies as above to formulate LMI conditions equivalent to the quadratic constraint \eqref{eq:pr-positivity-condition}. Retracing the steps to make the dependence on $s$ explicit leads to the definition of the matrix
\begin{equation}
    \mat{M}_0(\mvartheta) = \begin{pmatrix}
        \eye_\nu &\mat{0} &\mat{0}
        \\
        \mat{0}
        &
        \mat{0}
        &
        \mat{\Theta}^*(\mvartheta)
        \\
        \mat{0}  
        &
        \eye_\nu
        & \mat{0}
        \\
        \mat{0}  
        &
        \mat{0}
        & \eye_{K}
    \end{pmatrix}
\end{equation}
that is used analogously to $\mat{M}(\mvartheta)$ to write a semi-infinite constraint on $s$ similar to \eqref{eq:s-highlight}. Since, in this case, the resulting semi-infinite constraint must be checked only on the (extended) imaginary axis rather than on the closed RHP, a larger set of matrices may be considered according to the following definition: 
\begin{equation}
    \mathcal{Q}_0 = \left\{\mat{Q}\in\Hmat^{2\nu+K}: \exists \, \mat{X}\in\Hmat^{\nu},\;\mat{Q}\succeq \begin{pmatrix}
        \mat{0} & \mat{X} &\mat{0} \\ \mat{X} & \mat{0} & \mat{0}
        \\ \mat{0} & \mat{0} & \mat{0}
    \end{pmatrix}\right\}.
\end{equation}
The following statement represents the counterpart of the ESP criterion in Prop. \ref{prop:prl}(b) based on a parameter-independent certificate.
\begin{proposition}\label{prop:pr-criterion}
    Let $\beta(s,\mvartheta)=\vet{c}[\mat{\Delta}(s,\mvartheta)-\mat{A}]^{-1}\vet{b}$ with  $\mat{A}$, $\vet{b}$, $\vet{c}$, $\mat{\Delta}(s,\mvartheta)$ given in \eqref{eq:beta-realz-matrices}. Assume that $\mat{\Delta}(s,\mvartheta)-\mat{A}$ is nonsingular for $(s,\mvartheta)\in\bar{\mathbb{C}}_+\times [0,1]^k$. Then $\beta(\cdot,\mvartheta)$ is ESP for all $\mvartheta\in[0,1]^k$ if and only if there exists a Hermitian matrix $\mat{P}$ such that
    \begin{align}
&\mat{M}_0(\mvartheta)^*\mat{P}\mat{M}_0(\mvartheta)\succeq_{\mathcal{Q}_{0}} 0 \qquad \forall \,\mvartheta\in [0,1]^k,\label{eq:q-cond-pr}
     \\
     &\underbrace{\begin{pmatrix}
        0 & \mat{A}_{21}^* & \vet{0} \\ \star & \mat{0} &\vet{b}_{\theta}
        \\
       \star&\star&0
    \end{pmatrix}}_{\triangleq \mat{\Gamma}_{\rm esp}(\boldsymbol{\beta})}
     - 
     \begin{pmatrix}
             \mat{\Sigma} & \mat{0} & \vet{b}_{\sigma}
             \\
             \mat{0}&\mat{\Pi}^* & \vet{\eta}^*
          \\
             \eye & \mat{0}& \mat{0}
             \\
             \mat{0} & \eye & \mat{0}
     \end{pmatrix}^*
     \mat{P}
     \underbrace{
    \begin{pmatrix}
             \mat{\Sigma} & \mat{0} & \vet{b}_{\sigma}
             \\
             \mat{0}&\mat{\Pi}^* & \vet{\eta}^*
          \\
             \eye & \mat{0}& \mat{0}
             \\
             \mat{0} & \eye & \mat{0}
     \end{pmatrix}}_{\triangleq \mat{K}_{\rm esp}}\succ 0.
     \label{eq:lmi-pr-cond}
\end{align}
\end{proposition}
\begin{proof}
    See Appendix \ref{sec:pr-proof}.
\end{proof}
Inequality \eqref{eq:lmi-pr-cond} is an LMI in the model coefficients, as they appear only linearly in the matrices $\mat{A}_{21}$ and $\vet{b}_{\theta}$. It can also be written compactly as $\mat{\Gamma}_{\rm esp}(\boldsymbol{\beta}) - \mat{K}_{\rm esp}^*\mat{P}\mat{K}_{\rm esp}\succ 0$.

\section{Stability Enforcement }\label{sec:polya}
This section is concerned with using the criteria in Sec. \ref{sec:global-cert} to explicitly formulate optimization problems to find stable models minimizing the cost $J(\boldsymbol{\beta})$. The exact stability condition given in Prop. \ref{prop:stab-criterion} leads to the optimization problem
\begin{subequations}
\begin{align}
    \operatorname*{minimize}_{\boldsymbol{\beta}, \mat{P},\mu>0}\, J(\boldsymbol{\beta})\quad \textrm{s.t.}\; &
    \mu\mat{\Gamma}_{\rm st} - \mat{K}_{\rm st}(\boldsymbol{\beta})^*\mat{P}\mat{K}_{\rm st}(\boldsymbol{\beta})\succeq \eye,
   \label{eq:lmi-stab-cond-prob} \\&
\mat{M}(\mvartheta)^*\mat{P}\mat{M}(\mvartheta)\succeq_{\mathcal{Q}} 0 \quad\forall\,\mvartheta\in [0,1]^k.\label{eq:q-cond-prob}
\end{align}\label{prob:stable-model-abstract}
\end{subequations}
where the positive rescaling variable $\mu$ has been introduced to rewrite \eqref{eq:lmi-stab-cond} equivalently as a nonstrict inequality. Because of the product $\mat{K}_{\rm st}(\boldsymbol{\beta})^*\mat{P}\mat{K}_{\rm st}(\boldsymbol{\beta})$, \eqref{prob:stable-model-abstract} is not a convex SDP problem. On the other hand, the conservative ESP criterion of Prop. \ref{prop:pr-criterion} leads to a convex SDP problem
\begin{subequations}
\begin{align}
    \operatorname*{minimize}_{\boldsymbol{\beta}, \mat{P}}\, J(\boldsymbol{\beta})\quad \textrm{s.t.}\; &\mat{\Gamma}_{\rm esp}(\boldsymbol{\beta}) - \mat{K}_{\rm esp}^*\mat{P}\mat{K}_{\rm esp}\succeq \eye, 
    \label{eq:spr-model-1}
    \\
    &
\mat{M}_0(\mvartheta)^*\mat{P}\mat{M}_0(\mvartheta)\succeq_{\mathcal{Q}_0} 0 \quad\forall\,\mvartheta\in [0,1]^k.\label{eq:q-cond-pr-prob}
\end{align}
\label{prob:spr-model}
\end{subequations} 
It is possible to write \eqref{eq:spr-model-1} as a nonstrict inequality because the original \eqref{eq:lmi-pr-cond} is homogeneous in $\boldsymbol{\beta}$ and $\mat{P}$, and $J(\vet{\beta})$ is also homogeneous.
Problem \eqref{prob:spr-model} is convex and thus solvable numerically to obtain a model with an ESP denominator, hence stable.

Both problems \eqref{prob:stable-model-abstract}-\eqref{prob:spr-model} still contain the robust conditions \eqref{eq:q-cond-prob} or \eqref{eq:q-cond-pr-prob}. Focusing on the former, we now show how it is converted into a finite set of standard LMIs through a generalization of the Bernstein relaxation. Note that the use of the Bernstein basis in the relaxation is completely independent from the model parameterization. In fact, the Bernstein relaxation as formulated next is generally applicable to replace a generalized positivity condition on a polynomial matrix $\mat{M}(\mvartheta)^*\mat{P}\mat{M}(\mvartheta)\succeq_{\mathcal{Q}}0$ with a finite set of standard LMIs.
The Bernstein relaxation applied to the robust LMI \eqref{eq:q-cond-prob} requires to express the quadratic polynomial $\mat{W}(\mvartheta) = \mat{M}(\mvartheta)^*\mat{P}\mat{M}(\mvartheta)$ in the Bernstein basis with a given elevated degree $\bar{\delta}\geq 2$ and with Hermitian coefficients $\mat{W}_{{\rm ber},\boldsymbol{n}}^{(\bar{\delta})}(\mat{P})$ depending on $\mat{P}$, i.e.
\begin{equation}
    \mat{W}(\mvartheta) = \sum_{\boldsymbol{n}\leq \bar{\delta}} \mat{W}_{{\rm ber}, \boldsymbol{n}}^{(\bar{\delta})}(\mat{P})\bernbas^{(\bar{\delta})}_{\boldsymbol{n}}(\mvartheta).
\end{equation}
For any $\bar{\delta}\geq 2$, the finite set of conditions 
\begin{equation}
    \mat{W}_{{\rm ber},\boldsymbol{n}}^{(\bar{\delta})}(\mat{P})\succeq_{\mathcal{Q}} 0,\qquad  \boldsymbol{n}\leq\bar{\delta}\label{eq:bern-q}
\end{equation}
implies $\mat{W}(\mvartheta)\succeq_{\mathcal{Q}} 0$ for all $\mvartheta \in [0,1]^k$. This is because the Bernstein basis polynomials are nonnegative in the unit hypercube, making $\mat{W}(\mvartheta)$ a conic combination of its Bernstein coefficients $\mat{W}_{{\rm ber},\boldsymbol{n}}^{(\bar{\delta})}(\mat{P})$. Replacing \eqref{eq:q-cond-prob} with \eqref{eq:bern-q} yields the relaxation of Problem \eqref{prob:stable-model-abstract} with finitely many LMI constraints,
\begin{subequations}
\begin{align}
    \operatorname*{minimize}_{\boldsymbol{\beta}, \mat{P}, \{\mat{X}_{\boldsymbol{n}}\},\mu>0}\, J(\boldsymbol{\beta})\quad \textrm{s.t.}\; &\mu\mat{\Gamma}_{\rm st} - \mat{K}_{\rm st}(\boldsymbol{\beta})^*\mat{P}\mat{K}_{\rm st}(\boldsymbol{\beta})\succeq \eye,\label{eq:lmi-stab-cond-prob-2}
    \\
    &\mat{W}_{{\rm ber},\boldsymbol{n}}^{(\bar{\delta})}(\mat{P})\succeq \begin{pmatrix}
        0 & \mat{X}_{\boldsymbol{n}} & 0
        \\
        \star & 0 & 0
        \\
        \star&\star&0
    \end{pmatrix}, \; \mat{X}_{\boldsymbol{n}}\succeq 0
\label{eq:stable-model-b}
\end{align}\label{prob:stable-model}
\end{subequations}
where \eqref{eq:stable-model-b} is the explicit LMI corresponding to $\mat{W}_{{\rm ber},\boldsymbol{n}}^{(\bar{\delta})}(\mat{P})\succeq_{\mathcal{Q}} 0$, by definition of $\mathcal{Q}$. The relaxation of \eqref{prob:spr-model} follows by applying the same reasoning to the polynomial $\mat{W}(\mvartheta) = \mat{M}_0(\mvartheta)^*\mat{P}\mat{M}_0(\mvartheta)$.

The relaxed constraints \eqref{eq:stable-model-b} are more restrictive than \eqref{eq:q-cond-prob}. Therefore, the optimal cost of \eqref{prob:stable-model} is larger than that of \eqref{prob:stable-model-abstract}. 
However, it can be shown that the suboptimality introduced by replacing Problem \eqref{prob:stable-model-abstract} with \eqref{prob:stable-model} can be made negligible by increasing the elevation degree $\bar{\delta}$. The key point to consider is whether the relaxed conditions \eqref{eq:bern-q} are also necessary, rather than merely sufficient, for the original robust LMI.
To this end, we adapt the technique of \cite{scherer2005} to derive the following analogue of \cite[Theorem 7.1]{scherer2005}.
\begin{proposition}\label{prop:polya}
    Let $\mathcal{Q}\subset\Hmat$ be a closed convex cone and $\mat{W}(\mvartheta)$ be an Hermitian polynomial matrix such that $
    \mat{W}(\mvartheta)\succ_{\mathcal{Q}}0\; \forall \,\mvartheta\in[0,1]^k$.
    Then there is an elevation degree $\bar{\delta}$ for which the representation of $\mat{W}(\mvartheta)$ in the degree-$\bar{\delta}$ Bernstein basis
    \begin{equation}
        \mat{W}(\mvartheta) = \sum_{\boldsymbol{n}\leq \bar{\delta}}\mat{W}_{{\rm ber},\boldsymbol{n}}^{(\bar{\delta})}\mathscr{B}^{(\bar{\delta})}_{\boldsymbol{n}}(\mvartheta)
    \end{equation}
    satisfies $\mat{W}_{{\rm ber},\boldsymbol{n}}^{(\bar{\delta})}\succeq_{\mathcal{Q}}0$. Moreover, Bernstein coefficients $\mat{W}_{{\rm ber},\boldsymbol{n}}^{(\delta)}$ of a representation with higher degree $\delta\geq \bar{\delta}$ also satisfy $\mat{W}_{{\rm ber},\boldsymbol{n}}^{({\delta})}\succeq_{\mathcal{Q}}0$. 
\end{proposition}
\begin{proof}
    See Appendix \ref{sec:appendix3}.
\end{proof}
In other words, the strict version of \eqref{eq:q-cond-prob} is equivalent to its Bernstein relaxation \eqref{eq:stable-model-b} for a sufficiently large $\bar{\delta}$. 

If \eqref{prob:stable-model-abstract} was formulated only in terms of strict inequalities, the reasoning in \cite[Theorem 7.2]{scherer2005} would be directly applicable to deduce from Proposition \ref{prop:polya} that the optimal cost of the Bernstein relaxation \eqref{prob:stable-model} converges to the optimal cost of \eqref{prob:stable-model-abstract} as $\bar{\delta}\to \infty$. Since \eqref{eq:q-cond-prob} is nonstrict, we make the following observation to show that the same conclusion still holds.

\begin{remark}\label{rem:convergence}
Let $\mathcal{F}$ be the feasible set of \eqref{prob:stable-model-abstract}, and $\mathcal{F}_{\bar{\delta}}$ the feasible set of its degree-$\bar{\delta}$ relaxation \eqref{prob:stable-model}. Then for any $(\boldsymbol{\beta}, \mat{P},\mu)\in\mathcal{F}$, there is a point $(\boldsymbol{\beta}, \mat{P}',\mu')$ that is arbitrarily close to $(\boldsymbol{\beta}, \mat{P},\mu)$ and satisfies $(\boldsymbol{\beta}, \mat{P}',\mu')\in\mathcal{F}_{\bar{\delta}}$ for a sufficiently large $\bar{\delta}$. In fact, fix an arbitrarily small $\epsilon$ with $0<\varepsilon<1$. A sufficiently small $\mat{P}_{\rm sm}\succ 0$ always exists for which
$\mu\mat{\Gamma}_{\rm st} - \mat{K}_{\rm st}(\boldsymbol{\beta})^*(\mat{P}+\mat{P}_{\rm sm})\mat{K}_{\rm st}(\boldsymbol{\beta})\succeq (1-\varepsilon)\eye$. 
Define $\mat{P}' = (\mat{P}+\mat{P}_{\rm sm})/(1-\varepsilon)$ and $\mu' = \mu/(1-\varepsilon)$.
Since $\mat{M}(\mvartheta)$ has full column rank, $\mat{M}(\mvartheta)^*\mat{P}_{\rm sm}\mat{M}(\mvartheta)\succ  0$, hence $\mat{M}(\mvartheta)^*\mat{P}'\mat{M}(\mvartheta)\succ_{\mathcal{Q}}  0$ for all $\mvartheta \in [0,1]^k$. 
Therefore the point $(\boldsymbol{\beta}, \mat{P}', \mu')$ is still feasible for \eqref{prob:stable-model-abstract} with strict $\succ_{\mathcal{Q}}$. By Proposition \ref{prop:polya}, there is a $\bar{\delta}$ such that this triple is also feasible for the Bernstein relaxation \eqref{prob:stable-model} of degree $\bar{\delta}$, i.e. $(\boldsymbol{\beta}, \mat{P}',\mu')\in\mathcal{F}_{\bar{\delta}}$, or larger. 
A similar argument holds for problem \eqref{prob:spr-model}. In that case, with arbitrarily small $\varepsilon>0$ and $\mat{P}_{\rm sm}\succ 0$, any point $(\vet{\beta},\mat{P})$ in the feasible set $\mathcal{F}$ of \eqref{prob:spr-model} can be approximated by a point $(\vet{\beta}',\mat{P}')$ with $\vet{\beta}'=\vet{\beta}/(1-\varepsilon)$, $\mat{P}'=(\mat{P}+\mat{P}_{\rm sm})/(1-\varepsilon)$  that is in the feasible set $\mathcal{F}_{\bar{\delta}}$ of a Bernstein relaxation with sufficiently large $\bar{\delta}$. 
\end{remark}

 Remark \ref{rem:convergence} shows that in both problems \eqref{prob:stable-model-abstract}-\eqref{prob:spr-model} any feasible point can be approximated arbitrarily well by a point in the feasible set of a relaxation with sufficiently large $\bar{\delta}$. Consequently, the optimal values of the relaxed problems converge from above to those of \eqref{prob:stable-model-abstract}-\eqref{prob:spr-model} as $\bar{\delta}\to\infty$.
The next section finally addresses the non-convexity of Problem \eqref{prob:stable-model} to enable the construction of stable models with the exact criterion of Prop. \ref{prop:stab-criterion}.

\section{Local Solution Method}\label{sec:relax}
The purpose of this section is to formulate a computational approach to solve the optimization problem \eqref{prob:stable-model} that uses the exact criterion of Prop. \ref{prop:stab-criterion} instead of requiring ESP-ness. The only obstacle to a direct solution of Problem \eqref{prob:stable-model} as an SDP is the non-convexity of \eqref{eq:lmi-stab-cond-prob-2}.

Non-convexity can be addressed by resorting to a variation of the \emph{Iterative Convex Overbounding} technique (ICO) described in \cite{warner-ico,ico} or the similar technique of \cite{boydmimopid}. A crucial preliminary requirement is to have a feasible solution for \eqref{prob:stable-model}. This can be found through the following procedure: first, a conservatively stable model with an ESP denominator is obtained by solving the convex problem \eqref{prob:spr-model}. Denote the denominator coefficients of this model as $\bar{\boldsymbol{\beta}}$, so that the corresponding optimal cost is $\gamma_0 = J(\bar{\boldsymbol{\beta}})$. This model, although suboptimal, is stable and, as such, it must verify the conditions of Proposition \ref{prop:stab-criterion}. Therefore, a certificate $\bar{\mat{P}}$ can be found by solving the convex feasibility problem corresponding to \eqref{prob:stable-model} after fixing the variables $\boldsymbol{\beta}=\bar{\boldsymbol{\beta}}$. 
The solution $\bar{\mat{P}}$, $\bar{\boldsymbol{\beta}}$, $\bar{\mu}$ obtained in this way is feasible for \eqref{prob:stable-model}.

 At this point the strategy is to replace \eqref{eq:lmi-stab-cond-prob-2} with a stronger LMI condition. To separate convex and concave quadratic terms in \eqref{eq:lmi-stab-cond-prob-2}, we can express $\mat{P}=\mat{P}_+-\mat{P}_-$ as the difference of two positive definite matrices. It is not restrictive to require $\mat{P}_+\succeq \eye, \mat{P}_-\succeq \eye$. Also decompose the initially feasible $\bar{\mat{P}}$ as $\bar{\mat{P}}=\bar{\mat{P}}_+-\bar{\mat{P}}_-$ with $\bar{\mat{P}}_+\succeq \eye, \bar{\mat{P}}_-\succeq \eye$. We use the shorthand 
$\bar{\mat{K}}_{\rm st} = \mat{K}_{\rm st}(\bar{\boldsymbol{\beta}})$. The matrix
\begin{multline}
    \mathcal{R}(\mat{P}_+,\mat{P}_-,\boldsymbol{\beta},\mu; \bar{\mat{P}}_+,\bar{\mat{P}}_-, \bar{\boldsymbol{\beta}}) = \\\begin{pmatrix}
        2\bar{\mat{P}}^{-1}_{+}  - \bar{\mat{P}}_{+}^{-1}\mat{P}_+\bar{\mat{P}}_{+}^{-1}&\mat{0} &\mat{K}_{\rm st}(\boldsymbol{\beta}) 
        \\
        \star& \mat{P}_- & \bar{\mat{P}}_{-}\bar{\mat{K}}_{\rm st}
        \\
       \star &\star & \mu\mat{\Gamma}_{\rm st}-\eye+\bar{\mat{K}}_{\rm st}^*\bar{\mat{P}}_{-}\mat{K}_{\rm st}(\vet{\beta})+(\bar{\mat{K}}_{\rm st}^*\bar{\mat{P}}_{-}\mat{K}_{\rm st}(\vet{\beta}))^*
    \end{pmatrix}
\end{multline}
 is linear in the unknowns and it can be shown that the constraint $\mu\mat{\Gamma}_{\rm st} -\eye- \mat{K}_{\rm st}(\boldsymbol{\beta})^*\mat{P}\mat{K}_{\rm st}(\boldsymbol{\beta})\succeq 0$ that appears in \eqref{eq:lmi-stab-cond-prob-2} is implied by the LMI 
\begin{equation}
    \mathcal{R}(\mat{P}_+,\mat{P}_-,\boldsymbol{\beta},\mu; \bar{\mat{P}}_+,\bar{\mat{P}}_-, \bar{\boldsymbol{\beta}})\succeq 0.
\label{eq:relax}
\end{equation}
This is mainly a consequence of the relation $\mat{Y}^*\mat{X}\mat{Y}\succeq \bar{\mat{Y}}^*\bar{\mat{X}}\mat{Y} + \mat{Y}^*\bar{\mat{X}}\bar{\mat{Y}} - \bar{\mat{Y}}^*\bar{\mat{X}}\mat{X}^{-1}\bar{\mat{X}}\bar{\mat{Y}}$, known as Young's relation \cite{lmicookbook}, which is valid for generic $\mat{X},\bar{\mat{X}}\in\Hmat_{++}$, $\mat{Y},\bar{\mat{Y}}$ and it holds with equality when $\mat{X}=\bar{\mat{X}}$, $\mat{Y}=\bar{\mat{Y}}$. In fact, starting from $\mu\mat{\Gamma}_{\rm st} -\eye- \mat{K}_{\rm st}(\boldsymbol{\beta})^*(\mat{P}_+-\mat{P}_-)\mat{K}_{\rm st}(\boldsymbol{\beta})$, Young's inequality can be used to minorize terms like $\mat{K}_{\rm st}(\boldsymbol{\beta})^*\mat{P}_-\mat{K}_{\rm st}(\boldsymbol{\beta})$, and, combined with Schur complement, one can conclude that the LMI \eqref{eq:relax} implies the nonconvex constraint in \eqref{eq:lmi-stab-cond-prob-2}.

Upon replacing \eqref{eq:lmi-stab-cond-prob-2} with \eqref{eq:relax} in the nonconvex problem \eqref{prob:stable-model}, we obtain a convexification of \eqref{prob:stable-model} centered around $(\bar{\mat{P}}_+,\bar{\mat{P}}_- ,\bar{\boldsymbol{\beta}})$ defined as 
\begin{subequations}
\begin{align}
\operatorname*{minimize}_{\substack{\boldsymbol{\beta}, \mat{P}_+\succeq \eye,\mat{P}_-\succeq \eye,\\
\{\mat{X}_{\boldsymbol{n}}\},\mu>0  }}\, J(\boldsymbol{\beta})\quad 
    \textrm{s.t.}\; &\mathcal{R}(\mat{P}_+,\mat{P}_-,\boldsymbol{\beta},\mu; \bar{\mat{P}}_+,\bar{\mat{P}}_-, \bar{\boldsymbol{\beta}})\succeq 0,
    \\&
    \mat{W}_{{\rm ber},\boldsymbol{n}}^{(\bar{\delta})}(\mat{P}_+-\mat{P}_-)\succeq \begin{pmatrix}
        0 & \mat{X}_{\boldsymbol{n}} & 0
        \\
        \star & 0 & 0
        \\
        \star&\star&0
    \end{pmatrix}, \; \mat{X}_{\boldsymbol{n}}\succeq 0
    \end{align}
\label{prob:relaxed-stable-model}
\end{subequations}
The values $\mat{P}_+=\bar{\mat{P}}_+$, $\mat{P}_-=\bar{\mat{P}}_-$, $\boldsymbol{\beta} = \bar{\boldsymbol{\beta}}$, $\mu=\bar{\mu}$ are still feasible for \eqref{prob:relaxed-stable-model}. Therefore, the optimal cost $\gamma_1 = J(\boldsymbol{\beta})$ of \eqref{prob:relaxed-stable-model} is not greater than the initial $\gamma_0$, i.e. $\gamma_1\leq \gamma_0$. Because \eqref{eq:relax} implies the nonconvex constraint \eqref{eq:lmi-stab-cond-prob-2}, the optimal solution of  \eqref{prob:relaxed-stable-model} is feasible for the original \eqref{prob:stable-model}. Hence, the result of \eqref{prob:relaxed-stable-model} can be used again iteratively as a starting point to define a new convexification of \eqref{prob:stable-model} centered at the newly found solution (as in \cite{warner-ico}).

\begin{algorithm}
\caption{Multivariate Rational Fitting with Stability Constraints}
\label{alg:stable-identification}
\begin{algorithmic}[1]
\Require Dataset $\mathcal{D}$ as in \eqref{eq:dataset-def}, basis poles $\{\sigma_n\}_{n=1}^{\nu}$, parameterization $(\mat{\Theta}(\mvartheta),\mat{\Pi}, \mat{B}_{\theta}, \vet{\eta})$, elevation degree $\bar{\delta}$
\State Initialize weights $w_i \gets 1$, $i=1,\dots,I$

\For{$\ell=1,\dots,n_{\mathrm{psk}}$}
    \State Solve the Bernstein relaxation of \eqref{prob:spr-model} to obtain the optimal model coefficients $\boldsymbol{\beta}$
    \State With this $\boldsymbol{\beta}$, compute denominator values $\beta(s^{(i)},\mvartheta^{(i)})$ and update the weights
    \[
        w_i \gets |\beta(s^{(i)},\mvartheta^{(i)})|^{-1} \qquad i=1,\dots,I
    \]
    \State Redefine the cost function $J(\boldsymbol{\beta})$ using the updated weights $\{w_i\}_{i=1}^I$
\EndFor

\State Set $\bar{\boldsymbol{\beta}} \gets \boldsymbol{\beta}$ and compute cost $\gamma_0 \leftarrow J(\bar{\boldsymbol{\beta}})$.
\State Find feasible $\bar{\mat{P}}=\bar{\mat{P}}_+-\bar{\mat{P}}_-$ with $\bar{\mat{P}}_+,\bar{\mat{P}}_-\succeq \eye$ by solving the convex feasibility problem associated with \eqref{prob:stable-model} after fixing $\boldsymbol{\beta}=\bar{\boldsymbol{\beta}}$.
\For{$\ell=1,\dots,n_{\mathrm{ico}}$}
    \State Solve \eqref{prob:relaxed-stable-model} to obtain optimal $\boldsymbol{\beta}$, $\mat{P}_+$, $\mat{P}_-$, $\mu$ with optimal cost $\gamma_\ell = J(\boldsymbol{\beta})$.
    \State Update the convexification point $
        \bar{\boldsymbol{\beta}} \gets \boldsymbol{\beta}, \;
        \bar{\mat{P}}_+ \gets \mat{P}_+, \;
        \bar{\mat{P}}_- \gets \mat{P}_-
    $
    \If{$
        \frac{\gamma_{\ell-1}-\gamma_\ell}{\gamma_{\ell-1}} < \varepsilon_{\rm ico}
    $}
        \State \textbf{break}
    \EndIf
\EndFor

\State \Return $\boldsymbol{\beta}$
\end{algorithmic}
\end{algorithm}

The algorithmic procedure based on this convexification approach is summarized as part of the overall Algorithm \ref{alg:stable-identification} (lines 9-14).
Problem \eqref{prob:relaxed-stable-model} is solved repeatedly for at most $n_{\rm ico}$ iterations with updated convexification center $(\bar{\mat{P}}_+, \bar{\mat{P}}_-, \bar{\boldsymbol{\beta}})$ each time. At the $\ell$-th iteration, the optimal cost $\gamma_\ell$ is not greater than $\gamma_{\ell-1}$ and the procedure is terminated if the relative change with respect to $\gamma_{\ell-1}$ is smaller than a threshold $\varepsilon_{\rm ico}$. 

Before starting the ICO iteration, the initial part of Algorithm \ref{alg:stable-identification} comprising lines 2-6 is the standard PSK iteration with the additional conservative stability constraint based on ESP. The model retrieved after this first loop is thus expected to have a suboptimal cost $\gamma_0$. The purpose of the ICO iteration based on the exact stability criterion is to find a more accurate model by removing conservatism. 
Note that the ICO procedure leads to the convex problem \eqref{prob:relaxed-stable-model} whose solution satisfies the constraints of the original nonconvex problem \eqref{prob:stable-model}, but is not guaranteed to be globally optimal, as typically happens in nonconvex optimization. In general, ICO only converges to a local optimum as observed in \cite{warner-ico,ico}. However, ICO  returns a sequence of nonincreasing costs throughout iterations, i.e.  $\gamma_0\geq \gamma_1\geq\gamma_2\geq...$. Hence, this procedure is guaranteed to produce a model with a lower cost than the initial guess based on ESP even at the first ICO iteration, and this cost becomes lower as iterations proceed. 
The final model is thus at least as accurate as the one found with the ESP criterion. In practice, numerical results confirm that the final stable model has substantially better accuracy. 

\section{Numerical Results}\label{sec:results}
Four numerical examples are reported in this section to demonstrate the performance of Algorithm \ref{alg:stable-identification}. These include 
the synthetic model from \cite{penzl2006} as well as realistic electronic systems, namely transmission line networks and a high-speed data link.
We describe the experimental methodology before discussing the results.
\paragraph{Parameterization} In the following experiments we consider polynomial and rational parameterizations. Both are constructed from univariate bases defined along each of the $k$ parameters. The complexity of the parameterization is encoded in the $k$-tuple $\boldsymbol{\rho} = (\rho_1,\dots,\rho_k)$, where $\rho_{k'}$ is the number of univariate basis functions along dimension $k'$. The multivariate basis functions in $\boldsymbol{\varphi}(\mvartheta)$ are then taken to be the tensor-product basis using the construction in Appendix \ref{app:parameterization}. The overall dimension is thus $\rho = \rho_1\cdots\rho_k$. 

For the polynomial parameterization, we use the Chebyshev polynomial bases of degrees $\rho_{k'}-1$ along dimension $k'$. On the other hand, the rational (Rat.) parameterization is implemented through a partial fraction basis. The construction consists in choosing $\lfloor\rho_{k'}/2\rfloor$ complex poles that are uniformly spaced on a segment $[0,1]+h\jj $, i.e. $p_{k',n} = \frac{n-1}{\lfloor\rho_{k'}/2\rfloor-1}+h\jj$, $n = 1,\dots,\lfloor\rho_{k'}/2\rfloor$ with a non-zero imaginary part $h$. In case $\rho_{k'}\leq 3$, $p_{k',1}=1+h\jj$. Elementary univariate basis functions are defined along each dimension $k'$ as
\begin{multline}\label{eq:rat-basis}
    \left\{
        \re\left(\frac{1}{\vartheta_{k'}-p_{k',1}}\right), \;\;\operatorname{Im}\left(\frac{1}{\vartheta_{k'}-p_{k',1}}\right), \;\; \cdots,
        \right.
        \\
        \left.\;\; \re\left(\frac{1}{\vartheta_{k'}-p_{k',\lfloor\rho_{k'}/2\rfloor}}\right), \;\; \operatorname{Im}\left(\frac{1}{\vartheta_{k'}-p_{k',\lfloor\rho_{k'}/2\rfloor}}\right)
    \right\}
\end{multline}
with the realization in Appendix \ref{app:parameterization}.
If $\rho_{k'}$ is odd, the basis is completed with an additional simple fraction centered at a real pole outside $[0,1]$.

\paragraph{Comparative analysis} For all examples, \emph{ESP} models are those obtained after running the first part of Algorithm \ref{alg:stable-identification} (up to line 7). These are compared with so-called \emph{stable} models, which are obtained at the end of the algorithm after the ICO iteration. Higher accuracy is expected in stable models as they are based on the exact criterion of Prop. \ref{prop:stab-criterion}. We fix $n_{\rm PSK}=5$ PSK iterations in the first part of Algorithm \ref{alg:stable-identification} and a termination threshold $\varepsilon_{\rm ico} = 10^{-2}$ for the ICO iteration. The elevation degree for Bernstein relaxation is set to $\bar{\delta} = 6$. 

Model accuracy is first evaluated in terms of the value of the PSK cost function $J(\boldsymbol{\beta})$. The optimal cost is denoted as $\gamma_0$ for ESP models as defined in Algorithm \ref{alg:stable-identification}, and $\gamma_{\rm opt}$ for stable models. Besides this quantity, a more significant evaluation is done in terms of the actual \emph{RMS error}, that is the Root-Mean-Square error between data and model across all points in the dataset. In formulae,
$\sqrt{\sum_{i=1}^{I}|\data{H}^{(i)} - H(s^{(i)},\mvartheta^{(i)})|^2/I}.
$

In all examples, after fixing the number $\nu$ of basis poles $\{\sigma_n\}_{n=1}^{\nu}$, their location is chosen by running the standard VF algorithm on samples of $\data{H}(s,\mvartheta)$ for a fixed $\mvartheta$, as suggested by \cite{triverio2009psk}.
Optimization problems are solved using MOSEK \cite{mosek} through the YALMIP interface \cite{yalmip}. Our MATLAB implementation is tested on a 3.3-GHz machine with 64 GB of RAM.

\begin{table}[]
\centering
\begin{tabular}{cccccccc}
\toprule
 & \multicolumn{3}{c}{} 
 & \multicolumn{2}{c}{ESP model} 
 & \multicolumn{2}{c}{Stable model} \\
\cmidrule(lr){2-4}\cmidrule(lr){5-6}\cmidrule(lr){7-8}
Example & Param. & $\nu$ & $\boldsymbol{\rho}$ 
        & $\gamma_0$ & RMS err. 
        & $\gamma_{\rm opt}$ & RMS err. \\
\midrule
\multirow{6}{*}{\makecell{Penzl\\(Sec. \ref{sec:penzl})}} & Poly  & 14 & $4$ & $1.79\cdot 10^{-1}$ & $3.12\cdot 10^{-4}$ & $3.5\cdot 10^{-2}$ & $6.15\cdot 10^{-5}$ \\ 
 & Poly  & 14 & $5$ & $1.52\cdot 10^{-1}$ & $2.54\cdot 10^{-4}$ & $6.91\cdot 10^{-3}$ & $1.18\cdot 10^{-5}$ \\
 & Poly  & 14 & $6$ & $1.52\cdot 10^{-1}$ & $2.37\cdot 10^{-4}$ & $8.49\cdot 10^{-3}$ & $1.35\cdot 10^{-5}$ \\ 
 & Rat.  & 14 & $4$ & $1.13\cdot 10^{-1}$ & $1.7\cdot 10^{-4}$ & $1.06\cdot 10^{-1}$ & $1.62\cdot 10^{-4}$ \\ 
 & Rat.  & 14 & $5$ & $1.53\cdot 10^{-2}$ & $1.12\cdot 10^{-5}$ & $9.62\cdot 10^{-3}$ & $7.54\cdot 10^{-6}$ \\ 
 & Rat. & 14 & $6$ & $1.74\cdot 10^{-2}$ & $1.72\cdot 10^{-5}$ & $9.64\cdot 10^{-3}$ & $1.03\cdot 10^{-5}$ \\ 
\midrule
\multirow{12}{*}{\makecell{TL-RLC\\(Sec. \ref{sec:tl-example})}} & Poly  & 12 & $(4,3)$ & $9.08\cdot 10^{-1}$ & $3.11\cdot 10^{-2}$ & $2.82\cdot 10^{-1}$ & $1.18\cdot 10^{-2}$ \\
 & Poly  & 12 & $(5,3)$ & $4.42\cdot 10^{-1}$ & $1.18\cdot 10^{-2}$ & $1.68\cdot 10^{-1}$ & $4.57\cdot 10^{-3}$ \\
 & Poly  & 12 & $(5,4)$ & $4.29\cdot 10^{-1}$ & $1.04\cdot 10^{-2}$ & $1.79\cdot 10^{-1}$ & $4.11\cdot 10^{-3}$ \\
 & Poly  & 14 & $(4,3)$ & $9.83\cdot 10^{-1}$ & $2.83\cdot 10^{-2}$ & $2.39\cdot 10^{-1}$ & $8.12\cdot 10^{-3}$ \\
 & Poly  & 14 & $(5,3)$ & $5.51\cdot 10^{-1}$ & $1.02\cdot 10^{-2}$ & $2.16\cdot 10^{-1}$ & $4.41\cdot 10^{-3}$ \\
 & Poly  & 14 & $(5,4)$ & $5.17\cdot 10^{-1}$ & $9.23\cdot 10^{-3}$ & $2.96\cdot 10^{-1}$ & $5.26\cdot 10^{-3}$ \\
 & Rat.  & 12 & $(4,3)$ & $6.05\cdot 10^{-1}$ & $2.08\cdot 10^{-2}$ & $2.76\cdot 10^{-1}$ & $1.03\cdot 10^{-2}$ \\
 & Rat.  & 12 & $(5,3)$ & $2.14\cdot 10^{-1}$ & $7.84\cdot 10^{-3}$ & $1.27\cdot 10^{-1}$ & $5.09\cdot 10^{-3}$ \\
 & Rat.  & 12 & $(5,4)$ & $6.63\cdot 10^{-2}$ & $2.78\cdot 10^{-3}$ & $5.42\cdot 10^{-2}$ & $2.4\cdot 10^{-3}$ \\
 & Rat.  & 14 & $(4,3)$ & $5.92\cdot 10^{-1}$ & $1.91\cdot 10^{-2}$ & $2.37\cdot 10^{-1}$ & $8.58\cdot 10^{-3}$ \\
 & Rat.  & 14 & $(5,3)$ & $2.05\cdot 10^{-1}$ & $6.97\cdot 10^{-3}$ & $1.25\cdot 10^{-1}$ & $4.42\cdot 10^{-3}$ \\
 & Rat. & 14 & $(5,4)$ & $5.53\cdot 10^{-2}$ & $2.22\cdot 10^{-3}$ & $4.02\cdot 10^{-2}$ & $1.87\cdot 10^{-3}$ \\
\midrule
\multirow{4}{*}{\makecell{Data link\\(Sec. \ref{sec:schuster})}} & Poly  & 22 & $(3,3)$ & $2.45\cdot 10^{-2}$ & $4.28\cdot 10^{-3}$ & $8.65\cdot 10^{-3}$ & $2.72\cdot 10^{-3}$ \\
 & Poly  & 22 & $(4,4)$ & $1.9\cdot 10^{-2}$ & $2.36\cdot 10^{-3}$ & $3.61\cdot 10^{-3}$ & $5.51\cdot 10^{-4}$ \\
 & Rat.  & 22 & $(3,3)$ & $4.96\cdot 10^{-2}$ & $4.96\cdot 10^{-3}$ & $2.13\cdot 10^{-2}$ & $2.85\cdot 10^{-3}$ \\
 & Rat. & 22 & $(4,4)$ & $3\cdot 10^{-2}$ & $2.4\cdot 10^{-3}$ & $7.7\cdot 10^{-3}$ & $7.31\cdot 10^{-4}$ \\
\midrule
\multirow{2}{*}{\makecell{TL network\\(Sec. \ref{sec:tlnet-example})}} & Poly  & 14 & $(2,2,2)$ & $1.87\cdot 10^{-2}$ & $5.61\cdot 10^{-4}$ & $4.46\cdot 10^{-3}$ & $2.48\cdot 10^{-4}$ \\
 & Rat. & 14 & $(2,2,2)$ & $7.03\cdot 10^{-2}$ & $1.74\cdot 10^{-3}$ & $3.67\cdot 10^{-2}$ & $9.34\cdot 10^{-4}$ \\
\bottomrule
\end{tabular}
\caption{Comparison of ESP and stable models for different examples, parameterizations and model complexities.}
\label{tab:errors}
\end{table}

\subsection{Penzl model}\label{sec:penzl}
The Penzl model is a benchmark example that appears in \cite{ionita2014parameterized} as a parameterized version of an example originally presented in \cite{penzl2006}. It is given explicitly as $\data{H}(s,\mvartheta) = \vet{c}_{\rm pzl}[s\eye-\mat{A}_{\rm pzl}(\mvartheta)]^{-1}\vet{c}_{\rm pzl}^T$ with a block-diagonal $\mat{A}_{\rm pzl}(\mvartheta)=\mathrm{diag}(\mat{A}_1(\vartheta_1),\mat{A}_2,\mat{A}_3,\mat{A}_4)$,
\begin{align}
    \mat{A}_1(\vartheta_1)
    &=
    \begin{pmatrix}
    -1 & 100(0.5+\vartheta_1) \\ -100(0.5+\vartheta_1) & -1
    \end{pmatrix}
    ,\quad\mat{A}_2=
    \begin{pmatrix}
    -1 & 200 \\ -200 & -1
    \end{pmatrix},
    \\
    \mat{A}_3&=
    \begin{pmatrix}
    -1 & 400 \\ -400 & -1
    \end{pmatrix},\quad
    \mat{A}_4=\mathrm{diag}\{-1,-2,\dots,-10^3\},
\end{align}
$\vet{c}_{\rm pzl}=\left(1 \; 0 \; 1\;0\;1\;0\;1\;1\;\cdots \;1\right)$. The Penzl model is specifically constructed to have a pair of complex poles at $ -1\pm (50+100\vartheta_1)\jj$ whose imaginary part varies in $[50,150]$ as $\vartheta_1$ varies in $[0,1]$. The initial dataset is made up of samples of $\data{H}(s,\mvartheta)$ evaluated at $s=\jj 2\pi f$ for $2\cdot 10^3$ uniformly spaced values of $f$ in the interval $[1,100]$. As for $\mvartheta$, 50 uniformly spaced values are considered in $[0,1]$. The full-order transfer function is thus evaluated on a grid in $(s,\mvartheta)$ corresponding to all $10^5$ combinations of these $s$ and $\mvartheta$ values.

For this benchmark we set $n_{\rm ico} = 20$ and $h=0.5$.
With fixed $\nu=14$, several values of $\boldsymbol{\rho}$ are considered with the results summarized in Table \ref{tab:errors}. The ICO iteration reduces the cost function value by more than an order of magnitude (for polynomial parameterization with $\boldsymbol{\rho}= 6$). It is found that the relative gap between $\gamma_0$ and $\gamma_{\rm opt}$ for rationally-parameterized models is not as large as with the polynomial parameterization.

Similarly, the data in Table \ref{tab:errors} shows that the effective RMS error of stable models built with the exact criterion can be more than ten times smaller than that of ESP models. The same is confirmed by Fig.~\ref{fig:penzl-contour}, which compares the model-data error at every point in the $(s,\mvartheta)$ domain available in the dataset, considering polynomial parameterization and $\vet{\rho}=5$.

\begin{figure}
    \centering
    \subfloat[\label{fig:penzl-ico}]{
\includegraphics[width=0.49\textwidth]{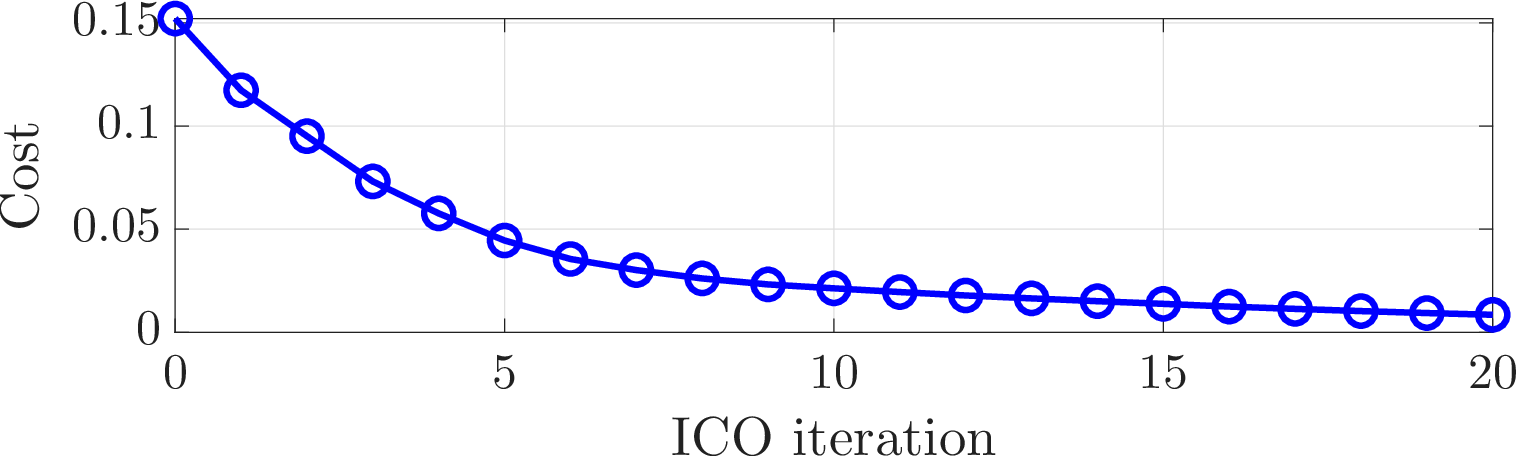}
    }
    \subfloat[\label{fig:penzl-zeros}]{
\includegraphics[width=0.43\textwidth]{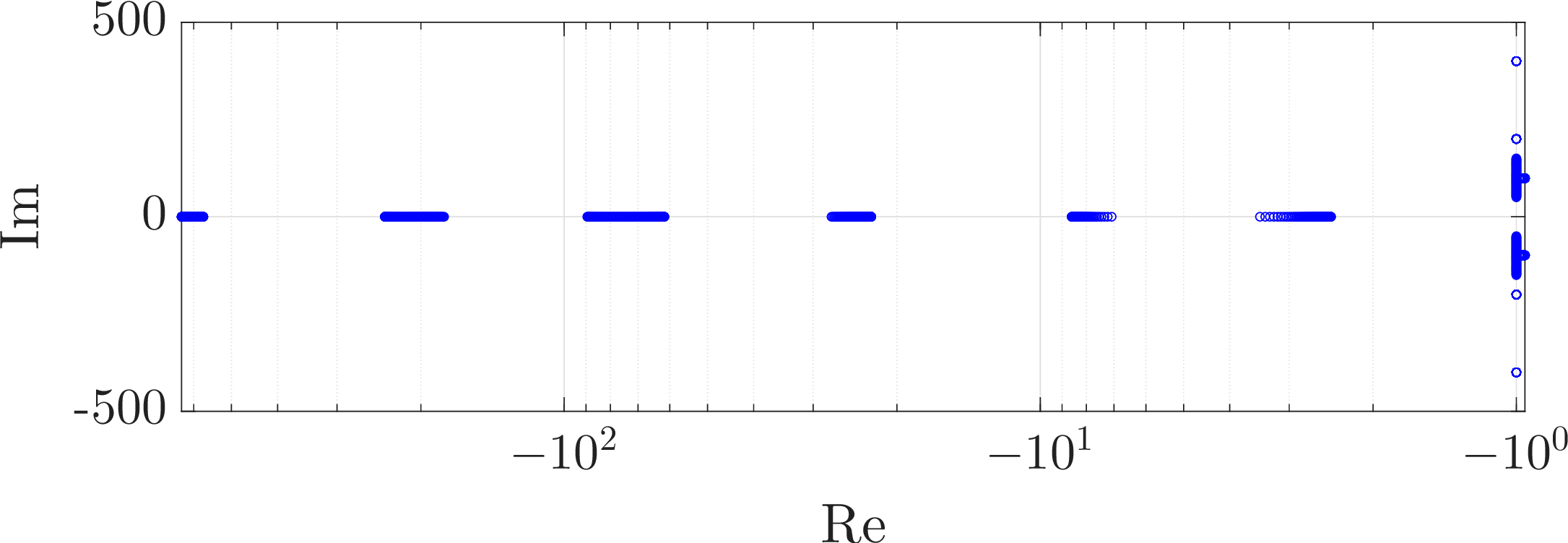}
}
    \caption{Penzl model with $\nu = 14$, $\boldsymbol{\rho} = 6$. (a) Cost function throughout ICO iterations. (b) Location of denominator zeros for $100$ fixed values of $\mvartheta$.}
    \label{fig:penzl-ico-zeros}
\end{figure}

As for runtime, it takes less than five seconds to build the ESP model with $\boldsymbol{\rho}=6$, and 211 seconds to complete twenty ICO iterations to obtain the stable model with the same complexity. The trend of the cost throughout ICO iterations is reported in Fig.~\ref{fig:penzl-ico}, where the zeroth iteration is $\gamma_0$, i.e. the value of $J(\vet{\beta})$ at the end of the first phase of the algorithm where the ESP criterion is enforced, and before starting the actual ICO iteration with the exact criterion. Figure \ref{fig:penzl-zeros} confirms that the model is stable as the denominator zeros are in the left half-plane as $\mvartheta$ varies in $[0,1]$.

\begin{figure}
\centering
\subfloat[Polynomial, $\boldsymbol{\rho}= 4$]{
    \includegraphics[width=0.49\textwidth]{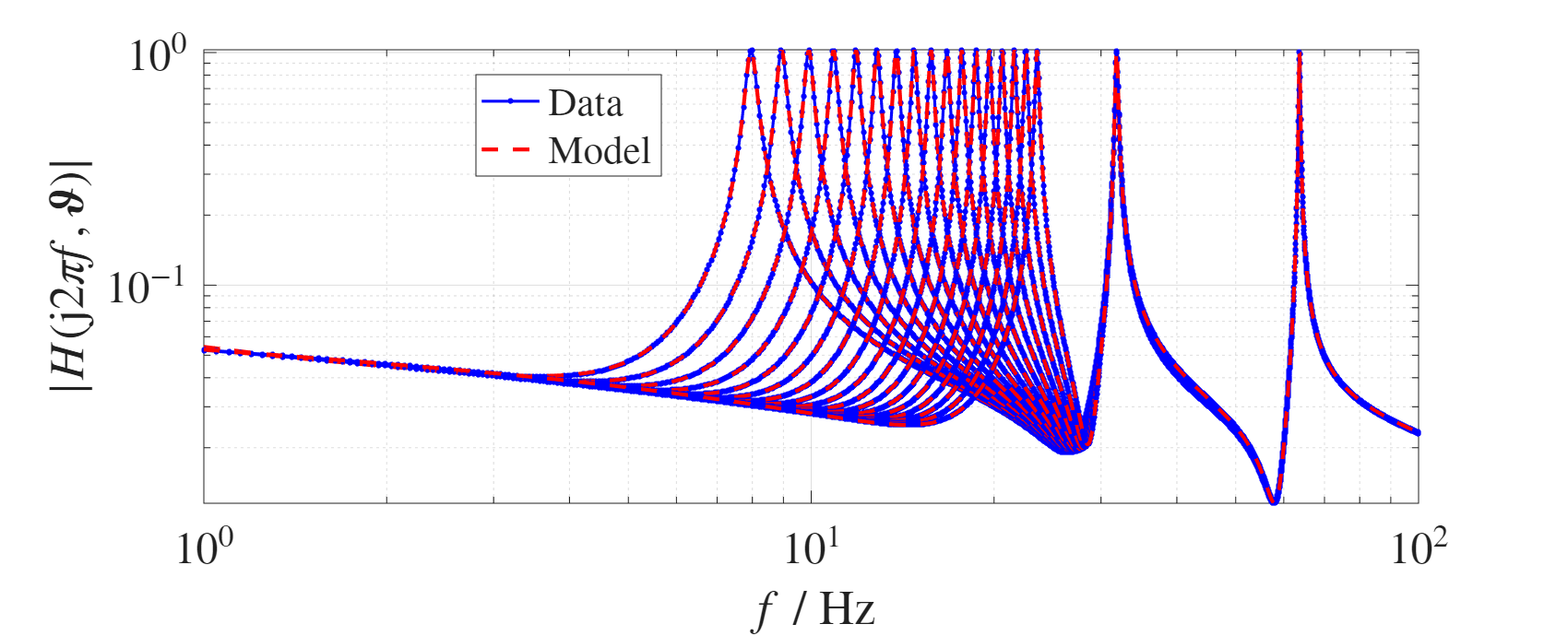}
}
\subfloat[Polynomial, $\boldsymbol{\rho} = 6$]{
    \includegraphics[width=0.49\textwidth]{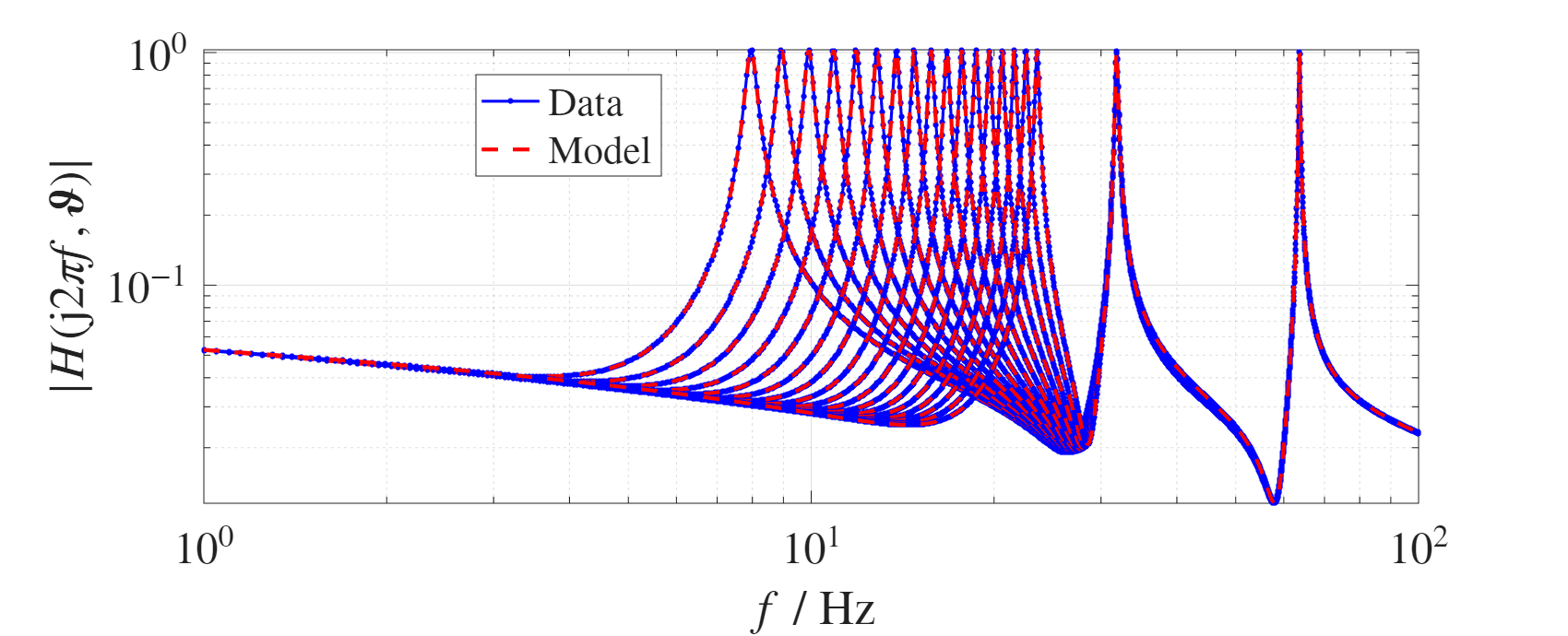}
    }
    \\
\subfloat[Rational, $\boldsymbol{\rho} = 4$]{
    \includegraphics[width=0.49\textwidth]{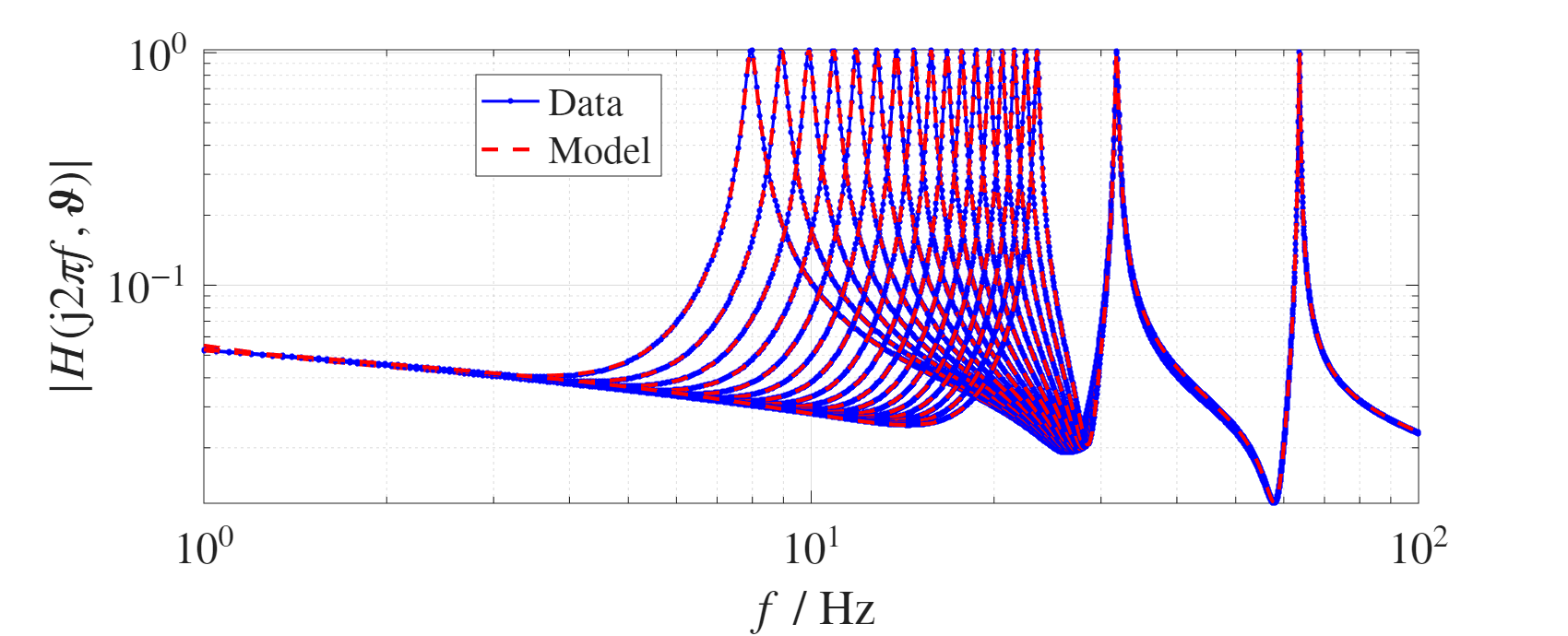}}
\subfloat[Rational, $\boldsymbol{\rho} = 6$]{
    \includegraphics[width=0.49\textwidth]{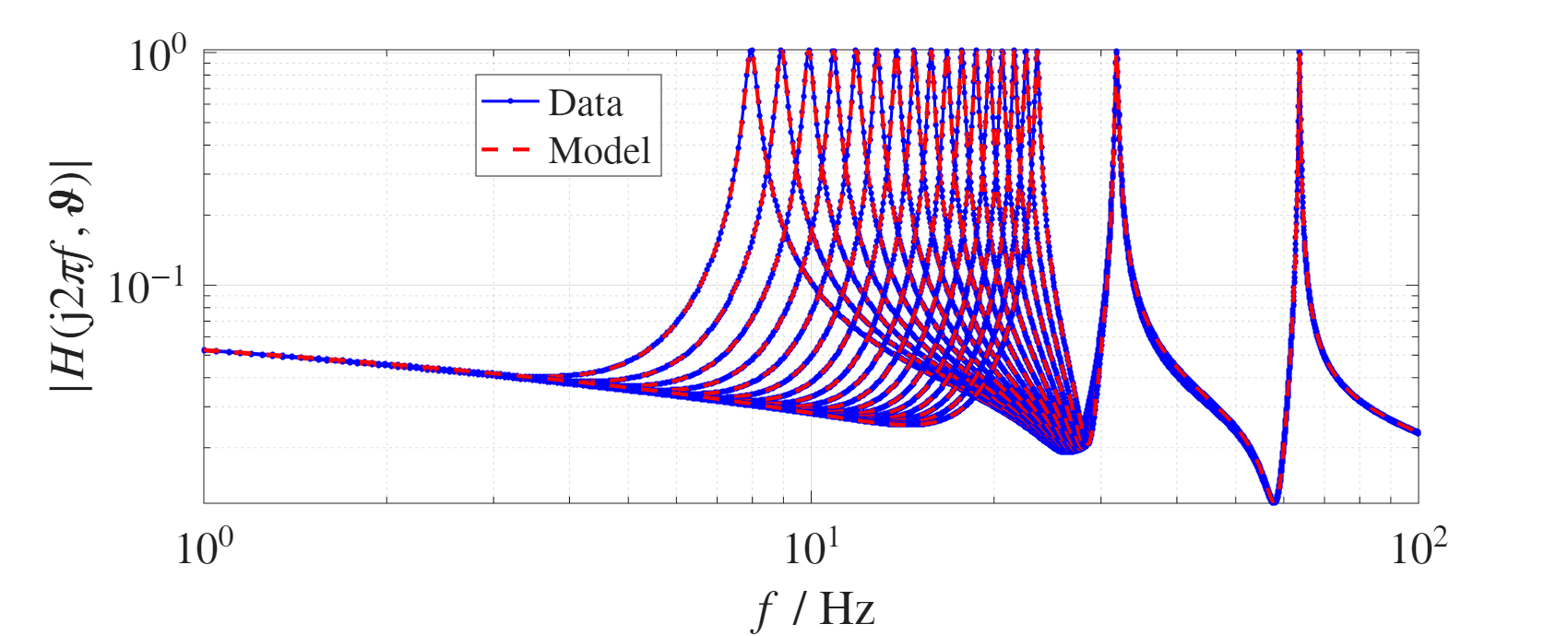}
    }
    \caption{Data-model comparison for the Penzl model, built with polynomial or rational parameterization. Different curves correspond to different fixed values of $\mvartheta$.}
    \label{fig:penzl-datamodel}
\end{figure}

\begin{figure}
    \centering
    \includegraphics[width=0.8\linewidth]{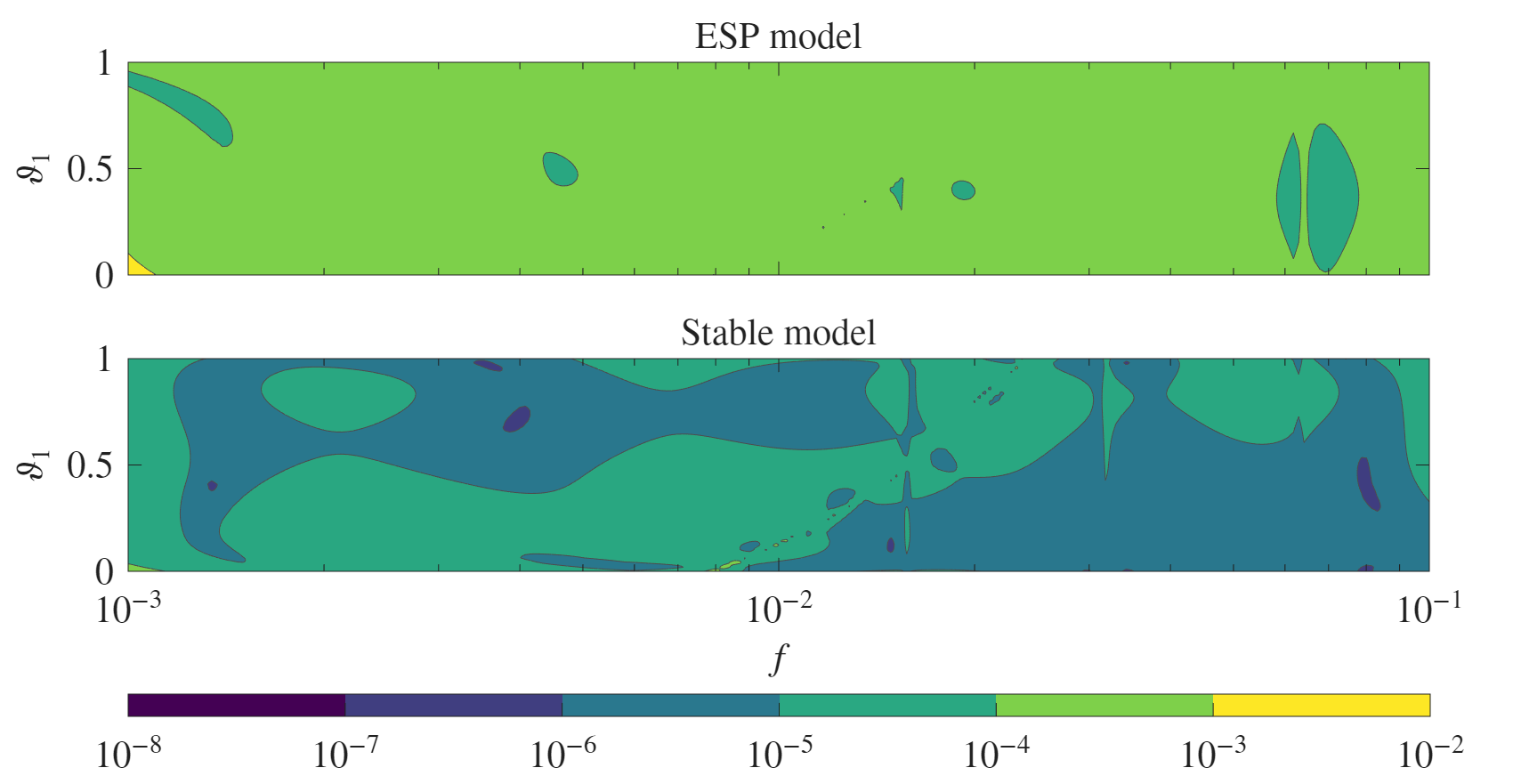}
    \caption{Penzl model of Sec. \ref{sec:penzl}. Contour plot of the model-data error for each point $(\jj2\pi f,\vartheta_1)$ in the dataset.}
    \label{fig:penzl-contour}
\end{figure}

\subsection{Transmission Line with RLC load (TL-RLC)}\label{sec:tl-example}
\begin{figure}
    \centering
\includegraphics[width=0.7\textwidth,clip,trim={0 23.7cm 9cm 0.05cm}]{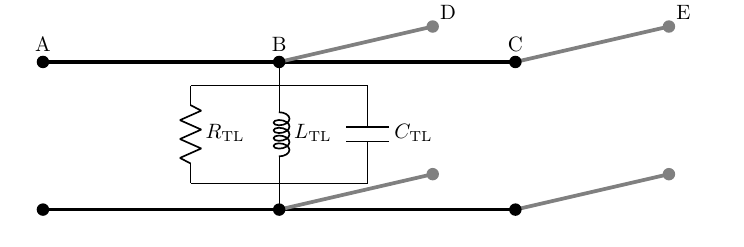}
    \caption{Network of transmission line segments and lumped components.}
    \label{fig:tline-diagram}
\end{figure}
\begin{figure}
\subfloat[Polynomial,  $\boldsymbol{\rho}= (5,3)$]{
\includegraphics[width=0.49\textwidth]{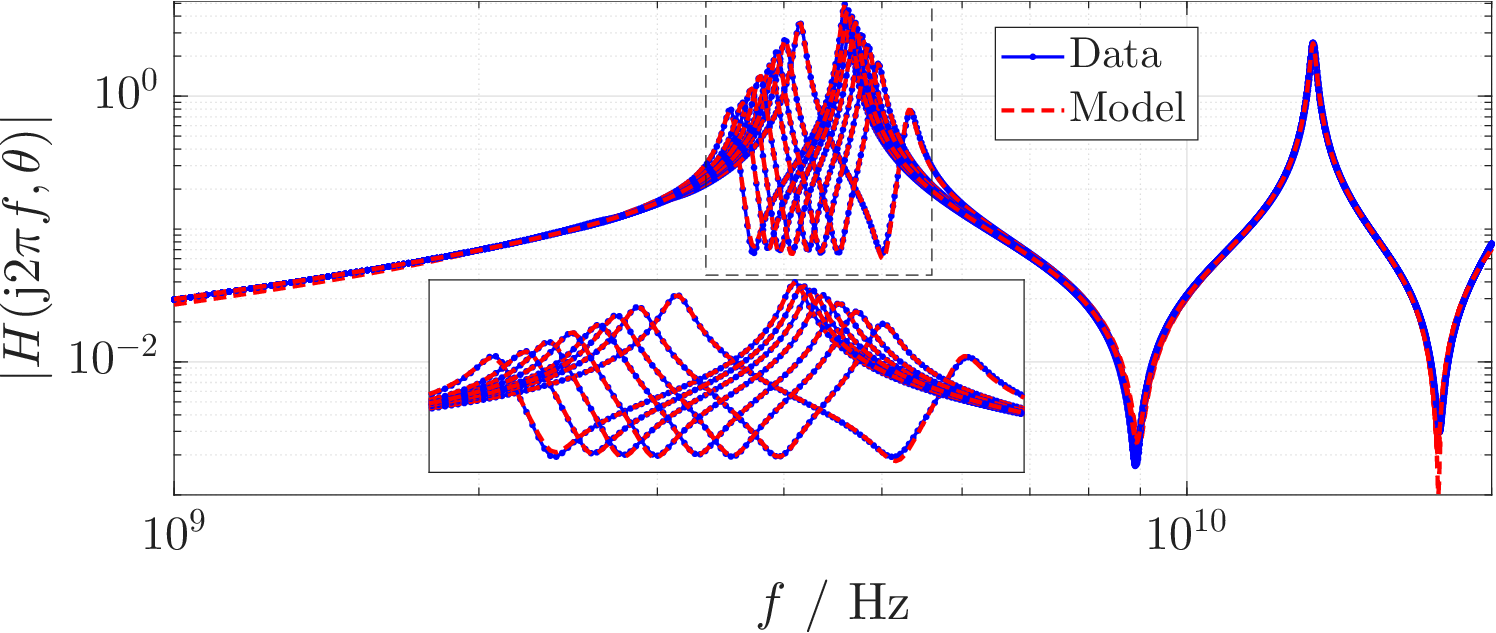}
}
\subfloat[Rational, $\boldsymbol{\rho} = (5,3)$]{
    \includegraphics[width=0.49\textwidth]{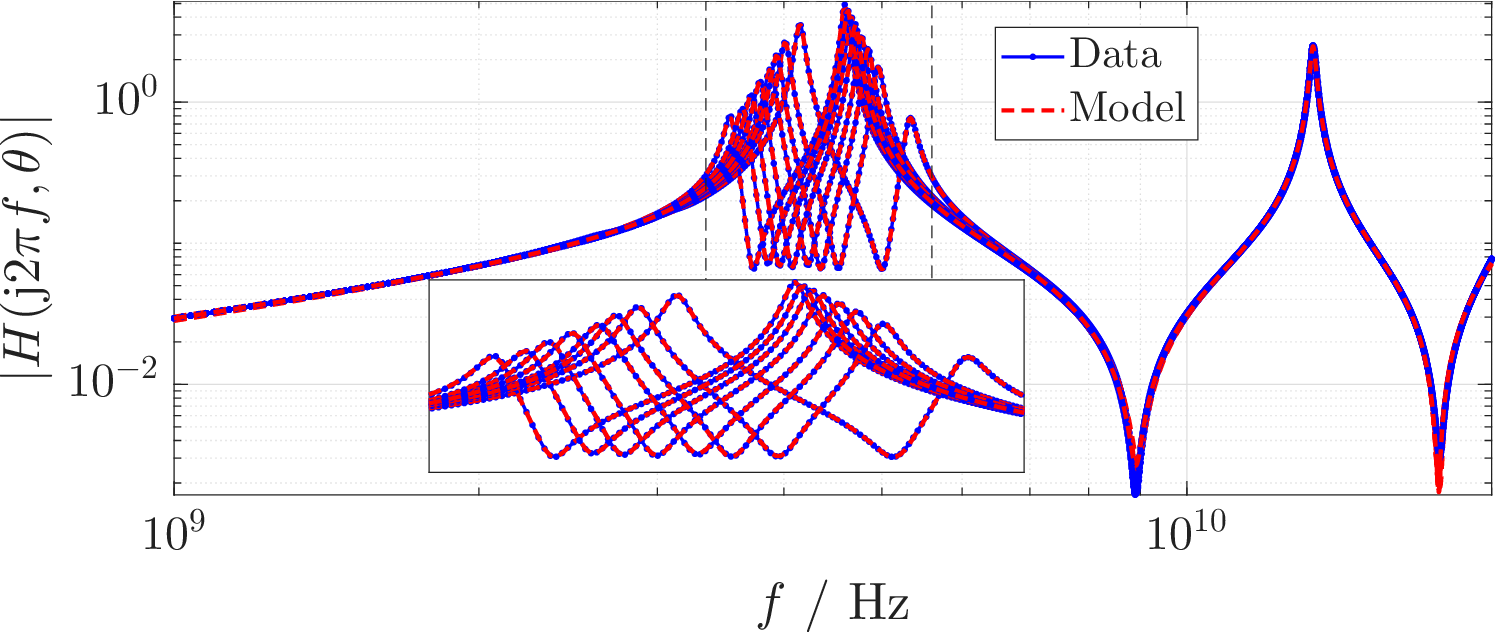}
}
\\
\subfloat[Polynomial, $\boldsymbol{\rho} = (5,4)$]{
    \includegraphics[width=0.49\linewidth]{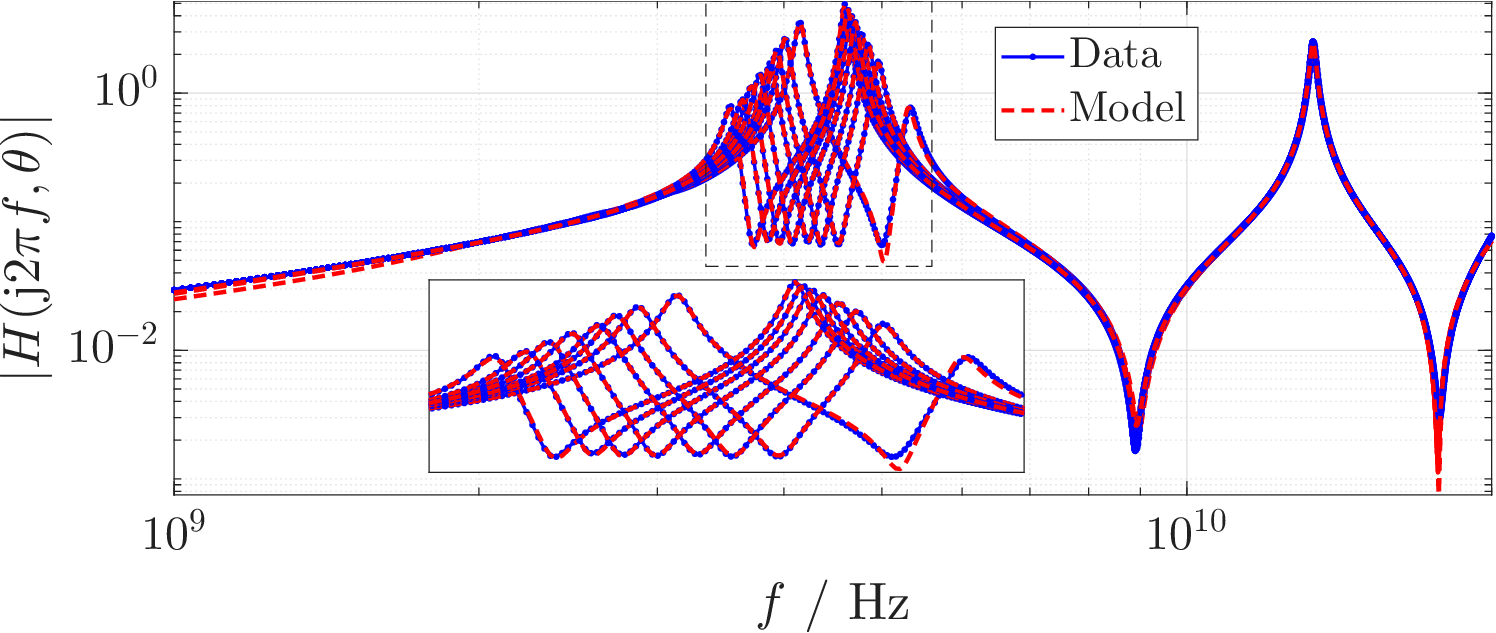}
    }
\subfloat[Rational, $\boldsymbol{\rho} = (5,4)$]{
    \includegraphics[width=0.49\linewidth]{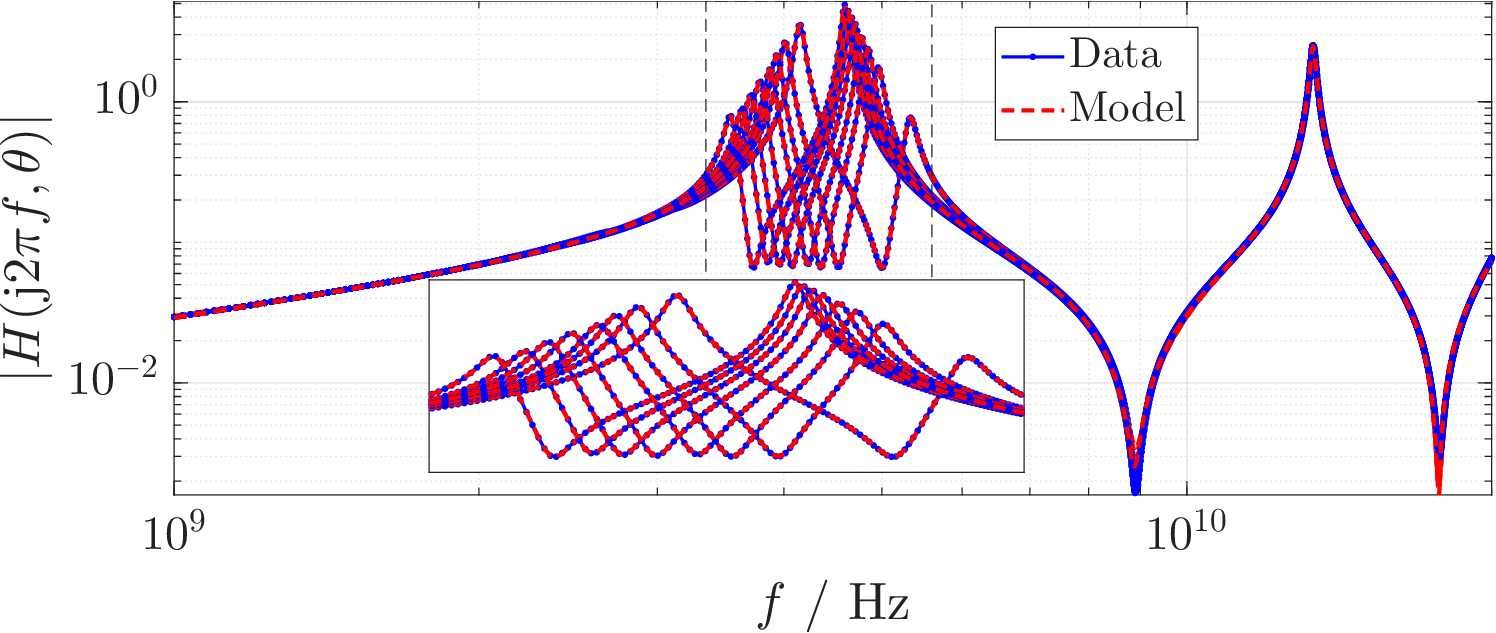}
   }
    \caption{Data-model comparison for the TL model of Sec. \ref{sec:tl-example} with $\nu = 14$. Different curves correspond to different fixed values of $\mvartheta$.}
    \label{fig:tline-datamodel}
\end{figure}
This section considers an electrical network made up of Transmission Line (TL) segments and lumped circuit elements (RLC), as shown in Fig.~\ref{fig:tline-diagram}. This electrical system is modeled in a SPICE-based circuit solver where TL segments correspond to microstrip lines whose conductors are $150\,\mu\mathrm{m}$ wide, $50\,\mu\mathrm{m}$ thick, and are placed on a $200\,\mu$m-thick SiO\textsubscript{2} dielectric layer (with $\epsilon_r=4.1$) that separates them from a perfectly conducting metallic backplane. This arrangement represents a single-layer printed circuit board (PCB). The length of the A-B segment is $1\,\mathrm{cm}$. An RLC resonant circuit is connected at point B, with fixed resistance $R_{\rm TL} = 120\,\Omega$ and inductance $L_{\rm TL} = 0.1\,\mathrm{nH}$. The capacitance $C_{\rm TL} = 10(1+\vartheta_1)\,\mathrm{pF}$ is parameterized by $\vartheta_1$ and it varies in $[10,20]\,\mathrm{pF}$. The TL segment from B to C has length $\ell = (1+0.4\vartheta_2)\,\mathrm{cm}$ and thus ranges in $[1,1.4]\,\mathrm{cm}$. The two additional TL stubs placed between B-D and C-E have fixed length $3.75\,\mathrm{mm}$. 

The input impedance at point A is a transfer function $\data{H}(s,\mvartheta)$ depending on $k=2$ parameters. Samples of $\data{H}(\jj2\pi f,\mvartheta)$ are collected at $10^3$ values of $f$ uniformly spaced in $[1,20]\,\mathrm{GHz}$ using the circuit solver HSPICE. 
These samples are evaluated for $50$ different values of $\mvartheta$ chosen according to a space-filling Sobol sequence in $[0,1]^2$. The dataset thus consists of $5\cdot 10^4$ evaluations of the target transfer function. The presence of distributed components implies that the TF of the original system is not rational. 

\begin{figure}
    \centering
\subfloat[\label{fig:tline-ico}]{
\includegraphics[width=0.49\linewidth]{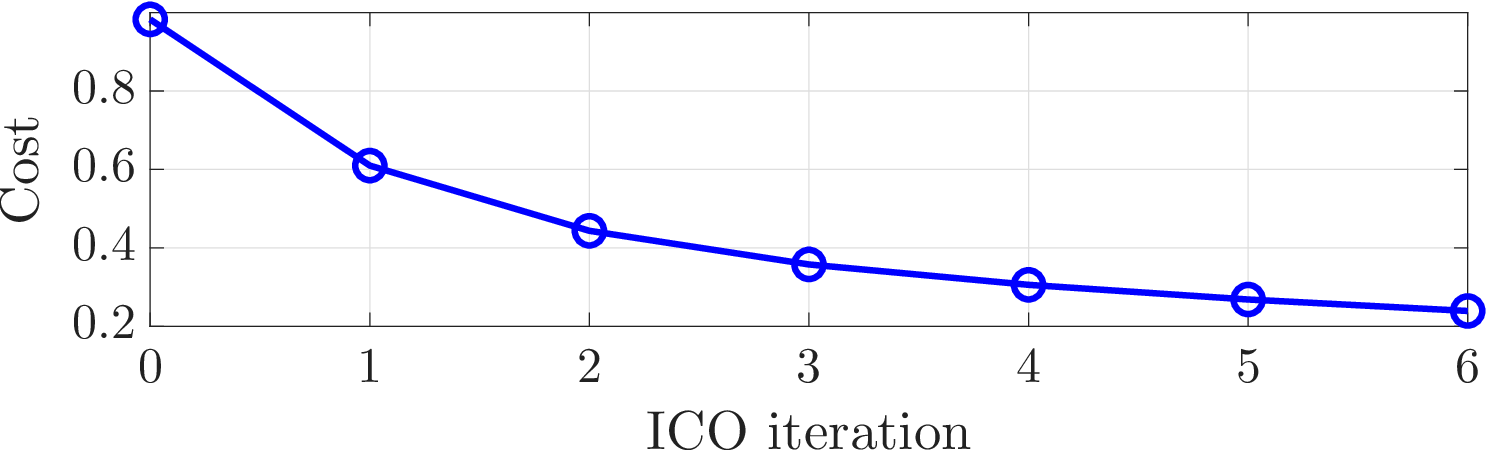}
}
    \subfloat[\label{fig:tline-poles}]{
\includegraphics[width=0.44\linewidth]{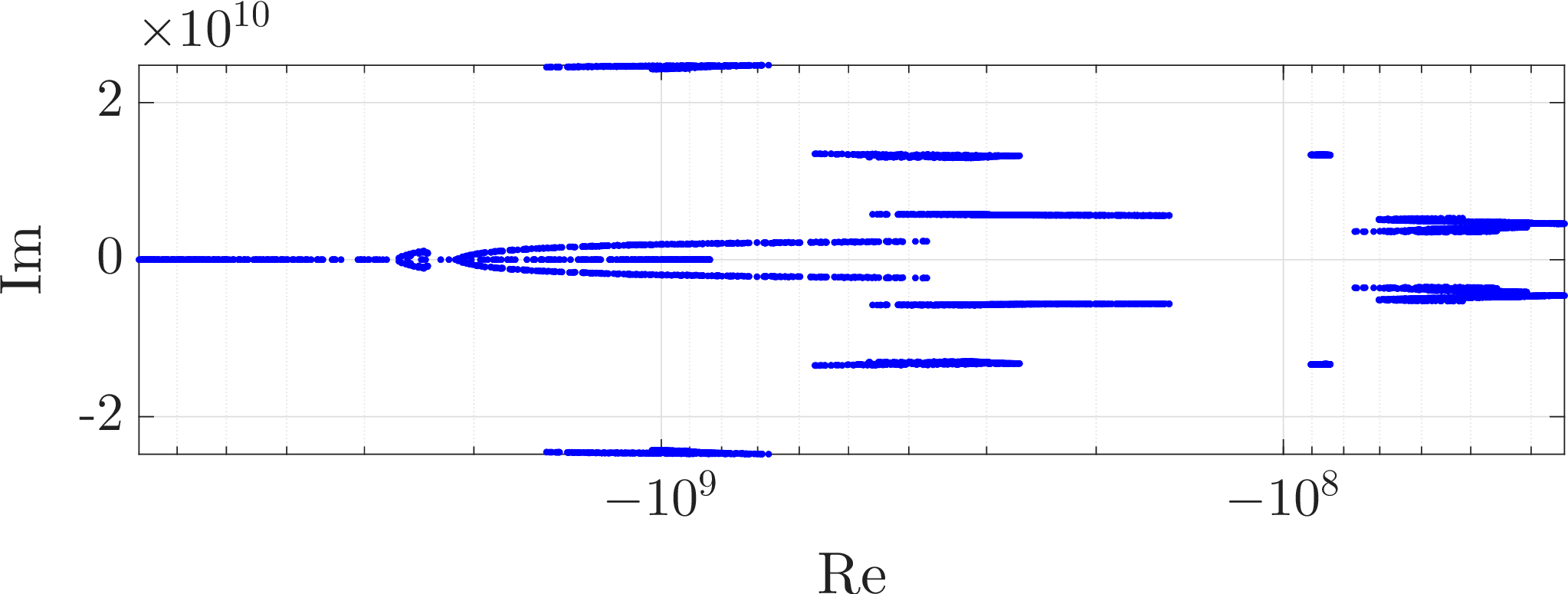}
    }
    \caption{TL model of Sec. \ref{sec:tl-example}, with $\nu=14$, $\boldsymbol{\rho}=(4,3)$. (a) Optimal cost along ICO iterations. (b) Location of denominator zeros for 100 values of $\mvartheta$.}
    \label{fig:tline-poles-hist}
\end{figure}
    
We fix $n_{\rm ico} = 6$ and $h = 0.5$.
The proposed method is used to construct models of varying complexity. The number of basis poles is set to either $\nu=12$ or $\nu=14$ and several combinations of $\boldsymbol{\rho}$ are considered. A graphical comparison between data and model is reported in Fig.~\ref{fig:tline-datamodel} for four particular choices of model parameterization and complexity. The complete set of experiments on this test case is summarized in Table \ref{tab:errors}. It is found that using the exact stability criterion decreases the RMS error by up to a factor of three. For this example, using a polynomial parameterization usually yields a slightly lower error than the partial fraction basis with the same complexity $\boldsymbol{\rho}$. However, there are also instances where the error of the rationally parameterized model is about half (e.g. $\boldsymbol{\rho}=(5,4)$, $\nu=12$) or one third of the error of the polynomially parameterized model with the same complexity (e.g. $\boldsymbol{\rho}=(5,4)$, $\nu=14$).

The behavior of the ICO iteration is reported in Fig.~\ref{fig:tline-ico}, showing the optimal cost function value for each iteration. Considering the final model, the locations of denominator zeros in the complex plane are shown in Fig.~\ref{fig:tline-poles} for 100 values of $\mvartheta$. Considering the model with $\nu=14$, $\boldsymbol{\rho} = (5,4)$, the first part of the algorithm takes 118 seconds to complete the PSK iterations and retrieve the ESP model. Then, six ICO iterations take 26 minutes (approximately $ 4.3$ minutes per iteration).

\subsection{High-speed Data Link}\label{sec:schuster}
\begin{figure}
    \centering
    \subfloat[Polynomial, $\boldsymbol{\rho} = (3,3)$]{
    \includegraphics[width=0.48\textwidth]{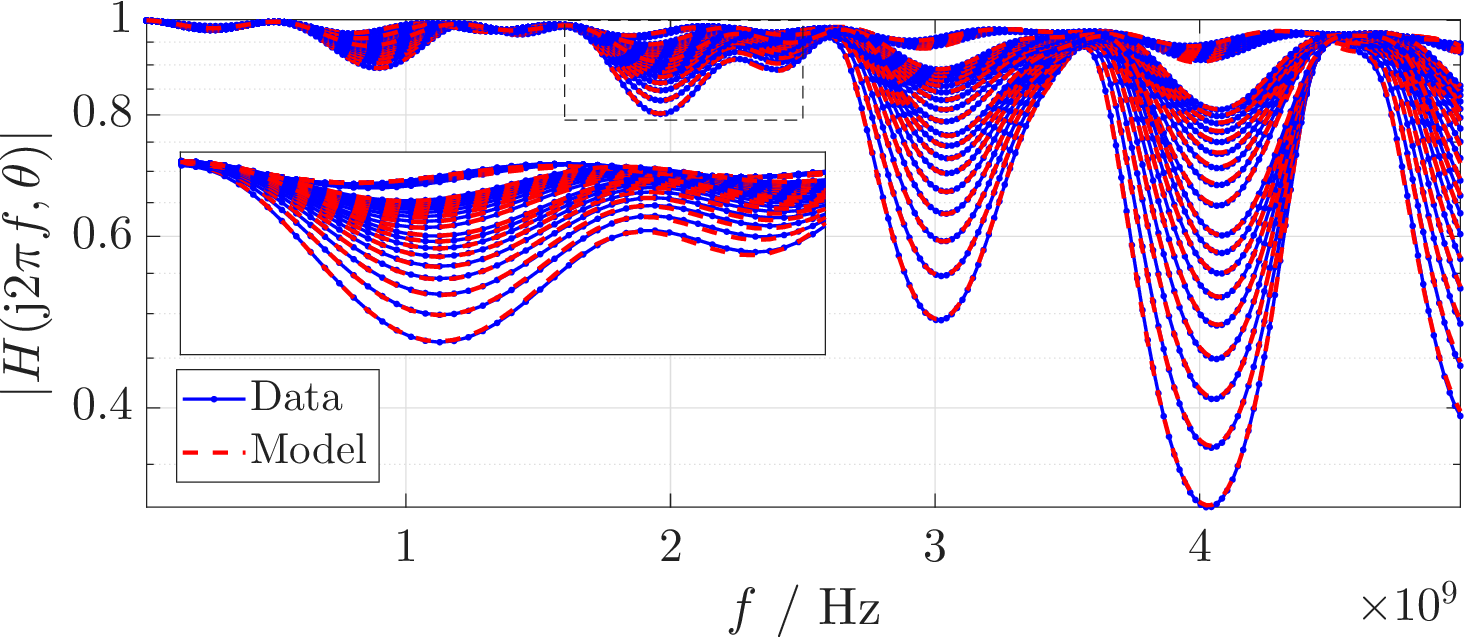}
    }
    \subfloat[Rational, $\boldsymbol{\rho} = (3,3)$]{\includegraphics[width=0.48\textwidth]{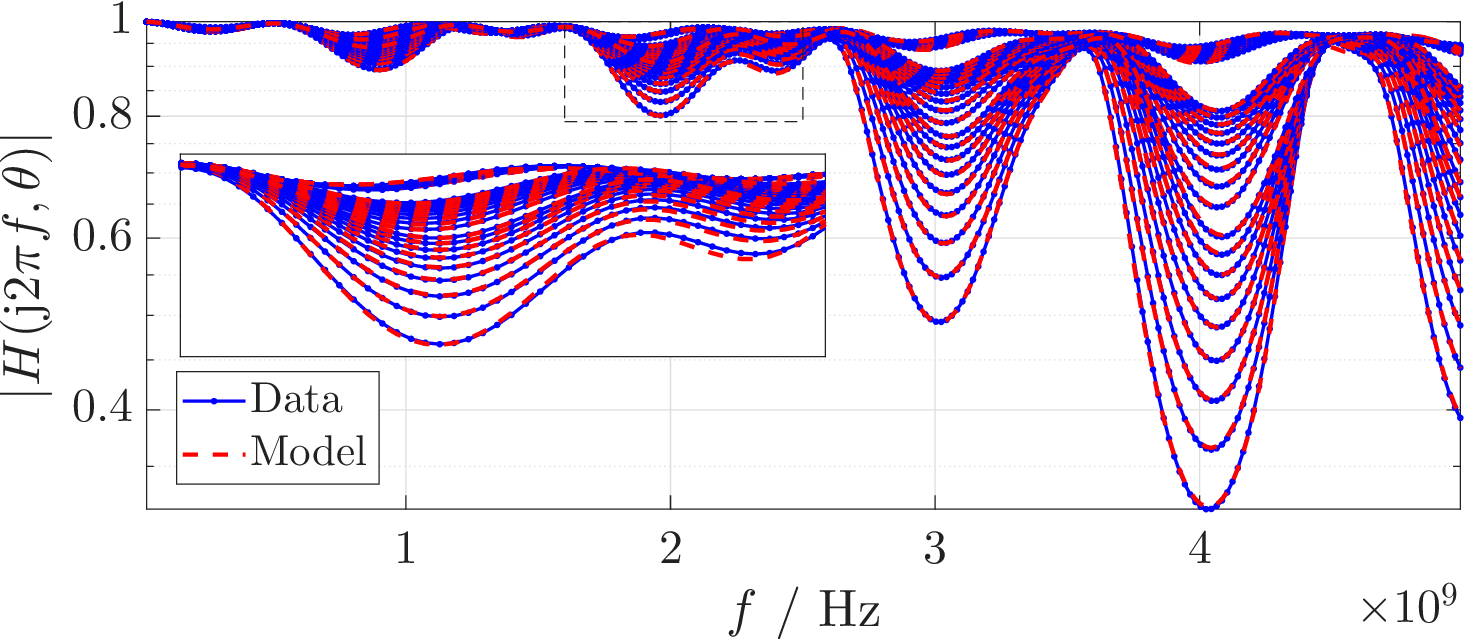}
    }
    \\
    \subfloat[Polynomial, $\boldsymbol{\rho}=(4,4)$]{
\includegraphics[width=0.48\textwidth]{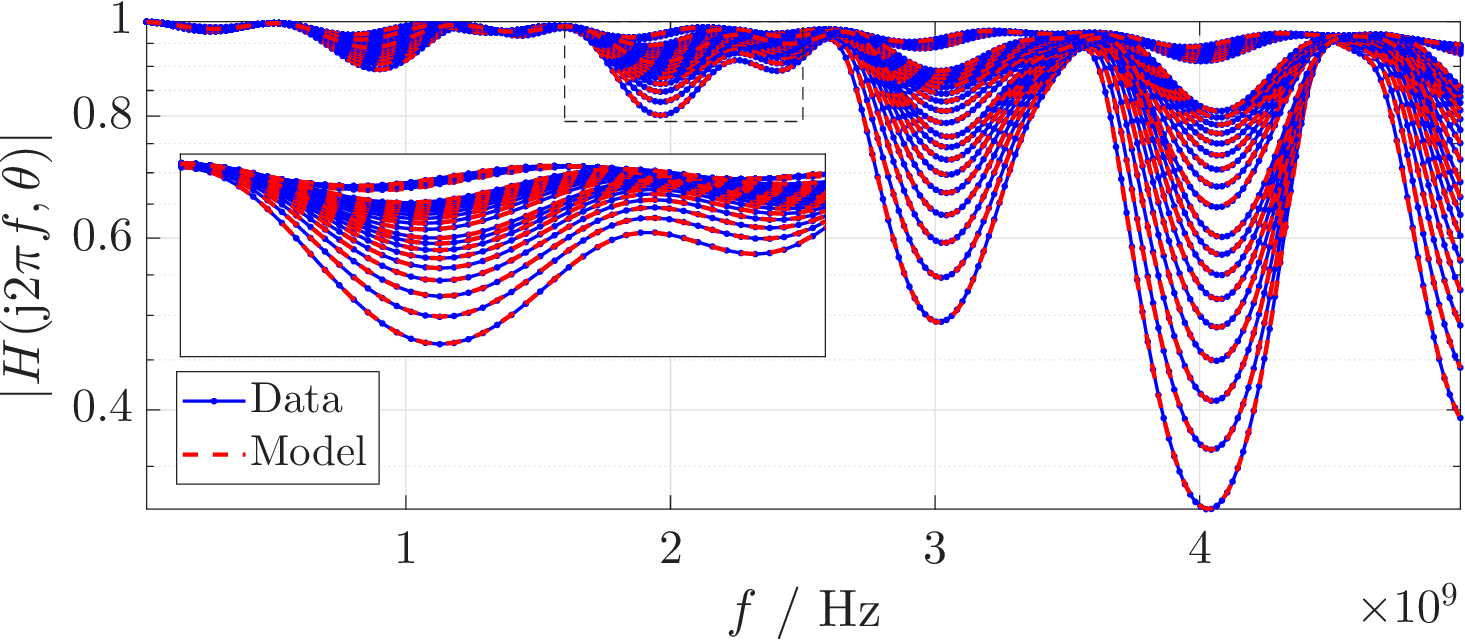}
    }
    \subfloat[Rational, $\boldsymbol{\rho}=(4,4)$]{
\includegraphics[width=0.48\textwidth]{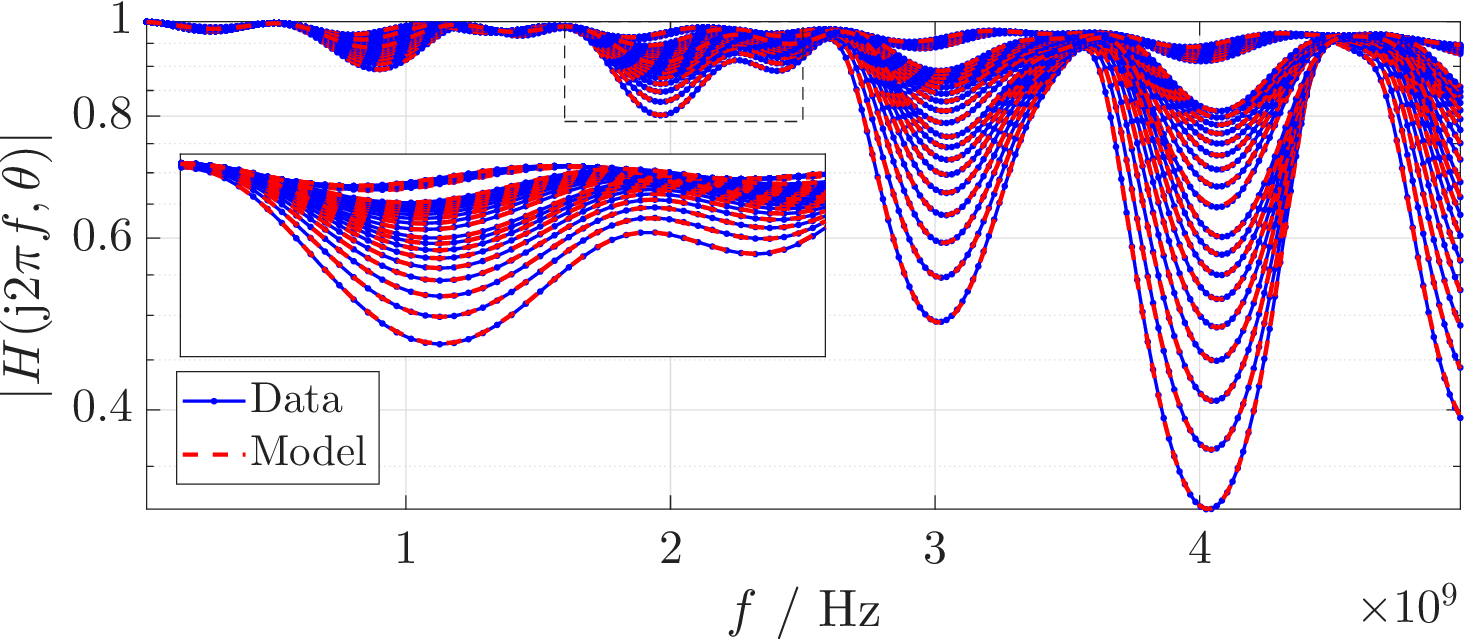}
    }
    \caption{Data-model comparison for the data link example in Sec. \ref{sec:schuster}.}
    \label{fig:schuster-responses}
\end{figure}
This section shows an application with $k=2$ parameters, originally presented in \cite{preibisch2017} and used in \cite[Sec. VII-C]{bradde-stability}. It consists in a PCB interconnect designed to operate at up to $5$ GHz to transmit signals between two electrical terminals (ports) placed on different PCB layers. This link includes a PCB via, whose geometry depends on the via radius $r_v$ and the antipad radius $r_a$. The link efficiency is characterized by the transmission coefficient. This quantity is the target $\data{H}(s,\mvartheta)$ to be modeled, and it is parameterized by the two radii $r_v = (100+200\vartheta_1)\,\mu\mathrm{m}$ and $r_a= (400+200\vartheta_2)\,\mu\mathrm{m}$. The initial TF samples were obtained from a field solver for $250$ values of $s=\jj2\pi f$, with $f$ uniformly spaced in $[20\,\mathrm{MHz}, 5\,\mathrm{GHz}]$, and $81$ values of $\mvartheta$. 
The number of basis poles is $\nu = 22$. We also fix $n_{\rm ico}=6$ and $h = 0.5$. Polynomial basis functions $\vet{\varphi}(\mvartheta)$ are considered first, with $\boldsymbol{\rho}=(3,3)$ or $\boldsymbol{\rho}=(4,4)$, i.e. polynomials of degree two or three in each variable.  The model responses are depicted in Fig.~\ref{fig:schuster-responses}. The errors reported in Table \ref{tab:errors} show that the exact stability criterion reduces the RMS error by more than a factor of four. 
\begin{figure}
    \centering
    \includegraphics[width=0.8\linewidth]{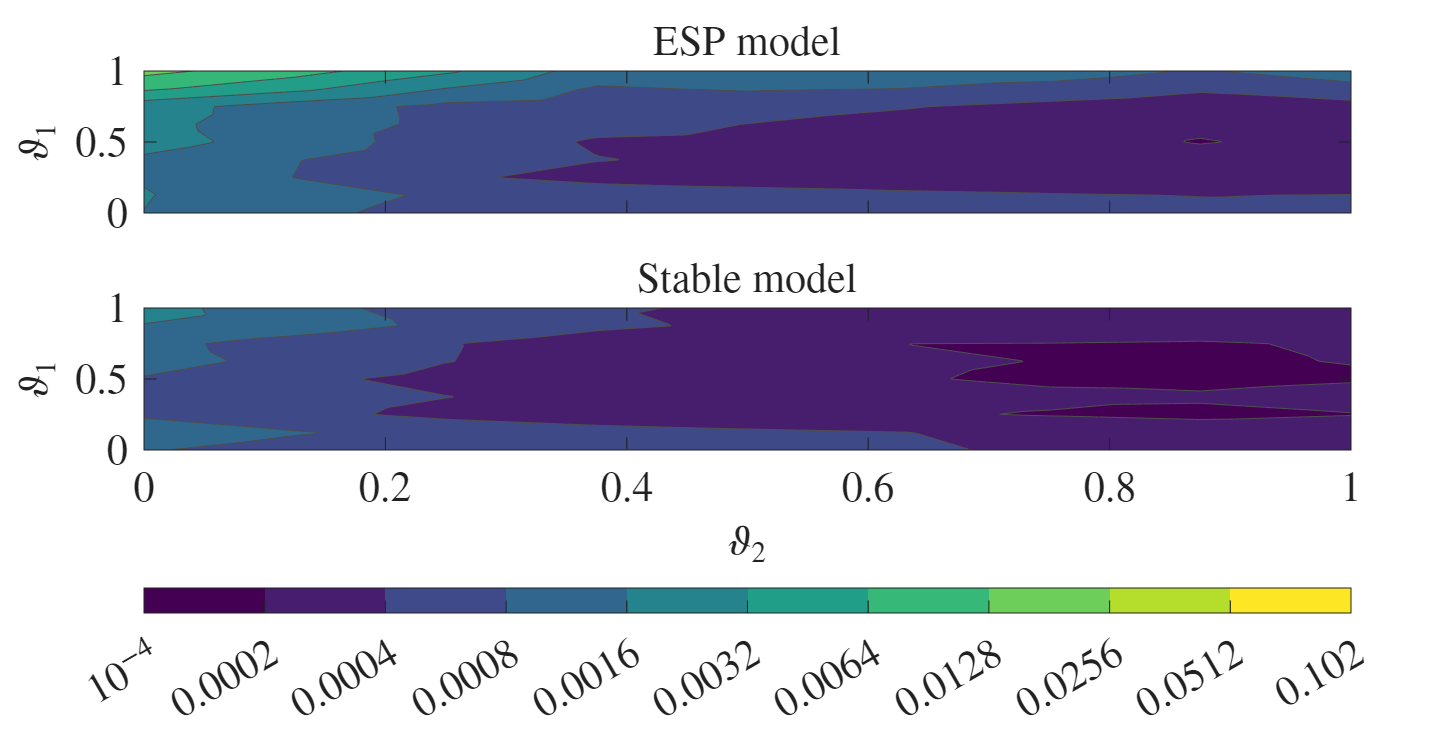}
    \caption{Data link model of Sec. \ref{sec:schuster}. Contour plot of the RMS error (over $f$) for each point $\mvartheta\in[0,1]^2$ in the dataset.}
    \label{fig:schuster-contour}
\end{figure}

For the most complex configuration with $\boldsymbol{\rho} = (4,4)$, polynomial parameterization and $\nu = 22$, the first part of the algorithm takes 5.8 minutes to produce an ESP model, and 45 minutes to complete $6$ ICO iterations (7.5 minutes per iteration).

\subsection{Transmission Line Network}\label{sec:tlnet-example}
In this section, we consider again a network of transmission lines as in Fig. \ref{fig:tline-diagram}, but with no RLC circuit. The structure is extended with two additional segments after point C, including TL stubs analogously to those at the end of the A-B and B-C sections. The final structure is thus made up of four TL segments interleaved with three TL stubs. All geometric dimensions of the microstrip lines are as in Sec. \ref{sec:tl-example}. The lengths $\ell_1$, $\ell_2$, $\ell_3$, $\ell_4$ of the four segments are nominally $1$ cm. The first three are parameterized as $\ell_i = (1+0.1\vartheta_i)\,\mathrm{cm}$, $i = 1,2,3$. The input impedance at point A is thus a transfer function depending on $k=3$ parameters. It is initially evaluated at $10^3$ values of $s=\jj2\pi f$, with $f\in[0.01,12]$ GHz, and $100$ values of $\mvartheta\in[0,1]^3$ according to a space-filling Sobol sequence. The dataset thus obtained is used to construct models $H(s,\mvartheta)$ with $\nu = 14$ and $\boldsymbol{\rho} = (2,2,2)$, with the results reported in Fig. \ref{fig:tlnet-datamodel}. The errors reported in Table \ref{tab:errors} confirm that even in this case the exact stability criterion leads to more accurate models.
The runtime to optimize the ESP models in the first part of the algorithm is 4.8 minutes. The final six ICO iterations take 21.5 minutes to construct the stable model.

 \begin{figure}
 \centering
\subfloat[Polynomial, $\boldsymbol{\rho}= (2,2,2)$]{ 
\includegraphics[width=0.49\textwidth]{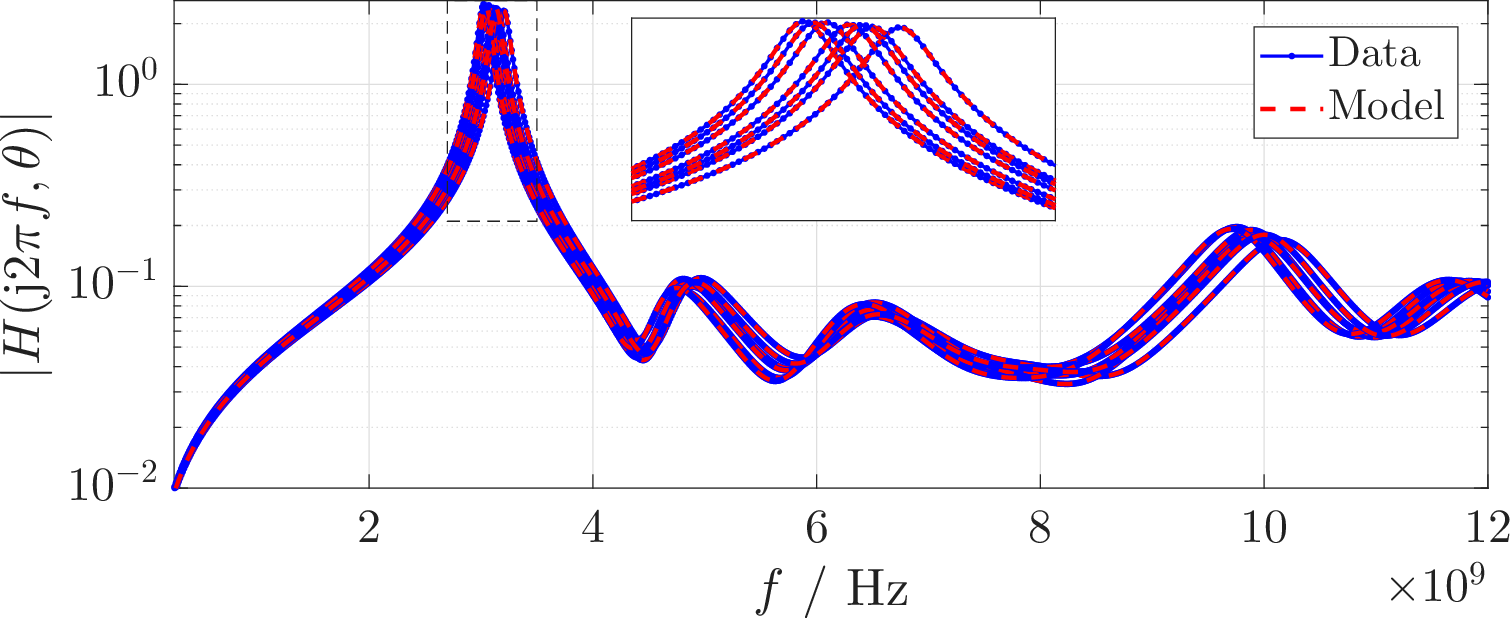}
}
\subfloat[Rational, $\boldsymbol{\rho} = (2,2,2)$]{
    \includegraphics[width=0.49\textwidth]{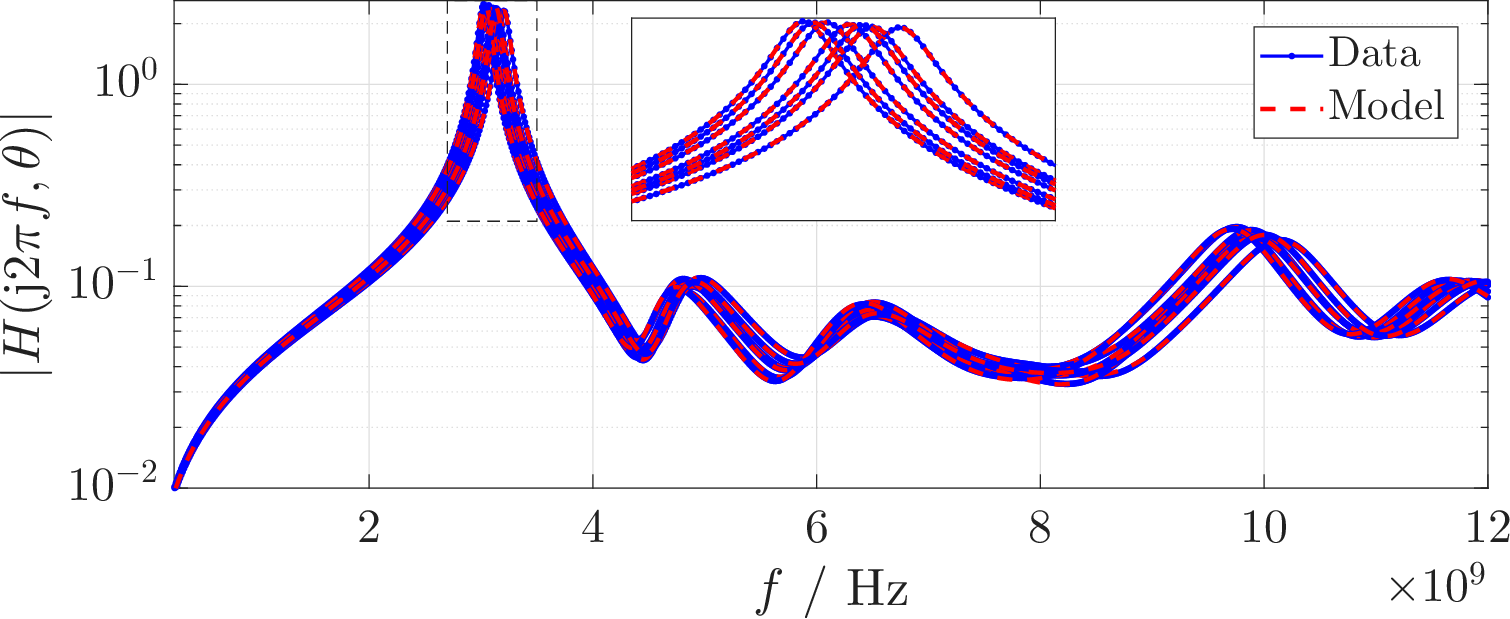}
    }
    \caption{Data-model comparison for the TL network model of Sec. \ref{sec:tlnet-example}. Different curves correspond to different fixed values of $\mvartheta$.}
    \label{fig:tlnet-datamodel}
\end{figure}

\section{Conclusions and Outlook}
This paper presents a stability enforcement scheme for multivariate rational approximation of parametric transfer functions. The PSK iteration has been equipped with stability constraints that apply to models with polynomial or rational parameterization. In addition, the stability constraints are exact and non-conservative. 

The results of this paper significantly extend the class of parameterized models for which stability can be enforced exactly. An interesting direction to explore further is the case where the parameterization is defined in terms of inverse quadratic RBFs. In fact, RBFs are known to be particularly well suited for approximation in high-dimensional spaces.    
In addition, the structure of pAAA models is essentially captured by the class considered here. However, the interpolatory constraints of pAAA have not been included, so that an extension to the pAAA algorithm is left for future work.

\begin{appendices}
\section{Proof of Proposition 
\ref{prop:stab-criterion}}
\label{sec:appendix2}
The following two auxiliary lemmas characterize the duals of $\mathcal{Q}$, $\mathcal{Q}_0$ and $\mathcal{P}$.
\begin{lemma}[Dual cone of $\mathcal{Q}$]\label{lemma:dualq}
Let $\mathcal{Q}$ be defined as in \eqref{eq:q-def}. Its dual cone is
\begin{equation}
    \label{eq:yfact}
    \mathcal{Q}^* =\left\{\begin{pmatrix}
        \mat{F} \\\mat{G} \\ \mat{V}
    \end{pmatrix}\begin{pmatrix}
        \mat{F} \\ \mat{G} \\\mat{V}
    \end{pmatrix}^*: \mat{F},\mat{G} \in\mathbb{C}^{\nu\times m},\;  
    \mat{V} \in\mathbb{C}^{K\times m}
    ,\;\mat{F}\mat{G}^*+\mat{G}\mat{F}^*\succeq 0\right\}
\end{equation}  
where $m = 2\nu+K$ is the size of matrices in $\mathcal{Q}$.
The dual cone $\mathcal{Q}_0^*$ has the same description with $\mat{F}\mat{G}^*+\mat{G}\mat{F}^*=\mat{0}$. 
\end{lemma}
\begin{proof}
    By definition, $\mat{Y}\in\mathcal{Q}^*$ if and only if $
        \dotprod{\mat{Q}}{\mat{Y}}\geq 0\; \forall\, \mat{Q} \in\mathcal{Q}
    $. 
    Since elements $\mat{Q}\in\mathcal{Q}$ are exactly the matrices that can be written as 
    \begin{equation}
        \mat{Q} = \begin{pmatrix}
            \mat{0} & \mat{X} &\mat{0} \\ \mat{X} & \mat{0} &\mat{0}\\\mat{0} & \mat{0} & \mat{0}
        \end{pmatrix} + \mat{S}
    \end{equation}
    for some $\mat{X}\succeq 0$, $\mat{S}\succeq 0$, we can set $\mat{X}=0$ and obtain that $\dotprod{\mat{S}}{\mat{Y}}\geq 0$ for all $\mat{S}\succeq 0$, hence $\mat{Y}\succeq 0$. This is equivalent to existence of a factorization of $\mat{Y} = \left(\mat{F};\mat{G};\mat{V}\right)\left(\mat{F};\mat{G};\mat{V}\right)^*$ in terms of matrices $\mat{F},\mat{G} \in\mathbb{C}^{\nu\times m},  
    \mat{V} \in\mathbb{C}^{K\times m}$. By setting $\mat{S}=0$ we have
    \begin{equation}
    \dotprod{\begin{pmatrix}
        \mat{F} \\ \mat{G}   \\ \mat{V} 
        \end{pmatrix}\begin{pmatrix}
        \mat{F} \\ \mat{G}  \\ \mat{V}  
        \end{pmatrix}^*}{\begin{pmatrix}
            \mat{0} & \mat{X} &\mat{0} \\ \mat{X} & \mat{0} & \mat{0}
            \\
            \mat{0} & \mat{0} & \mat{0}
        \end{pmatrix}}\geq 0,\qquad \forall\; \mat{X}\succeq 0
    \end{equation}
    that
    is
$0\leq\textrm{tr}(\mat{F}^*\mat{X}\mat{G} + \mat{G}^*\mat{X}\mat{F})=\textrm{tr}(\mat{X}(\mat{G}\mat{F}^* + \mat{F}\mat{G}^*))$ for any $\mat{X}\succeq 0$, which is equivalent to $\mat{F}\mat{G}^*+\mat{G}\mat{F}^*\succeq 0$. In the case of $\mathcal{Q}_0$, a similar reasoning leads to $\textrm{tr}(\mat{X}(\mat{G}\mat{F}^* + \mat{F}\mat{G}^*))\geq 0$ for all $\mat{X}\in\Hmat$, that is $\mat{F}\mat{G}^*+\mat{G}\mat{F}^*=0$.
\end{proof}

\begin{lemma} \label{lemma:p-dual}
Let $\mat{M}(\mvartheta)$ be a continuous matrix-valued function of $\mvartheta\in[0,1]^k$ such that $\mat{M}(\mvartheta)$ has full column rank. 
Assume $\mathcal{K}\subset \Hmat$ is a closed convex cone for which $\mathcal{K}^*\subseteq \Hmat_+$. 
Let $\mathcal{P} = \{\mat{P}\in\Hmat: \;\mat{M}(\mvartheta)^*\mat{P}\mat{M}(\mvartheta)\succeq_{\mathcal{K}} 0\;\forall \,\mvartheta\in[0,1]^k\}$. Then $\mathcal{P}^* = \mathrm{conv}\,\mathcal{C}$ is the convex hull of the cone
\begin{equation}
    \mathcal{C} = \{\mat{M}(\mvartheta)\mat{Y}\mat{M}(\mvartheta)^*: \; \mvartheta \in [0,1]^k,\;\; \mat{Y} \in\mathcal{K}^*\}.
\end{equation}

\end{lemma}

\begin{proof}
    Since $\mathcal{K}$ is a closed convex cone, $\mat{M}(\mvartheta)^*\mat{P}\mat{M}(\mvartheta)\succeq_{\mathcal{K}}0$ for any fixed $\mvartheta$ if and only if  
\begin{equation}
\dotprod{\mat{M}(\mvartheta)^*\mat{P}\mat{M}(\mvartheta)}{\mat{Y}} \geq 0\qquad \forall\; \mat{Y} \in\mathcal{K}^*\label{eq:eq44}
\end{equation}
Therefore $\mat{P}\in\mathcal{P}$ if and only if \eqref{eq:eq44} holds for all $\mvartheta \in [0,1]^k$, i.e.
\begin{equation}
\dotprod{\mat{P}}{\mat{M}(\mvartheta)\mat{Y}\mat{M}(\mvartheta)^*} \geq 0 \quad \forall\; \mvartheta \in [0,1]^k, \quad \mat{Y}\in\mathcal{K}^*.
\end{equation}
In turn, this is equivalent to  $\mat{P}\in\mathcal{C}^*$, hence
 $\mathcal{P} = \mathcal{C}^*$. By the bipolar theorem, the dual of $\mathcal{P}$ is thus the closed convex hull of $\mathcal{C}$, i.e. $\mathcal{P}^*=\mathcal{C}^{**}=\textrm{cl}\,\textrm{conv}\, \mathcal{C}$.
 
 Consider the set $\mathcal{C}_1 = \{\mat{M}(\mvartheta)\mat{Y}\mat{M}(\mvartheta)^*: \; \mvartheta \in [0,1]^k,\, \mat{Y} \in\mathcal{K}^*,\, \snorm{\mat{Y}}_F=1\}$. Since $\mathcal{K}^*\subseteq\Hmat_+$, also $\mathcal{C}_1\subseteq \Hmat_+$.
 The assumptions that $\mat{M}(\mvartheta)$ has full column rank for all $\mvartheta\in[0,1]^k$ and $\mat{Y}\succeq 0$ imply that $\mat{0}\not \in \mathcal{C}_1$. Elements of $\mathcal{C}_1$ are thus nonzero positive semidefinite matrices, and a convex combination of such matrices cannot yield zero. Hence, $\mat{0}\not\in\mathrm{conv}\,\mathcal{C}_1$. 
 Since $\mat{M}(\mvartheta)$ is continuous, the set $\mathcal{C}_1$ is compact as the continuous image of the compact set $[0,1]^k\times \{\mat{Y}\in\mathcal{K}^*:\,\snorm{\mat{Y}}_F=1\}$.
 Observe that elements of the cone $\mathcal{C}$ are exactly those obtained by nonnegative rescaling of elements of $\mathcal{C}_1$, from which it is possible to show that $\mathrm{conv}\,\mathcal{C} = \mathrm{conv}\{\alpha\mat{X}:\;\alpha\geq0,\,\mat{X}\in\,\mathcal{C}_1\}$ is the conic hull of $\mathcal{C}_1$. This set is closed by \cite[Prop. 1.4.7]{HiriartUrruty} because it is the conic hull of a nonempty compact set whose convex hull does not contain zero. Hence $\mathrm{conv}\,\mathcal{C}$ is closed and $\mathrm{conv}\,\mathcal{C} = \mathrm{cl}\,\mathrm{conv}\,\mathcal{C} = \mathcal{P}^*$
\end{proof}

The following lemma summarizes instrumental results originally presented in \cite{rantzer-on-the-kalman} and \cite{generalized-s-procedure}.
\begin{lemma}[\cite{rantzer-on-the-kalman},\cite{generalized-s-procedure}]\label{lemma:1vector}
    Let $\mat{F}, \mat{G}\in \mathbb{C}^{n\times m}$ be given and $\mat{U}\in\mathbb{C}^{m\times m}$ indicate a unitary matrix. Let $\vet{f}_1,\dots,\vet{f}_m\in\mathbb{C}^n$ be the columns of $\mat{F}\mat{U}$ and $\vet{g}_1,\dots,\vet{g}_m\in\mathbb{C}^n$ the columns of $\mat{G}\mat{U}$.
    \begin{itemize}
        \item[\emph{a)}] If $\mat{F}\mat{G}^*+\mat{G}\mat{F}^*\succeq \mat{0}$, then $\mat{U}$ can be chosen so that $\vet{f}_1\vet{g}_1^* + \vet{g}_1\vet{f}_1^*\succeq \mat{0}$;
        \item[\emph{b)}] If $\mat{F}\mat{G}^*+\mat{G}\mat{F}^*=\mat{0}$, then $\mat{U}$ can be chosen so that $\vet{f}_\ell \vet{g}_\ell^* + \vet{g}_\ell \vet{f}_\ell^*=\mat{0}$, $\forall\,\ell = 1,\dots,m$.
    \end{itemize}
    \end{lemma}
    \begin{proof}
        The first part is Lemma 5.1 in \cite{generalized-s-procedure}. The second part follows from Lemma 5 in \cite{rantzer-on-the-kalman} (with $\mat{W} = \eye$ in the notation of \cite{rantzer-on-the-kalman}).
    \end{proof}
The following is an auxiliary linear-algebraic lemma.
    \begin{lemma}\label{lemma:sqrt}
        Let $\mat{K}\in\mathbb{C}^{n\times m}$, $\mat{Z}\in\mathbb{C}^{m\times m}$ with $\mat{Z}\succeq 0$ and $\mat{G}\in\mathbb{C}^{n\times l}$. Assume $\mat{K}$ has linearly independent columns. If $\mat{K}\mat{Z}\mat{K}^* = \mat{G}\mat{G}^*$, then there exists $\mat{R}$ such that $\mat{Z} = \mat{R}\mat{R}^*$ and $\mat{K}\mat{R} = \mat{G}$.
    \end{lemma}
    \begin{proof}
    Since $\mat{K}$ has linearly independent columns, $\mat{K}\mat{Z}\mat{K}^*$ and $\mat{Z}$ have the same rank $r = \mathrm{rank}\,\mat{Z}$. Since $\mat{Z}\succeq 0$, there is a factorization $\mat{Z}=\mat{L}\mat{L}^*$ with $\mat{L}\in\mathbb{C}^{m\times r}$. The rank of $\mat{G}\mat{G}^*$ is also $r$, and it cannot be greater than the number of columns of $\mat{G}$, i.e. $r\leq l$. Extend $\mat{L}$ to a matrix $\tilde{\mat{L}} = (\mat{L} \;\; \mat{0})\in\mathbb{C}^{m\times l}$ with $l$ columns. Let $\mat{F} = \mat{K}\tilde{\mat{L}}$. Then $\mat{F}\mat{F}^* =\mat{K}\mat{Z}\mat{K}^* = \mat{G}\mat{G}^*$ and $\mat{F}$, $\mat{G}$ have the same size. By \cite[Theorem 7.3.11]{horn-johnson} or \cite[Lemma 3]{rantzer-on-the-kalman}, there is a unitary $\mat{U}\in\mathbb{C}^{l\times l}$ such that $\mat{F}\mat{U}=\mat{G}$. With $\mat{R} = \tilde{\mat{L}}\mat{U}$, we have $\mat{R}\mat{R}^* = \tilde{\mat{L}}\tilde{\mat{L}}^* = \mat{Z}$ and $\mat{G}=\mat{F}\mat{U} = \mat{K}\tilde{\mat{L}}\mat{U} = \mat{K}\mat{R}$.
    \end{proof}

We finally turn to the proof of Proposition \ref{prop:stab-criterion}.
\begin{proof}
     To show sufficiency, we follow the main steps of \cite{scherer2005}. Assume \eqref{eq:q-cond}-\eqref{eq:lmi-stab-cond} hold. Then condition \eqref{eq:q-cond} is equivalent to $\mat{P}\in\mathcal{P}$. Since  $\mathcal{P}$ is a subset of admissible matrices $\mat{P}$ for \eqref{eq:p-def} as discussed in Sec. \ref{sec:global-cert},
     \begin{equation}
         \begin{pmatrix}
             \mat{\Delta}(s,\mvartheta)\\\eye
         \end{pmatrix}^*\mat{P}
         \begin{pmatrix}
             \mat{\Delta}(s,\mvartheta)\\\eye
         \end{pmatrix}\succeq 0,\qquad \forall\; (s,\mvartheta) \in {\Lambda}.
     \end{equation}
     Let $\vet{\psi}(s,\mvartheta) = [\mat{\Delta}(s,\mvartheta) - \mat{A}]^{-1}\vet{b}$ so that $\beta(s,\mvartheta)=\vet{c}\vet{\psi}(s,\mvartheta)$. Note that, since the left side of \eqref{eq:lmi-stab-cond} is positive definite, there is an $\epsilon>0$ for which the same matrix is $\succ \epsilon \eye$. Multiply \eqref{eq:lmi-stab-cond} to the right and left by $(\vet{\psi}(s,\mvartheta);1)$ and its Hermitian transpose, respectively. Using the identity \cite{scherer2005}
     \begin{equation}
         \begin{pmatrix}
             \mat{A} & \vet{b} \\
             \eye & \vet{0}
         \end{pmatrix}\begin{pmatrix}
             \vet{\psi}(s,\mvartheta) \\ 1
         \end{pmatrix} =
         \begin{pmatrix}
             \mat{\Delta}(s,\mvartheta) \\ \eye
         \end{pmatrix}\vet{\psi}(s,\mvartheta)
     \end{equation}
     condition \eqref{eq:lmi-stab-cond} gives
     \begin{multline}
         \vet{\psi}(s,\mvartheta)^*\vet{c}^*\vet{c}\vet{\psi}(s,\mvartheta) > \\
             \vet{\psi}(s,\mvartheta)^*\begin{pmatrix}
             \mat{\Delta}(s,\mvartheta) \\ \eye
         \end{pmatrix}^*P\begin{pmatrix}
             \mat{\Delta}(s,\mvartheta) \\ \eye
         \end{pmatrix}
             \vet{\psi}(s,\mvartheta)  + \epsilon (\norm{\vet{\psi}(s,\mvartheta)}^2 + 1) 
     \end{multline}
     that is $|\beta(s,\mvartheta)|^2>\epsilon(\norm{\vet{\psi}(s,\mvartheta)}^2 + 1)>\epsilon \; \forall\, (s,\mvartheta)\in {\Lambda}$. This implies $\beta(s,\mvartheta)\neq 0$ for all $(s,\mvartheta)\in\Lambda $ and $\beta(\infty,\mvartheta)=\lim_{s\to \infty}\beta(s,\mvartheta)\neq 0$.

     Necessity is established next by contrapositive: we show that if \eqref{eq:q-cond}-\eqref{eq:lmi-stab-cond} have no solution $\mat{P}$, then $\beta(s,\mvartheta)$ must have a zero in $\bar{\Lambda}$. Recall that $\mathcal{P}$ is the set of matrices satisfying \eqref{eq:q-cond} and assume that \eqref{eq:lmi-stab-cond} is not verified for any $\mat{P}\in\mathcal{P}$. Equivalently, $\Hmat_{++}$ and the set spanned by the left side of \eqref{eq:lmi-stab-cond} as $\mat{P}$ spans $\mathcal{P}$ are disjoint. The latter set is convex as an affine image of the convex set $\mathcal{P}$. It follows from the separating hyperplane theorem that there is a nonzero $\mat{Z}\succeq 0$ such that
\begin{align}
    \dotprod{
    \begin{pmatrix}
         \mat{A} & \vet{b} 
         \\
         \eye &\vet{0}
     \end{pmatrix}\mat{Z}\begin{pmatrix}
         \mat{A} & \vet{b} 
         \\
         \eye &\vet{0}
     \end{pmatrix}^*}{
     \mat{P}
    }&\geq 0\qquad \forall\; \mat{P} \in\mathcal{P},\label{eq:dual-1}
    \\
    \dotprod{\begin{pmatrix}
        \vet{c}^*\vet{c} & \mat{0} \\ \mat{0} & 0
    \end{pmatrix}}{\mat{Z}} &\leq 0.\label{eq:dual-2}
\end{align}
 Condition \eqref{eq:dual-1} gives
\begin{equation}
     \begin{pmatrix}
            \mat{A} & \vet{b} 
         \\
         \eye &\vet{0}
     \end{pmatrix}\mat{Z}\begin{pmatrix}
         \mat{A} & \vet{b} 
         \\
         \eye &\vet{0}
     \end{pmatrix}^*\in\mathcal{P}^* \label{eq:z-in-c}
\end{equation}
By Lemma \ref{lemma:p-dual}, $\mathcal{P}^*=\textrm{conv}\,\{\mat{M}(\mvartheta)\mat{Y}\mat{M}(\mvartheta)^*: \; \mvartheta \in [0,1]^k,\;\; \mat{Y} \in\mathcal{Q}^*\}$. Therefore the matrix in \eqref{eq:z-in-c} can be expressed as convex combination 
\begin{equation}
    \begin{pmatrix}
         \mat{A} & \vet{b} 
         \\
         \eye &\vet{0}
     \end{pmatrix}\mat{Z}\begin{pmatrix}
         \mat{A} & \vet{b} 
         \\
         \eye &\vet{0}
     \end{pmatrix}^* = \sum_{\ell=1}^{\bar{\ell}}\mat{M}(\mvartheta_\ell)\mat{Y}_\ell \mat{M}(\mvartheta_\ell)^*\label{eq:quad-eq}
\end{equation}
in terms of $\mvartheta_\ell\in[0,1]^k$ and $\mat{Y}_\ell\in\mathcal{Q}^*$, $\ell = 1,\dots,\bar{\ell}$.
By Lemma \ref{lemma:dualq}, every $\mat{Y}_{\ell}\in\mathcal{Q}^*$ has a decomposition $\mat{Y}_\ell = \tilde{\mat{X}}_\ell \tilde{\mat{X}}_\ell^{*}$ with $\tilde{\mat{X}}_{\ell} = (\tilde{\mat{F}}_{\ell}; \tilde{\mat{G}}_{\ell}; \tilde{\mat{V}}_{\ell})$ and $\tilde{\mat{F}}_{\ell}\tilde{\mat{G}}_{\ell}^{*} + \tilde{\mat{G}}_{\ell}\tilde{\mat{F}}_{\ell}^{*} \succeq 0$. By Lemma \ref{lemma:1vector}(a) there is a unitary $\mat{U}_\ell$ for which the matrix $\tilde{\mat{X}}_{\ell}\mat{U}_\ell = \mat{X}_{\ell}$ is such that its first column $\vet{\xi}_{\ell}$ has the form $\vet{\xi}_\ell = (\vet{f}_\ell;\vet{g}_\ell;\vet{v}_\ell)$ with  $\vet{f}_\ell,\vet{g}_{\ell}\in\mathbb{C}^\nu$, $\vet{v}_{\ell}\in\mathbb{C}^{K}$ satisfying $\vet{f}_\ell \vet{g}_\ell^* + \vet{g}_\ell \vet{f}_\ell^*\succeq \mat{0}$. By Corollary 4 in \cite{rantzer-on-the-kalman}, this implies that either 
\begin{itemize}
    \item[\emph{a)}] $
\exists\;s_{\ell}\in\bar{\mathbb{C}}_+$, $\vet{f}_\ell=s_{\ell}\vet{g}_\ell,\;\;\vet{g}_\ell \neq \vet{0}
$, or
\item[\emph{b)}] $\vet{g}_\ell=\vet{0}$.
\end{itemize}
The decomposition $\mat{Y}_{\ell} = \mat{X}_{\ell}\mat{X}_{\ell}^*$ is unaffected by $\mat{U}_\ell$. Since $\vet{b}\neq \vet{0}$ in \eqref{eq:quad-eq}, the linear independence hypothesis of Lemma \ref{lemma:sqrt} is satisfied and it can be applied to \eqref{eq:quad-eq} to deduce that there is a matrix $\mat{R}$ such that $\mat{Z} = \mat{R}\mat{R}^*$ and 
\begin{equation}
    \begin{pmatrix}
         \mat{A} & \vet{b} 
         \\
         \eye &\vet{0}
     \end{pmatrix}\mat{R} = \begin{pmatrix}
         \mat{M}(\mvartheta_1)\mat{X}_1 &\cdots & \mat{M}(\mvartheta_{\bar{\ell}})\mat{X}_{\bar{\ell}}
     \end{pmatrix}
     \label{eq:RXrel}
\end{equation}
Let $\vet{r} = (\vet{r}_1;r_2)$, $r_2\in\mathbb{C}$, be the first column of $\mat{R}$ and denote the first column $\vet{\xi}_1$ of $\mat{X}_1$ as $\vet{\xi} = (\vet{f};\vet{g};\vet{v})$ with $\vet{f}$, $\vet{g}\in\mathbb{C}^{\nu}$. From \eqref{eq:RXrel}, 
\begin{equation}
    \begin{pmatrix}
         \mat{A} & \vet{b} 
         \\
         \eye &\vet{0}
     \end{pmatrix}\begin{pmatrix}
         \vet{r}_1\\r_2
     \end{pmatrix} = \mat{M}(\mvartheta_1)\vet{\xi}=
     \begin{pmatrix}
         \vet{f}\\
         \mat{\Theta}(\mvartheta_1)\vet{v}\\
         \vet{g}\\
         \vet{v}
     \end{pmatrix}
     \label{eq:r-xi-eq}
\end{equation}
with $\vet{\xi}$ satisfying either \emph{a)} or \emph{b)}. Replacing $\mat{Z} = \mat{R}\mat{R}^*$ in \eqref{eq:dual-2} gives $\norm{\begin{pmatrix}
    \vet{c} & 0
\end{pmatrix}\mat{R}}^2\leq 0$, hence $\begin{pmatrix}
        \vet{c} & 0
    \end{pmatrix}\vet{r}=0$. Since the second block-row of \eqref{eq:r-xi-eq} gives $\vet{r}_1=(\vet{g};\vet{v})$, 
\begin{equation}
    \vet{c}\begin{pmatrix}
        \vet{g} \\ \vet{v}
    \end{pmatrix} = 0\label{eq:outzero}
\end{equation}
 Consider case \emph{a)} first, with $\vet{f} = s_{1}\vet{g}$ and $\vet{g}\neq 0$. Equality \eqref{eq:r-xi-eq} is
\begin{equation}
    \begin{pmatrix}
        \mat{A} & \vet{b} \\ \eye & \vet{0}
    \end{pmatrix}\begin{pmatrix}
        \vet{r}_1\\r_2
    \end{pmatrix} = \mat{M}(\mvartheta_1)\begin{pmatrix}
        s_1\vet{g} \\ \vet{g}\\\vet{v}
    \end{pmatrix} = \begin{pmatrix}
     s_1\vet{g}
     \\
     \mat{\Theta}(\mvartheta_1)\vet{v}
     \\
     \vet{g} \\
     \vet{v}
 \end{pmatrix}.
\end{equation}
The second block-row gives gives $\vet{r}_1=(\vet{g};\vet{v})$ and the first block-row gives
\begin{equation}
    \mat{A}\vet{r}_1 + \vet{b}r_2 = \begin{pmatrix}
        s_1\eye_\nu & \mat{0} \\ \mat{0} & \mat{\Theta}(\mvartheta_1)
    \end{pmatrix}\begin{pmatrix}
        \vet{g} \\ \vet{v}
    \end{pmatrix} \implies 
    \begin{cases}
    \mat{\Sigma}\vet{g}+ \vet{b}_\sigma r_2= s_1\vet{g}
    \\
    \mat{A}_{21}\vet{g} + \mat{\Pi} \vet{v} +\vet{b}_{\theta}r_2= \mat{\Theta}(\mvartheta_1)\vet{v}
    \end{cases}
\end{equation}
that is $\vet{r}_1=(\vet{g};\,\vet{v})=[\mat{\Delta}(s_1,\mvartheta_1) - \mat{A}]^{-1}\vet{b}r_2$ and $r_2\neq 0$ because it is assumed that $\mat{\Delta}(s_1,\mvartheta_1)-\mat{A}$ is nonsingular.  By \eqref{eq:outzero}, $0=\vet{c}\vet{r}_1=\beta(s_1,\mvartheta_1)r_2$ and $\beta(s,\mvartheta)$ is zero at $(s_1,\mvartheta_1)\in\bar{\Lambda}$.

Let us now turn to case \emph{b)}, that is $\vet{g}=0$. In this case $\vet{r}_1 = (\vet{0};\vet{v})$,  the first block-row of \eqref{eq:r-xi-eq} yields
\begin{equation}
    \begin{cases}
        \vet{b}_\sigma r_2 = \vet{f} \\
        \mat{\Pi} \vet{v} + \vet{b}_{\theta} r_2 = \mat{\Theta}(\mvartheta_1)\vet{v}
    \end{cases}
\end{equation}
so that $r_2\neq 0$,  $\vet{v} = [\mat{\Theta}(\mvartheta_1)-\mat{\Pi}]^{-1}\vet{b}_{\theta} r_2$ and $0=\vet{c}\vet{r}_1=\vet{\eta} [\mat{\Theta}(\mvartheta_1)-\mat{\Pi}]^{-1}\vet{b}_{\theta} r_2$ implies
\begin{equation}
    0=\vet{\eta}[\mat{\Theta}(\mvartheta_1)-\mat{\Pi}]^{-1}\vet{b}_{\theta} =\lim_{s\to\infty} \beta(s,\mvartheta_1)
\end{equation}
that is $\beta(s,\mvartheta)$ has a zero at $(\infty,\mvartheta_1)\in\bar{\Lambda}$.
\end{proof}

\section{Proof of Proposition \ref{prop:polya}}
\label{sec:appendix3}
This proof follows the strategy of \cite{scherer2005}, with the differences that it applies to the Bernstein relaxation instead of P\'olya's, and it is formulated in terms of generalized inequalities.
\begin{proof} 
Consider a fixed arbitrary element $\mat{Y}\in\mathcal{Q}^*$ with $\norm{\mat{Y}}_F = 1$. The hypothesis $\mat{W}(\mvartheta)\succ_{\mathcal{Q}}0$ implies $\dotprod{\mat{W}(\mvartheta)}{\mat{Y}} > 0$. Hence, the scalar function $g(\mvartheta,\mat{Y})=\dotprod{\mat{W}(\mvartheta)}{\mat{Y}}$ is positive for all $\mvartheta\in[0,1]^k$ and $\mat{Y}\in\mathcal{Q}^*$ with $||\mat{Y}||_F=1$. Note that $g(\mvartheta,\mat{Y})$ is a real polynomial in $\mvartheta$ for any fixed $\mat{Y}$. 
Let $\delta_0$ be the maximum degree of $\mat{W}(\mvartheta)$ in each variable, and $\delta_{\rm tot}$ the total degree of $\mat{W}(\mvartheta)$.
The polynomial $g(\mvartheta,\mat{Y})$ can be expressed in the Bernstein basis of a given degree $\bar{\delta}\geq \delta_0$,
    \begin{equation}
g(\mvartheta,\mat{Y})= \dotprod{\mat{W}(\mvartheta)}{\mat{Y}} = \sum_{\boldsymbol{n}\leq \bar{\delta}}
\dotprod{\mat{W}^{(\bar{\delta})}_{{\rm ber},\boldsymbol{n}}}{\mat{Y}} \bernbas_{\boldsymbol{n}}^{(\bar{\delta})}(\mvartheta),
    \end{equation}
    with Bernstein coefficients $\dotprod{\mat{W}^{(\bar{\delta})}_{{\rm ber},\boldsymbol{n}}}{\mat{Y}}$.
    Let $g_{\boldsymbol{n}}(\mat{Y})$ be the coefficient of $g(\mvartheta,\mat{Y})$ in the monomial basis corresponding to the monomial $\mvartheta^{\boldsymbol{n}}$. 
In order to use the bounds presented in \cite{deklerk-bounds}, we introduce the quantities
\begin{equation}
    L(\mat{Y}) = \max_{\boldsymbol{n}\leq\delta_0}\frac{\boldsymbol{n}!}{|\boldsymbol{n}|!}|g_{\boldsymbol{n}}(\mat{Y})|\qquad 
    g_{\rm min}(\mat{Y}) = \min_{\mvartheta\in [0,1]^k}g(\mvartheta,\mat{Y})
\end{equation}
where $|\vet{n}|!=(n_1+\cdots+n_k)!$ and $\vet{n}!=n_1!\cdots n_k!$. 
For a given $\bar{\delta}\geq\delta_{\rm tot}$, Theorem 3.4(iii) in \cite{deklerk-bounds} gives a lower bound on the smallest Bernstein coefficient of the polynomial $g(\mvartheta,\mat{Y})$, that is 
\begin{equation}
    \dotprod{\mat{W}^{(\bar{\delta})}_{{\rm ber},\boldsymbol{n}}}{\mat{Y}}\geq g_{\rm min}(\mat{Y}) - C L(\mat{Y})/\bar{\delta}
\end{equation}
for all $\vet{n}\leq \bar{\delta}$, with $C=k^{\delta_{\rm tot}}\binom{\delta_{\rm tot}+1}{3}$ being a constant. Define the quantities
\begin{equation}\label{eq:bern-proof-defs}
    L_{\boldsymbol{n}} = \frac{\boldsymbol{n}!}{|\boldsymbol{n}|!}\max_{\substack{\mat{Y}\in\mathcal{Q}^*,\\\snorm{\mat{Y}}_F=1}} |g_{\boldsymbol{n}}(\mat{Y})|, 
\qquad g_{\rm min} = \min_{\substack{\mvartheta\in [0,1]^k,\\\mat{Y}\in\mathcal{Q}^*,\snorm{\mat{Y}}_F=1 } }g(\mvartheta,\mat{Y}),
\end{equation}
whose existence is guaranteed because they are extrema of continuous functions over compact sets. Note that $g_{\rm min}>0$ because $g(\mvartheta,\mat{Y})$ is positive in the considered domain. Also define the maximum of the $L_{\boldsymbol{n}}$'s as $L = \max_{\boldsymbol{n}\leq{\delta_0}}L_{\boldsymbol{n}}$. The definitions in \eqref{eq:bern-proof-defs} imply $L\geq L(\mat{Y})$ and $g_{\rm min}\leq g_{\rm min}(\mat{Y})$ for all $\mat{Y}\in\mathcal{Q}^*$ with $\snorm{\mat{Y}}_F=1$.
Choose an integer $\bar{\delta} > C L/g_{\rm min}$ and $\bar{\delta}\geq\delta_{\rm tot}$. Then the coefficients in the degree-$\bar{\delta}$ Bernstein basis satisfy
\begin{equation}
\dotprod{\mat{W}^{(\bar{{\delta}})}_{{\rm ber},\boldsymbol{n}}}{\mat{Y}} \geq g_{\rm min}(\mat{Y}) - C\frac{L(\mat{Y})}{\bar{\delta}}\geq g_{\rm min} - C\frac{L}{\bar{\delta}} \geq 0 
\end{equation}
for all $\mat{Y}\in\mathcal{Q}^*$ with $\snorm{\mat{Y}}_F=1$.
 Since any element of $\mathcal{Q}^*$ can be obtained by nonnegative rescaling of these particular $\mat{Y}$'s, the inequality $\dotprod{\mat{W}_{{\rm ber},\boldsymbol{n}}^{(\bar{\delta})}}{\mat{Y}}\geq 0$ holds for all $\mat{Y}\in\mathcal{Q}^*$, which implies $\mat{W}_{{\rm ber},\boldsymbol{n}}^{(\bar{\delta})}\in\mathcal{Q}^{**}=\mathcal{Q}$ because $\mathcal{Q}$ is a closed convex cone.
\end{proof}

\section{Proof of Proposition \ref{prop:pr-criterion}}
\label{sec:pr-proof}
\begin{proof}
Let us first prove sufficiency, that is fulfillment of \eqref{eq:q-cond-pr}-\eqref{eq:lmi-pr-cond} implies that $\beta(\cdot,\mvartheta)$ is ESP.  
Analyticity of $\beta(\cdot,\mvartheta)$ in a right-half plane containing $\bar{\mathbb{C}}_+$ follows from the assumption that $\mat{\Sigma}$ is Hurwitz. Assume \eqref{eq:q-cond-pr} and \eqref{eq:lmi-pr-cond} hold for some $\mat{P}$. Fix an arbitrary $\mvartheta\in[0,1]^k$. Condition \eqref{eq:q-cond-pr} implies that there is an $\mat{X}\in\Hmat^{\nu}$ such that
\begin{equation}
    \mat{M}_0(\mvartheta)^*\mat{P}\mat{M}_0(\mvartheta)\succeq \begin{pmatrix}
        \mat{0} & \mat{X} & \mat{0}
        \\
        \star & \mat{0} & \mat{0}
        \\
        \star & \star & \mat{0}
    \end{pmatrix}
\end{equation}
Introduce $\mat{\Delta}_0(s,\mvartheta) = \mathrm{blkdiag}\{s\eye_{\nu},\mat{\Theta}(\mvartheta)^*\}$, and restrict $s=\jj\omega$, $\omega\in\mathbb{R}$,
\begin{multline}
    \begin{pmatrix}
        \mat{\Delta}_0(s,\mvartheta)
        \\
        \eye
    \end{pmatrix}^*
   \mat{P}\begin{pmatrix}
        \mat{\Delta}_0(s,\mvartheta)
        \\
        \eye
    \end{pmatrix}
    =\\
    \begin{pmatrix}
        s\eye & \mat{0} 
        \\
        \eye & \mat{0}
        \\
        \mat{0} & \eye
    \end{pmatrix}^* \mat{M}_0(\mvartheta)^*\mat{P}\mat{M}_0(\theta)
    \begin{pmatrix}
        s\eye & \mat{0} 
        \\
        \eye & \mat{0}
        \\
        \mat{0} & \eye
    \end{pmatrix}
    \succeq \begin{pmatrix}
        (s+s^*)\mat{X} & \mat{0} & \mat{0} 
        \\
        \star & \mat{0} & \mat{0} 
        \\
        \star & \star & \mat{0}
    \end{pmatrix}
    = \mat{0}.\label{eq:deltajw0}
\end{multline}
As this holds for arbitrary $\mvartheta$, \eqref{eq:deltajw0} is verified for all $(s,\mvartheta)\in \jj\mathbb{R}\times[0,1]^k$. Introduce $\mat{A}_{\rm ext} = \mathrm{blkdiag}\{\mat{\Sigma},\mat{\Pi}^*\}$, $\vet{b}_{\rm ext} = (\vet{b}_{\sigma}; \vet{\eta}^*)$ and the vector $
    \vet{\psi}(s,\mvartheta) = 
    [\mat{\Delta}(s,\mvartheta)-\mat{A}_{\rm ext}]^{-1}\vet{b}_{\rm ext}=
    \left(
    [s\eye-\mat{\Sigma}]^{-1}\vet{b}_{\sigma}; [\mat{\Theta}(\mvartheta)^*-\mat{\Pi}^*]^{-1}\vet{\eta}^*
    \right)$.
    Multiply condition \eqref{eq:lmi-pr-cond} to the right and left by $(\vet{\psi}(s,\mvartheta);1)$ and its Hermitian transpose, respectively.  Inequality \eqref{eq:lmi-pr-cond} is strict so, analogously to the proof of Prop. \ref{prop:stab-criterion}, there is $\epsilon>0$ such that
    \begin{multline}
    \begin{pmatrix}
        \vet{\psi}(s,\mvartheta)\\1
    \end{pmatrix}^*
    \begin{pmatrix}
        \mat{0} & \mat{A}_{21}^* & \mat{0} \\ \star & \mat{0} &\vet{b}_{\theta}
        \\
       \star&\star&0
    \end{pmatrix}\begin{pmatrix}
        \vet{\psi}(s,\mvartheta)\\1
    \end{pmatrix}> \\
    \vet{\psi}(s,\mvartheta)^*\begin{pmatrix}
        \mat{\Delta}_0(s,\mvartheta)
        \\
        \eye
    \end{pmatrix}^*\mat{P}\begin{pmatrix}
        \mat{\Delta}_0(s,\mvartheta)
        \\
        \eye
    \end{pmatrix}\vet{\psi}(s,\mvartheta) + \epsilon\norm{\vet{\psi}(s,\mvartheta)}^2  +\epsilon
    \end{multline}
Since the left hand side is exactly $2\re{\beta(s,\mvartheta)}$, we have $\re{\beta(s,\mvartheta)}>\epsilon/2$ for all $(s,\mvartheta)\in \jj\mathbb{R}\times [0,1]^k$, implying that $\beta(\cdot,\mvartheta)$ is ESP for all $\mvartheta\in[0,1]^k$.

Necessity is also shown analogously to the proof in Appendix \ref{sec:appendix2}. Introduce the solution set of \eqref{eq:q-cond-pr}, $\mathcal{P}_0=\{\mat{P}\in\Hmat:\,\mat{M}_0(\mvartheta)^*\mat{P}\mat{M}_0(\mvartheta)\succeq_{\mathcal{Q}_0}0\;\forall\,\mvartheta\in[0,1]^k\}$. If \eqref{eq:q-cond-pr}-\eqref{eq:lmi-pr-cond} are infeasible, there is a $\mat{Z}\succeq 0$ such that
\begin{equation}\label{eq:71}
    \dotprod{\begin{pmatrix}
        \mat{0} & \mat{A}_{21}^* & \mat{0}\\ \star & \mat{0} & \vet{b}_{\theta}
        \\
        \star & \star & 0
    \end{pmatrix}}{\mat{Z}} \leq 0
\end{equation}
and
\begin{equation}
    \begin{pmatrix}
         \mat{A}_{\rm ext} & \vet{b}_{\rm ext} 
         \\
         \eye &\vet{0}
     \end{pmatrix}\mat{Z}\begin{pmatrix}
         \mat{A}_{\rm ext} & \vet{b}_{\rm ext} 
         \\
         \eye &0
     \end{pmatrix}^* = \sum_{\ell=1}^{\bar{\ell}}\mat{M}_0(\mvartheta_\ell)\mat{Y}_\ell \mat{M}_0(\mvartheta_\ell)^*\in\mathcal{P}_0^*\label{eq:72}
\end{equation}
with $\mvartheta_\ell\in[0,1]^k$, $\mat{Y}_\ell \in \mathcal{Q}_0^*$, and the dual $\mathcal{P}_0^*$ is characterized by Lemma \ref{lemma:p-dual}. Then Lemma \ref{lemma:dualq} implies that each $\mat{Y}_{\ell}$ has a decomposition $\mat{Y}_{\ell} = \sum_{m=1}^{m_{\ell}}\vet{\xi}_{\ell,m}\vet{\xi}_{\ell,m}^*$. Partition $\vet{\xi}_{\ell,m} = (\vet{f}_{\ell,m}; \vet{g}_{\ell,m}; \vet{v}_{\ell,m})$ with $\vet{f}_{\ell,m},\vet{g}_{\ell,m}\in\mathbb{C}^{\nu}$. By Lemma \ref{lemma:1vector}(b) the decomposition can be chosen so that $\vet{f}_{\ell,m}\vet{g}_{\ell,m}^* + \vet{g}_{\ell,m}\vet{f}_{\ell,m}^*=0$ holds for all $\ell$, $m$. Hence, by \cite[Corollary 4]{rantzer-on-the-kalman}, for each $\ell$, $m$ there is either a $\bar{\omega}_{\ell,m}\in\mathbb{R}$ such that $\vet{f}_{\ell,m} = \jj\bar{\omega}_{\ell,m}\vet{g}_{\ell,m}$ or $\vet{g}_{\ell,m}=0$. Lemma \ref{lemma:sqrt} applied to equation \eqref{eq:72} implies that there is a square-root $\mat{R}$ of $\mat{Z} = \mat{R}\mat{R}^*$ such that
\begin{equation}\label{eq:73}
     \begin{pmatrix}
         \mat{A}_{\rm ext} & \vet{b}_{\rm ext} 
         \\
         \eye &\vet{0}
     \end{pmatrix}\mat{R} = \begin{pmatrix}
         \mat{M}_0(\mvartheta_1)\vet{\xi}_{1,1} & \mat{M}_0(\mvartheta_1)\vet{\xi}_{1,2} &\cdots & \mat{M}_0(\mvartheta_{\bar{\ell}})\vet{\xi}_{\bar{\ell},{m_{\bar{\ell}}}}
     \end{pmatrix}
\end{equation}
Inserting $\mat{Z} = \mat{R}\mat{R}^*$ in the inequality \eqref{eq:71} shows that there is at least one column of $\mat{R}$, called $\vet{r}$, for which
\begin{equation}
    \vet{r}^*\begin{pmatrix}
        \mat{0} & \mat{A}_{21}^* & \mat{0}\\ \star & \mat{0} & \vet{b}_{\theta}
        \\
        \star & \star & \vet{0}
    \end{pmatrix}\vet{r} \leq 0
\end{equation}
Along the same lines as the proof in Appendix \ref{sec:appendix2}, combining this with equality \eqref{eq:73} allows to show that, for some $\bar{\mvartheta}\in[0,1]^k$, $\re{\beta(s,\bar{\mvartheta})} \leq 0$ for some $s\in\jj\mathbb{R}$ or at infinity, hence $\beta(\cdot,\bar{\mvartheta})$ is not ESP.
\end{proof}

\section{Realization of parameter-dependent basis functions}\label{app:parameterization}
When $k > 1$, the basis functions in $\boldsymbol{\varphi}(\mvartheta)$ can be generated as a tensor product of elementary univariate basis functions, that is
$
    \boldsymbol{\varphi}(\mvartheta) = \bigotimes_{k'=1}^{k}\boldsymbol{\psi}_{k'}(\vartheta_{k'})
$
where $\boldsymbol{\psi}_{k'}(\vartheta_{k'})$ are vectors of $\rho_{k'}$ univariate basis functions along $k'$ and $\otimes$ is the Kronecker product.  
\subsection{Realization of elementary univariate basis functions}
The monomial basis is obtained as $\vet{\psi}_{k'}(\vartheta_{k'}) = \vet{\eta}_{k'}[\mat{\Theta}_{k'}(\vartheta_{k'})-\mat{\Pi}_{k'}]^{-1}$ by setting
    \begin{equation}
        \mat{\Theta}_{k'}(\vartheta_{k'}) = \begin{pmatrix}
0&\vartheta_{k'}\eye_{\rho_{k'}-1} \\ 0 & 0
        \end{pmatrix},\quad \mat{\Pi}_{k'} = \eye_{\rho_{k'}},\quad
        \vet{\eta}_{k'} = \begin{pmatrix}
            -1 & 0& \cdots & 0
        \end{pmatrix}.
    \end{equation}
Any other polynomial basis is clearly obtained as $\boldsymbol{\psi}_{k'}\mat{T}$ through an invertible matrix $\mat{T}$ encoding the change of basis. Applying this $\mat{T}$ as $\mat{T}^{-1}\Theta_{k'}(\vartheta_{k'})$, $\mat{T}^{-1}\Pi_{k'}$ allows to obtain a realization for any desired polynomial basis.

 The model structure of pAAA as presented in \cite{paaa} is based on a set of real poles $\{p_{k',\ell}\}_{\ell=1}^{\rho_{k'}}$ along each dimension $k'$. With our notation, the corresponding univariate basis functions are partial fractions and $\vet{\psi}_{k'}(\vartheta_{k'}) = \vet{\eta}_{k'}[\mat{\Theta}_{k'}(\vartheta_{k'})-\mat{\Pi}_{k'}]^{-1}$ with $\mat{\Theta}_{k'}(\vartheta_{k'}) = \vartheta_{k'}\eye$,  $\mat{\Pi}_{k'}=\mathrm{diag}\,\{p_{k',\ell}\}_{\ell=1}^{\rho_{k'}}$, $\vet{\eta}=\vet{1}^T$.

 Finally, an extension to complex poles $p_{k',\ell}$ is as follows. 
 For even $\rho_{k'}$, a rational partial fraction basis with poles at $\{p_{k',1}, \dots, p_{k',\lfloor \rho_{k'}/2\rfloor}\}\subset \mathbb{C}\setminus\mathbb{R}$ as explicitly defined in \eqref{eq:rat-basis} is obtained by setting
\begin{equation}
    \mat{\Theta}_{k'}(\vartheta_{k'}) = \vartheta_{k'}\eye_{\rho_{k'}},\quad
    \mat{\Pi}_{k'} = \mathrm{diag}\left\{\begin{pmatrix}
        \re p_{k',\ell} & -\operatorname{Im} p_{k',\ell}
        \\
        \operatorname{Im} p_{k',\ell} & \re p_{k',\ell}
    \end{pmatrix}\right\}_{\ell=1}^{\lfloor\rho_{k'}/2\rfloor}
\end{equation}
and $\vet{\eta}_{k'} = \begin{pmatrix}
     1& 0 & 1 & 0 & \cdots
\end{pmatrix}$.

\subsection{Realization of multivariate basis as tensor product}
Starting from a choice of univariate functions, a realization of the multivariate basis obtained by tensor product is $\boldsymbol{\varphi}(\mvartheta) = \vet{\eta}[\mat{\Theta}(\mvartheta)- \mat{\Pi}]^{-1}\mat{B}_{\theta}$ with
\begin{equation}
    \mat{\Pi}=\begin{pmatrix}
    \mat{\Pi}_{1}\otimes \eye_{\rho_2\cdots \rho_k} & &
    \\
   \vet{\eta}_1 \otimes \eye_{\rho_2\cdots \rho_k} & \mat{\Pi}_{2}\otimes \eye_{\rho_3\cdots \rho_k}&
    \\
     &\ddots & 
        \\
       & \vet{\eta}_{k-1}\otimes \eye_{\rho_k}&\mat{\Pi}_k
    \end{pmatrix},
\end{equation}
$\mat{\Theta}(\mvartheta)=\mathrm{diag}\{\mat{\Theta}_{1}(\vartheta_{1})\otimes \eye_{\rho_2\cdots \rho_k}, \mat{\Theta}_{2}(\vartheta_{2})\otimes \eye_{\rho_3\cdots \rho_k},\cdots,\mat{\Theta}_k(\vartheta_k)\}$,
$\eta =
    \begin{pmatrix}
        0 & 0 &\cdots &\vet{\eta}_k^T
    \end{pmatrix}^T$
and $\mat{B}_{\theta}$ selects the first $\rho$ entries of $\vet{\eta}[\mat{\Theta}(\mvartheta)-\mat{\Pi}]^{-1}$, that is $\mat{B}_{\theta} = \left(\eye_{\rho}\;\; 0 \;\;\cdots\;\;0\right)^T$.
\end{appendices}

\bibliographystyle{siam}
\bibliography{sn-bibliography}

\end{document}